\documentclass[11pt,notitlepage,tightenlines,nofootinbib,superscriptaddress]{revtex4-2}

\usepackage{newpxtext,newpxmath}

\usepackage[latin1]{inputenc}
\usepackage{amsthm}
\usepackage{amssymb}
\usepackage{amsmath}
\usepackage{bbold}
\usepackage{bbm}
\usepackage[pdftex, backref=page]{hyperref}
\usepackage{braket}
\usepackage{dsfont}
\usepackage{mathdots}
\usepackage{mathtools}
\usepackage{enumerate}
\usepackage[shortlabels]{enumitem}
\usepackage{csquotes}
\usepackage{stmaryrd}
\usepackage[cal=boondox]{mathalfa}
\usepackage{graphicx}
\usepackage{stackengine}
\usepackage{scalerel}
\usepackage{tensor}       
\usepackage{array}
\usepackage{makecell}
\newcolumntype{x}[1]{>{\centering\arraybackslash}p{#1}}
\usepackage{tikz}
\usepackage{pgfplots}
\usetikzlibrary{shapes.geometric, shapes.misc, positioning, arrows, arrows.meta, decorations.pathreplacing, decorations.pathmorphing, patterns, angles, quotes, calc}
\usepackage{booktabs}
\usepackage{xfrac}
\usepackage{siunitx}
\usepackage{centernot}
\usepackage{comment}
\usepackage{chngcntr}
\usepackage[justification=justified,format=plain]{caption}
\usepackage{subcaption}
\usepackage{ragged2e}

\newtheorem{thm}{Theorem}
\newtheorem*{thm*}{Theorem}
\newtheorem{prop}[thm]{Proposition}
\newtheorem*{prop*}{Proposition}
\newtheorem{lemma}[thm]{Lemma}
\newtheorem*{lemma*}{Lemma}
\newtheorem{cor}[thm]{Corollary}
\newtheorem*{cor*}{Corollary}
\newtheorem{cj}[thm]{Conjecture}
\newtheorem*{cj*}{Conjecture}
\newtheorem{Def}[thm]{Definition}
\newtheorem*{Def*}{Definition}

\newtheorem*{question*}{Question}

\newtheorem*{problem*}{Problem}

\makeatletter
\def\thmhead@plain#1#2#3{%
  \thmname{#1}\thmnumber{\@ifnotempty{#1}{ }\@upn{#2}}%
  \thmnote{ {\the\thm@notefont#3}}}
\let\thmhead\thmhead@plain
\makeatother

\theoremstyle{definition}
\newtheorem{rem}[thm]{Remark}
\newtheorem*{note}{Note}
\newtheorem{ex}[thm]{Example}

\makeatletter
\newcommand{\manualifempty}[3]{%
  \edef\@tempa{#1}%
  \ifx\@tempa\@empty
    #2
  \else
    #3
  \fi
}
\makeatother

\makeatletter
\newtheoremstyle{manualstyle}
  {3pt}{3pt}{\itshape}{}{\bfseries}{.}{ }{}

\theoremstyle{manualstyle}

\newtheorem{manualthminner}{Theorem}
\newenvironment{manualthm}[1]{%
  \def\@tempa{#1}%
  \ifx\@tempa\@empty
    \renewcommand{\themanualthminner}{\unskip}%
  \else
    \renewcommand{\themanualthminner}{#1}%
  \fi
  \begin{manualthminner}
}{\end{manualthminner}}

\newtheorem{manualpropinner}{Proposition}

\newtheorem{manuallemmainner}{Lemma}
\newenvironment{manuallemma}[1]{%
  \def\@tempa{#1}%
  \ifx\@tempa\@empty
    \renewcommand{\themanuallemmainner}{\unskip}%
  \else
    \renewcommand{\themanuallemmainner}{#1}%
  \fi
  \begin{manuallemmainner}
}{\end{manuallemmainner}}

\newtheorem{manualcorinner}{Corollary}

\makeatother

\newcommand{\bb}{\begin{equation}\begin{aligned}\hspace{0pt}}
\newcommand{\bbb}{\begin{equation*}\begin{aligned}}
\newcommand{\ee}{\end{aligned}\end{equation}}
\newcommand{\eee}{\end{aligned}\end{equation*}}

\newcommand{\eqt}[1]{\stackrel{\mathclap{\scriptsize \mbox{#1}}}{=}}
\newcommand{\leqt}[1]{\stackrel{\mathclap{\scriptsize \mbox{#1}}}{\leq}}

\newcommand{\geqt}[1]{\stackrel{\mathclap{\scriptsize \mbox{#1}}}{\geq}}

\newcommand{\sumno}{\sum\nolimits}
\newcommand{\prodno}{\prod\nolimits}
\newcommand{\e}{\varepsilon}
\renewcommand{\epsilon}{\varepsilon}

\newcommand{\dd}{\mathrm{d}}

\newcommand{\R}{\mathds{R}}
\newcommand{\N}{\mathds{N}}

\newcommand{\E}{\mathds{E}}

\newcommand{\SEP}{\pazocal{S}}

\DeclareMathOperator{\rk}{rk}

\DeclareMathOperator{\co}{conv}

\DeclareMathAlphabet{\pazocal}{OMS}{zplm}{m}{n}

\DeclareMathOperator{\pr}{Pr}
\DeclareMathOperator{\supp}{supp}

\newcommand{\HH}{\pazocal{H}}

\newcommand{\D}{\pazocal{D}}

\newcommand{\XX}{\pazocal{X}}
\newcommand{\TT}{\pazocal{T}}
\newcommand{\PP}{\pazocal{P}}
\newcommand{\FF}{\pazocal{F}}

\newcommand{\lsmatrix}{\left(\begin{smallmatrix}}
\newcommand{\rsmatrix}{\end{smallmatrix}\right)}

\newcommand{\deff}[1]{\textbf{\emph{#1}}}
\newcommand{\rel}[3]{#1\big(#2\,\big\|\,#3\big)}

\stackMath
\newcommand\xxrightarrow[2][]{\mathrel{%
  \setbox2=\hbox{\stackon{\scriptstyle#1}{\scriptstyle#2}}%
  \stackunder[5pt]{%
    \xrightarrow{\makebox[\dimexpr\wd2\relax]{$\scriptstyle#2$}}%
  }{%
   \scriptstyle#1\,%
  }%
}}

\newcommand{\tends}[2]{\xxrightarrow[\! #2 \!]{\mathrm{#1}}}
\newcommand{\tendsn}[1]{\xxrightarrow[\! n\rightarrow \infty\!]{\mathrm{#1}}}

\stackMath

\makeatletter
\newcommand*\rel@kern[1]{\kern#1\dimexpr\macc@kerna}
\newcommand*\widebar[1]{%
  \begingroup
  \def\mathaccent##1##2{%
    \rel@kern{0.8}%
    \overline{\rel@kern{-0.8}\macc@nucleus\rel@kern{0.2}}%
    \rel@kern{-0.2}%
  }%
  \macc@depth\@ne
  \let\math@bgroup\@empty \let\math@egroup\macc@set@skewchar
  \mathsurround\z@ \frozen@everymath{\mathgroup\macc@group\relax}%
  \macc@set@skewchar\relax
  \let\mathaccentV\macc@nested@a
  \macc@nested@a\relax111{#1}%
  \endgroup
}

\counterwithin*{equation}{part}
\counterwithin*{thm}{part}
\counterwithin*{figure}{part}

\tikzset{meter/.append style={draw, inner sep=10, rectangle, font=\vphantom{A}, minimum width=30, line width=.8, path picture={\draw[black] ([shift={(.1,.3)}]path picture bounding box.south west) to[bend left=50] ([shift={(-.1,.3)}]path picture bounding box.south east);\draw[black,-latex] ([shift={(0,.1)}]path picture bounding box.south) -- ([shift={(.3,-.1)}]path picture bounding box.north);}}}
\tikzset{roundnode/.append style={circle, draw=black, fill=gray!20, thick, minimum size=10mm}}
\tikzset{squarenode/.style={rectangle, draw=black, fill=none, thick, minimum size=10mm}}

\definecolor{Blues5seq1}{RGB}{239,243,255}
\definecolor{Blues5seq2}{RGB}{189,215,231}
\definecolor{Blues5seq3}{RGB}{107,174,214}
\definecolor{Blues5seq4}{RGB}{49,130,189}
\definecolor{Blues5seq5}{RGB}{8,81,156}

\definecolor{Greens5seq1}{RGB}{237,248,233}
\definecolor{Greens5seq2}{RGB}{186,228,179}
\definecolor{Greens5seq3}{RGB}{116,196,118}
\definecolor{Greens5seq4}{RGB}{49,163,84}
\definecolor{Greens5seq5}{RGB}{0,109,44}

\definecolor{Reds5seq1}{RGB}{254,229,217}
\definecolor{Reds5seq2}{RGB}{252,174,145}
\definecolor{Reds5seq3}{RGB}{251,106,74}
\definecolor{Reds5seq4}{RGB}{222,45,38}
\definecolor{Reds5seq5}{RGB}{165,15,21}

\allowdisplaybreaks

\usepackage[most,breakable]{tcolorbox}
\makeatletter

\def\boxed@gobegin#1{\def\@tempa{#1}\def\@tempb{orange}\ifx\@tempa\@tempb\begin{tcolorbox}[colback=red!15,colframe=orange!70,breakable,enhanced]\else\begin{tcolorbox}[colback=Blues5seq1,colframe=Blues5seq5,breakable,enhanced]\fi}
\def\boxed@gobegin@empty{\begin{tcolorbox}[colback=Blues5seq1,colframe=Blues5seq5,breakable,enhanced]}
\makeatother

\newcommand{\stein}{\mathrm{Stein}}

\newcommand{\RR}{\pazocal{R}}

\newcommand{\YY}{\pazocal{Y}}
\newcommand{\ZZ}{\pazocal{Z}}
\renewcommand{\SS}{\pazocal{S}}

\renewcommand{\D}{\mathcal{D}}

\renewcommand{\SEP}{\mathrm{SEP}}

\newcommand{\mm}{\pazocal{D}_{\delta,R}}

\renewcommand{\deff}[1]{\emph{#1}}

\newtheorem{ax}{Axiom}
\newtheorem{newax}{Axiom}

\makeatletter

\renewcommand*{\p@subsection}{}

\renewcommand*{\p@subsubsection}{}
\makeatother

\begin{document}

\title{A doubly composite Chernoff--Stein lemma and its applications
}

\author{Ludovico Lami}
\email{ludovico.lami@gmail.com}
\affiliation{Scuola Normale Superiore, Piazza dei Cavalieri 7, 56126 Pisa, Italy}

\begin{abstract}
Given a sequence of random variables $X^n=X_1,\ldots, X_n$, discriminating between two hypotheses on the underlying probability distribution is a key task in statistics and information theory. Of interest here is the Stein exponent, i.e.\ the largest rate of decay (in $n$) of the type II error probability for a vanishingly small type I error probability. When the hypotheses are simple and i.i.d., the Chernoff--Stein lemma states that this is given by the relative entropy between the single-copy probability distributions. Generalisations of this result exist in the case of composite hypotheses, but mostly to settings where the probability distribution of $X^n$ is not genuinely correlated, but rather, e.g., a convex combination of product distributions with components taken from a base set. Here, we establish a general Chernoff--Stein lemma that applies to the setting where both hypotheses are composite and genuinely correlated, satisfying only generic assumptions such as convexity (on both hypotheses) and some weak form of permutational symmetry (on either hypothesis). Our result, which strictly subsumes most prior work, is proved using a refinement of the blurring technique developed in the context of the generalised quantum Stein's lemma \href{https://ieeexplore.ieee.org/abstract/document/10898013}{[Lami, IEEE Trans.\ Inf.\ Theory 2025]}. In this refined form, blurring is applied symbol by symbol, which makes it both stronger and applicable also in the absence of permutational symmetry. The second part of the work is devoted to applications: we provide a single-letter formula for the Stein exponent characterising the discrimination of broad families of null hypotheses vs a composite i.i.d.\ or an arbitrarily varying alternative hypothesis, and establish a `constrained de Finetti reduction' statement that covers a wide family of convex constraints. Applications to quantum hypothesis testing are explored in a related paper~[Lami, arXiv:today].
\end{abstract}

\maketitle

\tableofcontents

\section{Introduction}

\subsection{Background}

Hypothesis testing is a fundamental primitive in statistics, and, as such, an essential ingredient of the scientific method. It also has profound ramifications in information theory~\cite[Chapter~4]{HAN}, where it can be connected, e.g., with coding theory. One of the technical keystones of the theory is the Chernoff--Stein lemma~\cite{stein_unpublished, chernoff_1956}, which establishes an operational interpretation of the Kullback--Leibler divergence~\cite{Kullback-Leibler}, also called the relative entropy, in the task of deciding whether a random variable $X$ is distributed according to a certain law $P$ (null hypothesis) or an alternative law $Q$ (alternative hypothesis), given many i.i.d.\ realisations of $X$. The lemma states that the relative entropy $D(P\|Q)$ coincides with the optimal rate of decay of the probability of a type II error (mistaking $Q$ for $P$), under the constraint that the probability of a type I error (mistaking $P$ for $Q$) be smaller than a fixed threshold. Remarkably, such rate can be connected with the maximum size of reliable codes for communication over a channel~\cite{Feinstein1954, Blackwell1959, Verdu1994}. 

In the decades since its inception, the Chernoff--Stein lemma has been extended in several different directions. Looking at the problem from the point of view of large deviation theory, Sanov~\cite{Sanov1957} (see also~\cite{Hoeffding1965}) generalised it to the case of a composite i.i.d.\ null hypothesis. In this context, composite (i.e.\ non-simple) hypotheses are those that contain not one but many probability distributions, and one is interested in tests that work for all of them --- equivalently, in the worst-case scenario. Composite hypotheses comprising arbitrarily varying sources have been investigated in~\cite[Theorem~4.1]{Fangwei1996} (see also~\cite{Levitan2002}), and in~\cite[Theorem~2]{brandao_adversarial} the analysis has been expanded to encompass also adversarially chosen distributions. The case where the composite hypotheses include a potentially infinite number of distributions has been tackled in~\cite[Theorem~III.7]{Mosonyi2022}.  

Most works so far, however, have dealt with cases where the extremal points of the sets of probability distributions representing the two hypotheses have a product structure across the copies --- i.e.\ the corresponding random variables are independent. Here we are instead interested in treating `genuinely correlated' hypotheses, i.e.\ hypotheses that do \emph{not} have this property. Genuinely correlated but simple (i.e.\ non-composite) hypotheses have been considered already, and can be analysed with the information spectrum method~\cite[Chapter~4]{HAN}. Tackling \emph{composite} genuinely correlated hypotheses, however, requires significantly more effort, as well as more refined tools. 

Our motivation to embark on this endeavour is twofold. First, composite and genuinely correlated hypotheses are the most general class of hypothesis one might think of, and arise naturally in operational contexts --- consider, for example, classes of sources, or channels, with memory. Secondly, they are fundamental in quantum information theory, where, due to the presence of entanglement~\cite{Horodecki-review}, it in general impossible to write a multi-partite quantum state as a convex combination of product states. A paradigmatic example of this behaviour occurs in the setting of the `generalised quantum Stein's lemma', which has attracted much attention recently~\cite{Brandao2010, gap, gap-comment, Hayashi-Stein, GQSL}. Although this may seem like an exquisitely quantum problem, it also reflects back on classical information theory and classical statistics, because many quantum results in hypothesis testing are obtained by `lifting' corresponding classical results. This is the case already for Hiai and Petz's ground-breaking work in proving the original quantum Stein's lemma~\cite{Hiai1991}, as well as for more modern approaches and results~\cite{brandao_adversarial, berta_composite, generalised-Sanov}.

The aforementioned work~\cite{GQSL} introduced a new technique to deal with composite and genuinely correlated hypotheses, called \emph{blurring}. Intuitively, blurring allows us to make a probability distribution more regular by adding some noise to it, thereby `smearing' its weight over nearby type classes. Besides leading to a simple proof of the classical version of the generalised Stein's lemma~\cite[Theorem~4]{GQSL}, the blurring technique has also been used to establish a complementary statement, the generalised quantum Sanov theorem~\cite{generalised-Sanov}.

\subsection{Contribution}

In this paper we prove a generalised, doubly composite version of the classical Chernoff--Stein lemma, which applies to scenarios in which both the null and the alternative hypotheses are not only composite but also genuinely correlated (Theorem~\ref{double_Stein_thm}). Our result holds under a small set of basic compatibility assumptions on the families of probability distributions defining the hypotheses. These assumptions are relatively loose, allowing our theorem to encompass a broad range of previously studied settings, which are subsumed by our general framework. The resulting Stein exponent is given by the minimum regularised relative entropy distance of the single-copy probability distributions in the null hypothesis to the sets representing the alternative hypothesis. 

In general, the regularisation cannot be removed (Example~\ref{additivity_violation_ex}). However, it \emph{can} be removed when the alternative hypothesis is either composite i.i.d.\ or arbitrarily varying, while the null hypothesis is still allowed to be genuinely correlated --- provided it obeys our compatibility assumptions. This is stated in Theorem~\ref{stronger_generalised_Sanov_thm}, which is a relatively straightforward application of Theorem~\ref{double_Stein_thm} but has the advantage of providing a single-letter formula for the Stein exponent. 

These results are obtained by extending and generalising the blurring technique introduced in~\cite{GQSL}. Here we devise a more sophisticated version of this technique that we refer to as `symbol-by-symbol blurring', due to the fact that some noise is added to a given probability distribution over a product space by acting on each of its components independently. The advantage of this approach is that it requires fewer assumptions to be implemented, meaning that the obtained result is more general. In particular, one assumption that we are able to forgo is permutational symmetry on one of the two hypotheses, which is known to be superfluous~\cite{Hayashi-Stein}. On the technical level, our advancements are enabled by more refined estimates on the size of Hamming distance neighbourhoods of large sets in the Hamming space $\XX^n$ (Lemma~\ref{Alon_Spencer_lemma}). The culmination of these efforts is the new \emph{symbol-by-symbol blurring lemma} (Lemma~\ref{sbs_blurring_lemma}).

While conceptually transparent, the blurring technique can become technically cumbersome to wield. Thus, we use the symbol-by-symbol blurring lemma only to fabricate ourselves a handier tool, the \emph{`meta-lemma'} (Lemma~\ref{meta_lemma}; see also the simplified version in Lemma~\ref{meta_lemma_perm_symm}).
To appreciate why this is a much easier statement to handle, consider a family $\FF = (\FF_n)_n$ of sets $\FF_n$ of probability distributions over strings of length $n$ made of symbols taken from some finite alphabet $\XX$. The meta-lemma then formalises an intuitive truth: if $\FF$ represents a physically meaningful hypothesis, then any $Q_n\in \FF_n$ should, with high probability, output strings whose associated empirical probability distribution, that is, the `type' of the string~\cite{CSISZAR-KOERNER}, belongs to $\FF_1$. That is, loosely speaking, $\FF$ should be closed under the operation of taking types. Lemma~\ref{meta_lemma_perm_symm} makes this intuition quantitative, and along the way it will tell us something else: the combined weight of all the strings whose empirical probability distribution is \emph{far} from $\FF_1$ is exponentially suppressed.

The rest of the paper is devoted to presenting the applications of our main results to classical information theory. For applications in quantum information theory, instead, we refer the reader to~\cite{doubly-comp-quantum}. In Corollary~\ref{both_composite_iid_or_av_cor}, we refine earlier results for the case where both hypotheses are either composite i.i.d.\ or arbitrarily varying, while Corollary~\ref{almost_iid_GSL_cor} extends the classical version of the generalised Stein's lemma from~\cite{GQSL}, covering the case of an `almost i.i.d.'\ null hypothesis. Outside the context of hypothesis testing, we obtain a general `constrained de Finetti reduction' statement (Lemma~\ref{constrained_deFinetti_lemma}), which allows us to upper bound any permutationally symmetric probability distribution in $\FF_n$ by a `small' multiple of a convex combination of i.i.d.\ distributions, where only those close to $\FF_1$ are assigned a weight that does not vanish exponentially. Our estimate for the coefficients governing the decay is based on the relative entropy and improves upon the original (quantum) findings from~\cite{Lancien2017}, which employed the fidelity.

The rest of the paper is organised as follows. In Sections~\ref{subsec_general_setting} and~\ref{subsec_prior_results} we formulate the problem and present a brief overview of some prior results. Section~\ref{sec_main_results} then includes the complete technical statements of our main results and of some notable corollaries thereof. In Section~\ref{sec_preliminaries} we present the basic technical tools needed to prove our main results (Theorems~\ref{double_Stein_thm} and~\ref{stronger_generalised_Sanov_thm}), something we then do in Section~\ref{sec_proof_main_result}. In the latter section we also establish our workhorse result, the meta-lemma (Lemma~\ref{meta_lemma}). Section~\ref{sec_classical_applications} is then devoted to the applications of our methods.

\subsection{General setting} \label{subsec_general_setting}

In its most basic form, the task of classical hypothesis testing can be defined as follows. Let $X^n=X_1,\ldots, X_n$ be a string of $n$ random variables from a finite alphabet $\XX$, which might represent readings of a physical instrument, output signals of a channel, or something else entirely. We will denote as $\PP(\XX)$ the set of probability distributions on $\XX$. 


While we do not know the probability distribution that has generated the string, we are promised that one of the following two hypotheses holds:
\begin{itemize}
\item[$\mathrm{H}_0$.] Null hypothesis: $X^n \sim P_n$, for some $P_n\in \RR_n$;
\item[$\mathrm{H}_1$.] Alternative hypothesis: $X^n \sim Q_n$, for some $Q_n\in \SS_n$.
\end{itemize}
Our goal is to guess which option is the correct one. Here,
\bb
\RR_n,\,\SS_n \subseteq \PP(\XX^n)
\ee
are two a priori generic sets of probability distributions on $n$ copies of the alphabet $\XX$, which we can collect into two sequences $\RR = (\RR_n)_n$ and $\SS = (\SS_n)_n$. Our goal is to make a guess as to which hypothesis holds by looking only at the realisation of $X^n$.

Stated in these general terms, the problem subsumes many known scenarios, e.g.\ those corresponding to the following choices of the sets $\RR_n$ and $\SS_n$: 
\begin{itemize}
\item Simple i.i.d.\ hypotheses: 
\bb
\RR_n = \big\{P^{\otimes n}\big\}\, ,\qquad \SS_n = \big\{Q^{\otimes n}\big\}\, , 
\ee
for some fixed $P,Q$. These hypotheses are called `simple' because they comprise single probability distributions. 

\item Composite i.i.d.\ hypotheses: for some base sets $\RR_1,\SS_1\subseteq \PP(\XX)$, 
\bb
\RR = \RR_1^{\mathrm{iid}} &\coloneqq \big( \RR_1^{\otimes n,\,\mathrm{iid}}\big)_n\, \qquad \RR_1^{\otimes n,\,\mathrm{iid}} \coloneqq \big\{P^{\otimes n}:\, P\in \RR_1\big\}\, ,
\label{F_n_iid} 
\ee
and analogously for $\SS_1$. These hypotheses are non-simple, i.e.\ they are composite, because they comprise multiple probability distributions.

\item Composite arbitrarily varying hypotheses: for some base sets $\RR_1,\SS_1\subseteq \PP(\XX)$ of probability distributions on $\XX$,
\bb
\RR = \RR_1^{\mathrm{av}} &\coloneqq \big( \RR_1^{\otimes n,\,\mathrm{av}}\big)_n\, \qquad \RR_1^{\otimes n,\, \mathrm{av}} \coloneqq \big\{P_1\!\otimes\! \ldots \!\otimes P_n\!:\ P_1,\ldots, P_n\in \RR_1\big\}\, ,
\label{F_n_av}
\ee
and the same for $\SS_1$.
\end{itemize}

Naturally, hybrid settings are also possible --- for instance, scenarios in which one of the two hypotheses is simple i.i.d.\ while the other is composite i.i.d. However, it is even more interesting for us to consider broader classes of composite hypotheses, whose underlying probability distributions do not exhibit a product structure over the $X_i$ variables. We refer to such hypotheses as \emph{genuinely correlated}. (We are not interested in defining this term rigorously, but a possible definition would be as follows: a convex set of probability distributions over $\XX^n$ is genuinely correlated if some of its extreme points are not product distributions.) Our main result, Theorem~\ref{double_Stein_thm} below, applies to general classes of hypotheses and subsumes, as special cases, the simple i.i.d., composite i.i.d., and arbitrarily varying settings, as well as genuinely correlated ones.

A natural goal of hypothesis testing is to design suitable tests that minimise the error probabilities. There are two different types of errors:
\begin{itemize} 
\item \emph{Type I error}: $\mathrm{H}_0$ was correct, but we guessed $\mathrm{H}_1$.

\item \emph{Type II error}: $\mathrm{H}_1$ was correct, but we guessed $\mathrm{H}_0$.
\end{itemize}

In this context, a (probabilistic) \deff{test} is simply a function $A_n :\XX^n\to [0,1]$, where $A(x^n)$ represents the probability that we guess $\mathrm{H}_0$ upon seeing the string $x^n$. The worst-case probabilities of the two types of error are 
\bb
\alpha_n(A_n) \coloneqq \sup_{P_n\in \RR_n} \sum_{x^n\in \XX^n} \big(1-A_n(x^n)\big) P_n(x^n)\, ,\qquad \beta_n(A_n) \coloneqq \sup_{Q_n\in \SS_n} \sum_{x^n\in \XX^n} A_n(x^n) Q_n(x^n)\, ,
\ee
respectively, where the dependence on $\RR_n$ and $\SS_n$ is implicit. Note that the above error probabilities are left invariant if we replace $\RR_n$ and $\SS_n$ by their convex hulls. The minimal type II error probability for a given constraint on the type I error probability is thus obtained as
\bb
\beta_\e(\RR_n\|\SS_n) \coloneqq \inf\left\{ \beta_n(A_n):\ \ A_n:\XX^n\to [0,1],\ \alpha_n(A_n)\leq \e \right\} .
\label{beta_e_level_n}
\ee
In many applications, including coding theory and quantum information theory, it is of interest to minimise the rate of decay in $n$ of $\beta_\e(\RR_n\|\SS_n)$. We can formalise this by introducing the \deff{Stein exponent} between the hypotheses $\RR = (\RR_n)_n$ and $\SS = (\SS_n)_n$, defined as
\bb
\stein(\RR\|\SS) \coloneqq \lim_{\e\to 0^+} \liminf_{n\to\infty} \left\{ -\frac1n\log \beta_\e(\RR_n\|\SS_n)\right\} .
\label{Stein}
\ee

Our goal is to calculate the above limit with a limited set of assumptions on $\RR$ and $\SS$, and, in particular, for some interesting classes of genuinely correlated hypotheses. To this end, we begin by recalling an important set of axioms introduced by Brand\~{a}o and Plenio~\cite{BrandaoPlenio2, Brandao2010} (see also~\cite{Brandao-Gour}), which we therefore refer to as the \deff{Brand\~{a}o--Plenio axioms}.\footnote{We have adapted them to the classical setting, as the original axioms concern quantum states. The translation is however straightforward.}  Although we will \emph{not} rely on these axioms in our analysis, they have played a historically important role and provide a useful point of comparison. In terms of a generic sequence $(\FF_n)_n$ of sets $\FF_n\subseteq \PP(\XX)$, which might represent either of the two hypotheses, they can be stated as follows:

{\renewcommand{\theax}{BP\arabic{ax}}
\begin{ax} \label{BP_ax_convex_closed}
Each $\FF_n$ is a convex and closed subset of $\PP(\XX^{n})$.
\end{ax}

\begin{ax} \label{BP_ax_full_rank}
$\FF_1$ contains some probability distribution $R\in \FF_1$ with full support, i.e.\ such that $\min_{x\in \XX} R(x) \geq c > 0$.
\end{ax}

\begin{ax} \label{BP_ax_partial_traces}
The family $(\FF_{n})_n$ is closed under partial traces, i.e.\ if $n\in \N^+$ and $Q_n = Q_{X_1\ldots X_{n}}\in \FF_{n}$, then $Q_{X_1\ldots X_{n-1}} \in \FF_{n-1}$, where $Q_{X_1\ldots X_{n-1}}$ denotes the probability distribution obtained by discarding the last symbol.
\end{ax}

\begin{ax} \label{BP_ax_tensor_products}
The family $(\FF_{n})_n$ is closed under tensor products: if $Q_n\in \FF_n$ and $Q'_m\in \FF_m$, then the product distribution belongs to $\FF_{n+m}$, i.e.\ $Q_n\otimes Q'_m\in \FF_{n+m}$.
\end{ax}

\begin{ax} \label{BP_ax_permutations}
Each $\FF_n$ is closed under permutations: if $Q_n \in \FF_n$ and $\pi\in S_n$ denotes an arbitrary permutation of a set of $n$ elements, then also $Q_n \circ \pi \in \FF_n$, where $\pi$ acts on $\XX^n$ by permuting the string symbols. 
\end{ax}
}

For how operationally reasonable the Brand\~{a}o and Plenio axioms might be, we will not adopt them in this form, for at least three reasons. First, they do not subsume all of the above basic settings. Namely, a composite i.i.d.\ hypothesis of the form $\FF_n = \co\big(\FF_1^{\otimes n,\, \mathrm{iid}}\big)$, where $\FF_1\subseteq \PP(\XX)$ and $\FF_1^{\otimes n,\, \mathrm{iid}}$ is defined as in~\eqref{F_n_iid}, violates Axiom~\ref{BP_ax_tensor_products}, simply because the tensor product of different i.i.d.\ distributions is not itself i.i.d. Secondly, recent approaches to the generalised quantum Stein's lemma~\cite{Hayashi-Stein} have shown that some of these axioms on the alternative hypothesis can be removed --- specifically, Axioms~\ref{BP_ax_partial_traces} and~\ref{BP_ax_permutations} (see below). Thirdly, it has also been shown that the fact that the null hypothesis satisfies the Brand\~{a}o--Plenio axioms does not suffice to calculate the Stein exponent, even when the alternative hypothesis is simple and i.i.d.~\cite[Appendix~E.2]{generalised-Sanov}. For all these reasons, we will base our analysis on a somewhat different set of axioms (see Section~\ref{subsec_new_axioms}).

\subsection{Prior results} \label{subsec_prior_results}

In the case of two simple i.i.d.\ hypotheses represented by probability distributions $P$ and $Q$ (see above), the Chernoff--Stein lemma~\cite{stein_unpublished, chernoff_1956} states that
\bb
\stein(P\|Q) = D(P\|Q) \coloneqq \sum_{x\in \XX} P(x) \log \frac{P(x)}{Q(x)}\, ,
\label{KL}
\ee
where $D(P\|Q)$ is the \emph{relative entropy}, also called the \emph{Kullback--Leibler divergence}. Note that, with a slight abuse of notation, we identified
\bb
\stein(P\|Q) \coloneqq \stein\big(\big(\{P^{\otimes n}\}\big)_n\, \big\|\, \big(\{Q^{\otimes n}\}\big)_n\big)\, .
\ee

Several generalisations of the Chernoff--Stein lemma are known. Without any claim of completeness, here we list some of the most notable ones. To simplify the notation, we adopt the conventions from~\eqref{F_n_iid}--\eqref{F_n_av}. We also henceforth establish the following notation: for a function $\mathds{D} : \PP(\XX) \times \PP(\XX) \to \R \cup \{+\infty\}$ and any two sets $\RR_1,\SS_1\subseteq \PP(\XX)$, we set
\bb
\mathds{D}(\RR_1\|\SS_1) \coloneqq \inf_{P\in \RR_1,\ Q\in \SS_1} \mathds{D}(P\|Q)\, .
\label{divergence_sets}
\ee
We will also write compactly $\mathds{D}(\{P\}\|\SS_1) = \mathds{D}(P\|\SS_1)$ if, say, the first set is a singlet. 

\begin{enumerate}[(A)]
\item When the alternative hypothesis is simple but the null hypothesis is composite i.i.d., Sanov showed that~\cite{Sanov1957, Hoeffding1965}
\bb
\rel{\stein}{\RR_1^\mathrm{iid}}{Q} = D(\RR_1\|Q)
\label{Sanov_theorem}
\ee
for all closed sets $\RR_1\subseteq \PP(\XX)$. On the left-hand side the symbol $Q$ is again a shorthand for the sequence of simple hypotheses $\big(\{Q^{\otimes n}\}\big)_n$. 

\item It is also known that~\cite[Theorem~III.2]{Mosonyi2022}
\bb
\rel{\stein}{\RR_1^\mathrm{iid}}{\SS_1^\mathrm{iid}} = D(\RR_1\|\SS_1)
\label{Stein_iid_without_convexity}
\ee
for all pairs of finite sets of probability distributions $\RR_1,\SS_1\subseteq \PP(\XX)$. 

\item For any two closed sets $\RR_1,\SS_1\subseteq \PP(\XX)$, it holds that~\cite[Theorem~III.7]{Mosonyi2022}
\bb
\rel{\stein}{\RR_1^{\mathrm{av}}}{\SS_1^{\mathrm{av}}} &= \rel{\stein}{\co(\RR_1)^{\mathrm{av}}}{\co(\SS_1)^{\mathrm{av}}} \\
&= \rel{\stein}{\co(\RR_1)^{\mathrm{iid}}}{\co(\SS_1)^{\mathrm{iid}}} \\
&= \rel{D}{\co(\RR_1)}{\co(\SS_1)}\, .
\label{Stein_av}
\ee

\item In the case where the hypotheses are composite i.i.d.\ or arbitrarily varying, with convex and closed base sets $\RR_1,\SS_1\subseteq \PP(\XX)$, we have~\cite{Fangwei1996, Levitan2002, brandao_adversarial,Fang2025}\footnote{It is not difficult to show  that~\eqref{Stein_iid_or_av_convex} actually subsumes~\eqref{Stein_av}.}
\bb
\rel{\stein}{\RR_1^{\mathrm{a}}}{\SS_1^{\mathrm{b}}} = D(\RR_1\|\SS_1) \qquad \forall\ \mathrm{a},\mathrm{b}\in \{\mathrm{iid},\mathrm{av}\}\, .
\label{Stein_iid_or_av_convex}
\ee

\item \emph{Generalised classical Stein's lemma}~\cite{Hayashi-Stein, GQSL}: for a simple i.i.d.\ null hypothesis represented by $P$ and a composite (and possibly genuinely correlated) alternative hypothesis $\SS = (\SS_n)_n$ that satisfies Axioms~\ref{BP_ax_convex_closed},~\ref{BP_ax_full_rank}, and~\ref{BP_ax_tensor_products}, it holds that
\bb
\stein(P\|\SS) = D^\infty(P\|\SS) \coloneqq \lim_{n\to\infty} \frac1n\, \min_{Q_n\in \SS_n} D(P^{\otimes n}\| Q_n)\, .
\label{GSL}
\ee
This version of the result, which does not rely on Axioms~\ref{BP_ax_partial_traces} and~\ref{BP_ax_permutations}, is due to~\cite[Theorem~1]{Hayashi-Stein}. In~\cite{GQSL}, all the Brand\~{a}o--Plenio axioms are assumed instead, yielding a stronger statement that works even for a certain class of `almost i.i.d.'\ null hypotheses. Denoting with $\RR^{\mathrm{aiid}}_{r,P}$ the sequence of sets of probability distributions on the random variable $X^n = (X_1,\ldots, X_n)$ such that, for all $n$, at least $n-r$ among the $X_i$'s are independent and distributed according to $P$, it follows from~\cite[Theorem~32]{GQSL} that
\bb
\rel{\stein}{\RR^{\mathrm{aiid}}_{r,P}}{\SS} = D^\infty(P\|\SS)
\label{GSL_almost_iid}
\ee
for all $r\in \N^+$ and $P\in \PP(\XX)$, provided that $\SS = (\SS_n)_n$ satisfies Axioms~\ref{BP_ax_convex_closed}--\ref{BP_ax_permutations}. We will explain and strengthen this result in Section~\ref{subsec_GSL_almost_iid}. Note that~\eqref{GSL_almost_iid} is the first extension of the Chernoff--Stein lemma that deals with the case where \emph{both} hypotheses are genuinely correlated --- albeit, admittedly, this is more of a formal rather than a conceptual difference. 

\item \emph{Generalised classical Sanov theorem}~\cite{generalised-Sanov, Hayashi-Sanov-2}: In~\eqref{GSL}, we considered an i.i.d.\ null hypothesis and a general alternative hypothesis, but we can also investigate the opposite scenario in which $\RR = (\RR_n)_n$ is general, while $\SS = \big(\{Q^{\otimes n}\}\big)_n$ is i.i.d. However, it turns out that assuming only the Brand\~{a}o--Plenio axioms on $\RR$ does not yield a simple expression for the Stein exponent~\cite[Appendix~E.2]{generalised-Sanov}. To remedy this, one needs to impose an additional regularity assumption, and there is some arbitrariness in this choice. In~\cite{generalised-Sanov}, the choice fell on the following axiom, stated here for a general sequence $\FF = (\FF_n)_n$:

{\renewcommand{\theax}{BP\arabic{ax}} 
\begin{ax} \label{BP_ax_6}
The function $D^\infty(\cdot\|\FF)$ of~\eqref{GSL} is faithful on $\FF_1$, i.e.\ $D^\infty(P\|\FF) > 0$ whenever $P\notin \FF_1$.
\end{ax}
}

Now, if $\RR=(\RR_n)_n$ satisfies Axioms~\ref{BP_ax_convex_closed}--\ref{BP_ax_permutations} and also Axiom~\ref{BP_ax_6}, for all $Q\in \PP(\XX)$ we have~\cite[Theorem~8]{generalised-Sanov}
\bb
\stein(\RR\|Q) = D(\RR_1\|Q)\, .
\label{generalised_Sanov}
\ee
Notably, this shows that the Stein exponent is given by a single-letter expression, in spite of the fact that the null hypothesis can be genuinely correlated. This is in stark contrast with~\eqref{GSL}, which features a regularised expression on the right-hand side. A result similar to~\eqref{generalised_Sanov}, albeit relying on a slightly different set of assumptions, is obtained in~\cite[Theorem~7]{Hayashi-Sanov-2}.

\item Another result that deals with the case where both hypotheses are composite and genuinely correlated was obtained in~\cite[Theorem~25]{Fang2025}. The required assumptions, however, are rather restrictive~\cite[Assumption~24]{Fang2025}, and are typically not satisfied by many relevant sets of probability distributions. For example, composite i.i.d.\ hypotheses violate~\cite[Assumption~24(A.3)]{Fang2025}, and, perhaps more importantly, the families of classical probability distributions obtained by measuring fundamental sets of quantum states such as separable states~\cite{Werner} or stabiliser states~\cite{Veitch2014} violate~\cite[Assumption~24(A.4)]{Fang2025}.

\end{enumerate}

\section{Main results} \label{sec_main_results}

\subsection{New axioms} \label{subsec_new_axioms}

To formulate our general result on the calculation of classical Stein exponents, we start by discussing the axiomatic framework underpinning it. An important definition in this regard is the following.

\begin{Def} \label{depolarising_def}
Given a finite alphabet $\XX$, some $\delta\in [0,1]$, and a probability distribution $R\in \PP(\XX)$ on $\XX$, we denote with $\mm:\PP(\XX)\to \PP(\XX)$ the channel that replaces the input symbol with a symbol drawn from $R$ with probability $\delta$, and acts as the identity channel with probability $1-\delta$. In other words, 
\bb
\left(\mm(P)\right)(x) = (1-\delta) P(x) + \delta R(x)\, .
\label{depolarising}
\ee
\end{Def}

{\renewcommand{\thenewax}{\Roman{newax}}

We can use the above map $\mm$ to state our central assumption:

\begin{newax} \label{new_ax_depolarising}
There exists some $R\in \PP(\XX)$ such that, for all $n\in \N^+$ and all $Q_n\in \FF_n$:
\vspace{-1.2ex}
\begin{enumerate}[(a)] 
\item $\supp(Q_n) \subseteq \supp(R)^n$; and 
\item $\mm^{\otimes n}(Q_n)\in \FF_n$ for all $\delta\in [0,1]$, where $\mm$ is as in Definition~\ref{depolarising_def}. 
\end{enumerate}
\vspace{-1ex}
We denote by $c$ a constant with the property that $\min_{x\in \supp(R)} R(x) \geq c > 0$.
\end{newax}

In~\cite{GQSL}, the Brand\~{a}o--Plenio axioms are used to implement a procedure called \emph{blurring}, in which some noise is added to a probability distribution to make it more regular. One of the conceptual contributions of this paper is to recognise that the same effect can be achieved by means of the much weaker Axiom~\ref{new_ax_depolarising}, which, in the context of our work, should thus be viewed as a sort of condensed version of Axioms~\ref{BP_ax_convex_closed}--\ref{BP_ax_permutations}. We refer to the new blurring procedure that is enabled by Axiom~\ref{new_ax_depolarising} as \emph{symbol-by-symbol blurring}, to reference the fact that the blurring effect will be obtained by applying the map $\mm$ independently to every symbol of the input string --- equivalently, to every random variable. The new statement replacing the classical blurring lemma of~\cite[Lemma~9]{GQSL} is the forthcoming Lemma~\ref{sbs_blurring_lemma}. 

Among the immediate advantages of adopting Axiom~\ref{new_ax_depolarising} over the Brand\~{a}o--Plenio axioms, we note that the former can also cover the case of a composite i.i.d.\ hypothesis $\FF_n = \FF_1^{\otimes n,\, \mathrm{iid}}$ with convex base set $\FF_1$, defined as in~\eqref{F_n_iid}, which, as we saw before, violates Axiom~\ref{BP_ax_tensor_products}.

We now introduce a weakened version of Axiom~\ref{BP_ax_tensor_products}, followed by the original statement for completeness.

\begin{newax} \label{new_ax_tensor_powers}
$(\FF_n)_n$ is closed under tensor powers from $\FF_1$, in the sense that $Q_1^{\otimes n} \in \FF_{n}$ for all $Q_1\in \FF_1$ and all $n\in \N^+$.
\end{newax}

{\renewcommand{\thenewax}{\Roman{newax}+}
\setcounter{newax}{1}
\begin{newax} \label{new_ax_tensor_products_n_and_m}
The family $(\FF_{n})_n$ is closed under tensor products: if $Q_n\in \FF_n$ and $Q'_m\in \FF_m$, then $Q_n\otimes Q'_m\in \FF_{n+m}$.
\end{newax}
}

For completeness, we also report again Axiom~\ref{BP_ax_permutations} on the closedness of $\FF_n$ under permutations, unchanged, together with a stronger form that will be useful later on:

\begin{newax} \label{new_ax_closed_permutations}
Each $\FF_n$ is closed under permutations: if $Q_n \in \FF_n$ and $\pi\in S_n$ denotes an arbitrary permutation of a set of $n$ elements, then also $Q_n \circ \pi \in \FF_n$, where $\pi$ acts on $\XX^n$ by permuting the string symbols. 
\end{newax}


{\renewcommand{\thenewax}{\Roman{newax}{\large +}}
\setcounter{newax}{2}
\begin{newax} \label{new_ax_permutational_symmetry}
Each $\FF_n$ contains only permutationally symmetric probability distributions.
\end{newax}
}

As mentioned, Axioms~\ref{new_ax_depolarising}--\ref{new_ax_closed_permutations} are directly implied by the original Brand\~ao--Plenio axioms (Lemma~\ref{BP_ax_1_5_imply_new_ax_1_3_lemma}). However, as already mentioned, even Axioms~\ref{BP_ax_convex_closed}--\ref{BP_ax_permutations} together do not appear to suffice to solve the Stein exponent~\cite[Appendix~E.2]{generalised-Sanov}, making it necessary to introduce an additional assumption of a different nature. In~\cite{generalised-Sanov} we chose Axiom~\ref{BP_ax_6}; here, we distil this condition down to the following: if $Q_n\in \FF_n$ outputs strings whose type is close to $P$ `too often', i.e.\ with probability that vanishes \emph{sub-}exponentially for large $n$, then it must be the case that $P\in \FF_1$:

\begin{newax}[(Type stability)] \label{new_ax_type_stability}
If a probability distribution $P\in \PP(\XX)$ is such that there exists a constant $K>0$ with the property that, for all $\delta>0$,
\bb
\sup_{Q_n\in \FF_n} \pr_{X^n\sim Q_n}\!\big\{ \tfrac12 \|P_{X^n} - P \|_1 \leq \delta \big\} \geq \frac{1}{n^K}
\ee
holds for infinitely many values of $n$, then $P\in \FF_1$. Here, $P_{X^n}$ is the type of the string $X^n$.
\end{newax}

As a corollary of our results, we will see later that, in the presence of Axioms~\ref{BP_ax_convex_closed}--\ref{BP_ax_permutations}, the above Axiom~\ref{new_ax_type_stability} is implied by, and hence strictly weaker than, Axiom~\ref{BP_ax_6} (Lemma~\ref{BP_ax_1_6_imply_new_ax_1_4_lemma}).

This exhausts the list of axioms we will actually need in order to prove our doubly composite Chernoff--Stein lemma. However, it is helpful for the applications to state two more assumptions, which, when satisfied, make our life easier. The first one is inspired by the work by Piani~\cite{Piani2009}; it allows us to verify immediately the slightly obscure Axiom~\ref{new_ax_type_stability}:

\begin{newax} \label{new_ax_filtering}
There exists a classical channel $W:\XX\to \YY$ (with $|\YY|<\infty$) such that:
    \begin{enumerate}[A., topsep=0pt,itemsep=0pt]
    \item $W$ is informationally complete, in the sense that the output statistics determines the input completely;\footnote{In other words, $\rk \big(W(y|x)\big)_{x,y} = |\XX|$, where $\rk$ is the matrix rank.}
    \item $W$ is compatible with $(\FF_n)_n$, in the sense that for all $Q_n = Q_{X_1\ldots X_n}\in \FF_n$ and all $y_n\in \YY$, defining $Y_n\coloneqq W(X_n)$ we have $Q_{X_1\ldots X_{n-1}\,|\,Y_n=y_n} \in \FF_{n-1}$.
    \end{enumerate}
\end{newax}
}

\subsection{Main result: doubly composite Chernoff--Stein lemma} \label{subsec_main_result}

We are now ready to state our general, doubly composite Chernoff--Stein lemma:

\begin{thm}[(Doubly composite Chernoff--Stein lemma)] \label{double_Stein_thm}
Let $\XX$ be a finite alphabet, and let $\RR = (\RR_n)_n$ and $\SS = (\SS_n)_n$ be two families of sets of probability distributions $\RR_n,\SS_n\subseteq \PP(\XX^n)$, representing the null and the alternative hypotheses, respectively. Assume that:
\begin{enumerate}[(a)]
\item $\RR$ satisfies Axioms~\ref{new_ax_tensor_powers} and~\ref{new_ax_type_stability}; also, $\RR_1$ is topologically closed;
\item $\SS$ satisfies Axiom~\ref{new_ax_depolarising};
\item either $\RR$ satisfies Axiom~\ref{new_ax_permutational_symmetry}, or $\SS$ satisfies Axiom~\ref{new_ax_closed_permutations}.
\end{enumerate}
Then the Stein exponent, defined by~\eqref{Stein}, is given by
\bb
\stein(\RR\|\SS) = \inf_{P\in \RR_1} D^\infty(P\|\co(\SS)) = \inf_{P\in \RR_1} \liminf_{n\to\infty} \frac1n\, \rel{D}{P^{\otimes n}}{\co(\SS_n)}\, .
\label{double_Stein}
\ee
In particular, Eq.~\eqref{double_Stein} holds under assumption~(b), if in addition
\begin{enumerate}[(a')]
\item $\RR$ satisfies Axioms~\ref{new_ax_depolarising},~\ref{new_ax_tensor_powers}, and~\ref{new_ax_filtering}, all sets $\RR_n$ are convex, and $\RR_1$ is topologically closed; and 
\setcounter{enumi}{2}
\item either $\RR$ satisfies Axiom~\ref{new_ax_permutational_symmetry}, or both $\RR$ and $\SS$ satisfy Axiom~\ref{new_ax_closed_permutations}.
\end{enumerate}
\end{thm}

The proof can be found in Section~\ref{subsec_proof_double_Stein}. Here, we will instead discuss some notable aspects of the above result. First, it provides an explicit solution for the Stein exponent of a general class of hypothesis testing tasks, where both hypotheses are allowed to be both composite and genuinely correlated. It is interesting to observe that the requirements on the null hypothesis are in general stronger than those on the alternative hypothesis. As we mentioned already, this is somewhat unavoidable (see~\cite[Appendix~E.2]{generalised-Sanov}).

Secondly, we will see in Section~\ref{sec_classical_applications} that the assumptions of Theorem~\ref{double_Stein_thm}, while simple to state, are general enough to encompass as special cases --- and, in many case, refine --- almost all previously known results, including those presented in~(A), (C), (D), (E), and~(F) in Section~\ref{subsec_prior_results}.\footnote{Curiously, however, that in~(B) does not seem to fit into our framework. Also~\cite[Theorem~1]{Hayashi-Stein} and~\cite[Theorem~25]{Fang2025} are incomparable to our Theorem~\ref{double_Stein_thm}, as they rely on slightly different sets of assumptions. See the discussion below.}
It is instructive, for example, to examine what our Theorem~\ref{double_Stein_thm} predicts in the very special situation where $\RR_n = \big\{ P^{\otimes n} \big\}$ is simple and i.i.d. In this case, the only constraints imposed on $\SS$ is that it satisfies Axiom~\ref{new_ax_depolarising}, and Theorem~\ref{double_Stein_thm} markedly improves on the generalised (classical) Stein's lemma of~\cite[Theorem~4]{GQSL}, which hinges on all the Brand\~{a}o--Plenio axioms (Axioms~\ref{BP_ax_convex_closed}--\ref{BP_ax_permutations}). Most notably, it does away with the assumption of closure under permutations, showing that the blurring technique can circumvent it. The statement one obtains is, strictly speaking, incomparable with the classical case of~\cite[Theorem~1]{Hayashi-Stein}, which requires closure under tensor products and the existence of a full-support element in $\SS_1$, rather than Axiom~\ref{new_ax_depolarising}. The former assumptions, however, tend to be somewhat more stringent than Axiom~\ref{new_ax_depolarising} in practice: for instance, they are violated in the paradigmatic case of a composite i.i.d.\ hypothesis, which is not closed under tensor products. We will also see in Corollary~\ref{almost_iid_GSL_cor} that our techniques can improve upon~\cite[Theorem~32]{GQSL} and handle the general case of an `almost i.i.d.'\ null hypothesis, which does not seem amenable to the methods of~\cite{Hayashi-Stein}.

Thirdly, one may wish to compare our Theorem~\ref{double_Stein_thm} with the classical case of the quantum~\cite[Theorem~25]{Fang2025}, which likewise addresses the general setting where both hypotheses are composite and genuinely correlated. Although the two results rest on incomparable sets of assumptions, we already noted in Section~\ref{subsec_prior_results}(G) that~\cite[Assumption~24]{Fang2025} excludes many interesting families of probability distributions --- for instance, those obtained by measuring the sets of separable or stabiliser quantum states. Consequently, \cite[Theorem~25]{Fang2025} cannot be applied to the quantum hypothesis testing problems studied in the companion paper~\cite{doubly-comp-quantum}, which are instead amenable to an attack consisting of a quantum-to-classical reduction and, ultimately, Theorem~\ref{double_Stein_thm}.

Lastly, the formula~\eqref{double_Stein} for the Stein exponent involves a regularisation, i.e.\ an asymptotic limit over the number of symbols $n$. One might hope to remove this limit and obtain instead the single-letter distance $D(\RR_1\|\SS_1)$. However, we will show with a simple example (Example~\ref{additivity_violation_ex}) that, \emph{in general}, this is not possible. Indeed, we deem it unlikely that, in the very broad setting we consider here, a universal single-letter formula for the Stein exponent might exist.

\subsection{A key tool: the meta-lemma}

The fundamental tool we will use to prove Theorem~\ref{double_Stein_thm} is an improved version of the blurring technique from~\cite{GQSL}. Blurring, however, is not applied directly; instead, we first use it to establish an intuitive statement that we call a \emph{meta-lemma}. We include it here because we find it of independent conceptual interest. Roughly speaking, it asserts that any sequence of hypotheses $\FF = (\FF_n)_n$ satisfying Axiom~\ref{new_ax_depolarising} must have the following property: if some $Q_n \in \FF_n$ is `sufficiently flat' on a type class\footnote{See~\eqref{all_types} and~\eqref{type_class} for definitions related to the notion of type.} $T_{n,V}$, in the sense that $Q_n(x^n) \approx \frac{q_n}{|T_{n,V}|}$ for a significant fraction of strings $x^n \in T_{n,V}$, then
\bb
- \log q_n \gtrsim \rel{D}{V^{\otimes n}}{\FF_n}\, .
\ee
Typically, this entails that $q_n$ is exponentially suppressed unless $V$ is close to $\FF_1$. A technically precise statement is as follows:

\begin{lemma}[(Meta-lemma)] \label{meta_lemma}
For a finite alphabet $\XX$, let $(\FF_n)_n$ be a sequence of sets $\FF_n\subseteq \PP(\XX^n)$ that obeys Axiom~\ref{new_ax_depolarising} with respect to a probability distribution $R\in \PP(\XX)$ and a constant $c$ such that $\min_{x\in \supp(R)} R(x) \geq c > 0$. Take two real-valued functions $o_L(n)$ and $o_R(n)$ with the property that $\lim_{n\to\infty} \frac{o_L(n)}{n} = \lim_{n\to\infty} \frac{o_R(n)}{n} = 0$. For any $\Delta>0$, we can find $N = N(\Delta,c,o_L,o_R,|\XX|) \in \N^+$ such that, for all integers $n\geq N$, the following holds: given some $Q_n\in \FF_n$, an $n$-type $V\in \TT_n$, $P\in \PP(\XX)$ with $\supp(P)\subseteq \supp(R)$ and $\frac12 \|V-P\|_1\leq \xi\in (0,1/3)$, and some $\lambda\geq 0$, if
\bb
\left|\left\{x^n\in T_{n,V}:\ Q_n(x^n) \geq \frac{\exp[-n\lambda - o_L(n)]}{|T_{n,V}|} \right\}\right| \geq \exp[- o_R(n)]\, |T_{n,V}|\, ,
\label{meta_lemma_0}
\ee
then
\bb
\frac1n\, D\big(P^{\otimes n}\, \big\|\, \FF_n\big) \leq \lambda + \phi(\xi) + \Delta\, ,
\label{meta_lemma_1}
\ee
where $\phi$ is a continuous function that depends only on $c$ and $|\XX|$ and vanishes at $0$.
\end{lemma}

\subsection{A general single-letter formula for the Stein exponent} \label{subsec_main_result_applications}

While unavoidable in general, the regularised formula in~\eqref{double_Stein} is typically difficult to handle analytically. Our most notable application of Theorem~\ref{double_Stein_thm}, therefore, is to the setting where the alternative hypothesis is either composite i.i.d.\ or arbitrarily varying; in all those cases it is possible to remove the regularisation and give a single-letter formula for the Stein exponent:

\begin{thm} \label{stronger_generalised_Sanov_thm}
Let $\XX$ be a finite alphabet, $\SS_1\subseteq \PP(\XX)$ a set of probability distributions on $\XX$, and $\RR = (\RR_n)_n$ a family of sets $\RR_n\subseteq \PP(\XX^n)$. Assume that either
\begin{enumerate}
\item[(a)] $\RR$ satisfies Axioms~\ref{new_ax_tensor_powers} and~\ref{new_ax_type_stability}; also, $\RR_1$ is topologically closed; or
\item[(a')] $\RR$ satisfies Axioms~\ref{new_ax_depolarising},~\ref{new_ax_tensor_powers},~\ref{new_ax_closed_permutations}, and~\ref{new_ax_filtering}, all sets $\RR_n$ are convex, and $\RR_1$ is topologically closed.
\end{enumerate}
Then, with the notation in~\eqref{F_n_av}, the Stein exponent defined as in~\eqref{Stein} is given by
\bb
\rel{\stein}{\RR}{\SS_1^{\mathrm{av}}} &= D(\RR_1\|\co(\SS_1)) = \inf_{P\in \RR_1,\, Q\in \co(\SS_1)} D(P\|Q)\, .
\label{stronger_generalised_Sanov_av}
\ee
If, moreover, 
\begin{enumerate}[(a)] \setcounter{enumi}{1}
\item $\SS_1$ is star-shaped around some $R\in \SS_1$ such that $\supp(Q)\subseteq \supp(R)$ for all $Q\in \SS_1$,
\end{enumerate}
then it also holds that
\bb
\rel{\stein}{\RR}{\SS_1^{\mathrm{iid}}} &= D(\RR_1\|\SS_1) = \inf_{P\in \RR_1,\, Q\in \SS_1} D(P\|Q)\, , \label{stronger_generalised_Sanov_iid}
\ee
where the notation is defined in~\eqref{F_n_iid} and~\eqref{Stein}.
\end{thm}

The above result, proved in Section~\ref{subsec_proof_stronger_Sanov}, is quite flexible, and in Section~\ref{sec_classical_applications} we use it to deduce several useful corollaries that apply to different setting. See, for instance, Corollaries~\ref{both_composite_iid_or_av_cor} and~\ref{almost_iid_GSL_cor}.

\section{Preliminary considerations} \label{sec_preliminaries}

\subsection{Notation}

In what follows, we will denote as $\PP(\XX)$ the set of probability distributions on a given finite alphabet $\XX$, whose cardinality we will denote by $|\XX|$. The \deff{support} of some $P\in \PP(\XX)$ is defined as
\bb
\supp(P) \coloneqq \left\{ x\in \XX:\ P(x) > 0\right\} .
\ee
We will write $X\sim P$ to signify that a random variable $X$ is distributed according to the law $P\in \PP(\XX)$. The set of strings of symbols in $\XX$ of length $n\in \N^+$ will be denoted as $\XX^n$. If $X^n \coloneqq (X_1,\ldots, X_n)$ is the collection of $n$ independent and identically distributed (i.i.d.) random variables on $\XX$, and each $X_i$ follows the law $X_i\sim P$, we will also write that $X^n\sim P^{\otimes n}$. (For the i.i.d.\ extension of $P$, we prefer to use the notation $P^{\otimes n}$ instead of the more common $P^n$, so as to better highlight the difference with generic correlated distributions over $\XX^n$, which will be denoted as $P_n$, $Q_n$, etc.)

The \deff{entropy} of a probability distribution $P\in \PP(\XX)$ is defined by
\bb
H(P) \coloneqq - \sum_x P(x) \log P(x)\, ,
\label{entropy}
\ee
with the convention that $0\log 0 = 0$. The \deff{total variation distance} between two probability distributions $P, Q \in \PP(\XX)$ is defined as
\bb
\frac12 \|P - Q\|_1 \coloneqq \frac12 \sum_{x \in \XX} |P(x) - Q(x)|\, .
\label{TV_distance}
\ee
For two finite sets $\XX,\YY$, a \deff{channel} from $\XX$ to $\YY$ is a map $W:\PP(\XX)\to \PP(\YY)$ represented by a conditional probability distribution (that is, a stochastic matrix) $W(y|x)$.

An \deff{$n$-type} (or simply a type) over $\XX$ is a distribution $V \in \PP(\XX)$ such that $n V(x) \in \N$ for all $x \in \XX$~\cite{CSISZAR-KOERNER}. The set of all $n$-types is then given by
\bb
\TT_n \coloneqq \left\{ \left(\tfrac{k(x)}{n}\right)_{x \in \XX}\! :\  k(x) \in \N\ \ \forall x \in \XX,\ \sumno_{x \in \XX} k(x) = n \right\} .
\label{all_types}
\ee
A standard counting argument shows that
\bb
|\TT_n| = \binom{n + |\XX| - 1}{|\XX| - 1} \leq (n + 1)^{|\XX|}\, .
\label{counting_types}
\ee
The \deff{type of a string} $x^n\in \XX^n$ is the probability distribution $P_{x^n}\in \PP(\XX)$ defined by
\bb
P_{x^n}(x) \coloneqq \frac{N(x|x^n)}{n}\, ,\qquad \text{$N(x|x^n) \coloneqq$ number of times $x$ appears in $x^n$,}
\label{type_of_sequence}
\ee
for all $x\in \XX$. We denote as $T_{n,V}$ the \deff{type class} associated with a type $V\in \TT_n$, defined by
\bb
T_{n,V} \coloneqq \left\{ x^n \in \XX^n : P_{x^n} = V \right\} .
\label{type_class}
\ee
Type classes are invariant under permutations, and any string in $T_{n,V}$ can be obtained from any another by permuting symbols. Simple combinatorial considerations reveal that the cardinality of any $T_{n,V}$ can be calculated as
\bb
|T_{n,V}| = \frac{n!}{\prod_{x\in \XX} (nV(x))!}\, .
\label{cardinality_type_class}
\ee
It is often convenient to have handier estimates for~\eqref{cardinality_type_class}. A standard one is the following~\cite[Lemma~2.3]{CSISZAR-KOERNER}:
\bb
(n+1)^{-|\XX|} \exp\left[ n H(V) \right] \leq |T_{n,V}| \leq \exp\left[ n H(V) \right] ,
\label{cardinality_type_class_estimate}
\ee
where $H(V)$ is the entropy of $V$, as defined in~\eqref{entropy}.

\subsection{Relative entropies}

The most important of all relative entropies is the Kullback--Leibler divergence~\cite{Kullback-Leibler}, which we already encountered in~\eqref{KL}. In what follows, however, we will need also several related quantities. The first one is the \emph{max-relative entropy}, defined for any pair $P,Q\in \PP(\XX)$ as~\cite{Datta08}
\bb
D_{\max}(P\|Q) \coloneqq \inf\left\{ \lambda\in \R:\ P(x) \leq \exp[\lambda]\, Q(x)\ \ \forall\ x\in \XX \right\} .
\label{D_max}
\ee

\begin{note}
As is customary in information theory, we adopt a base-agnostic notation in which $\log$ and $\exp$ are the inverse functions of each other, but can be taken with respect to any base that is strictly larger than $1$.
\end{note}

It is elementary to show that
\bb
D(P\|Q) \leq D_{\max}(P\|Q)\, .
\label{D_vs_D_max}
\ee
In general, this inequality can be very loose. To try to tighten it, one can consider a variation of~\eqref{D_max} known as the \deff{smooth max-relative entropy}, defined, for $P,Q\in \PP(\XX)$ and $\e\in [0,1]$, by~\cite[Definition~3]{Datta-alias}
\bb
D_{\max}^\e(P\|Q) \coloneqq \inf_{P'\in \PP(\XX):\ \frac12\|P-P'\|_1\leq \e} D_{\max}(P'\|Q)\, .
\label{D_max_epsilon}
\ee

When Axiom~\ref{new_ax_filtering} is applicable, it is also useful to define the \deff{filtered relative entropy}. 
Here, `filtering' refers to the application of a channel $W$ with input alphabet $\XX$ (and arbitrary finite output alphabet). For $P,Q\in \PP(\XX)$, one defines
\bb
D^{W}(P\|Q) \coloneqq D\big(W(P)\,\big\|\, W(Q)\big)\, .
\label{general_filtered_relent}
\ee

\subsection{Hypothesis testing}

Following the discussion in Section~\ref{subsec_general_setting}, we now formalise the notation on hypothesis testing. Given two sets $\RR_1,\SS_1\subseteq \PP(\XX)$ representing the null and the alternative hypotheses, respectively,  the minimal type II error probability for a given threshold $\e\in (0,1)$ on the type I error probability can be defined as
\bb
\beta_\e(\RR_1\|\SS_1) \coloneqq \inf\left\{ \sup_{Q\in \SS_1} \sum_x A(x) Q(x):\ \ A:\XX\to [0,1],\ \ \sup_{P\in \RR_1} \sum_x \big(1-A(x)\big) P(x) \leq \e \right\} .
\label{beta_e}
\ee
The presence of the sets $\RR_1$ and $\SS_1$ inside the infimum makes this quantity slightly cumbersome to work with. We can remedy this problem by means of~\cite[Lemma~31]{Fang2025}, which shows that\footnote{The first three equalities in~\eqref{convexification_sets_beta_e} hold by inspection, because $\sup_{Q\in \SS_1} \sum_x A(x) Q(x) = \sup_{Q\in \co(\SS_1)} \sum_x A(x) Q(x)$ and $\sup_{P\in \RR_1} \sum_x \big(1-A(x)\big) P(x) = \sup_{P\in \co(\RR_1)} \sum_x \big(1-A(x)\big) P(x)$.}
\bb
- \log \beta_\e(\RR_1\|\SS_1) &= - \log \beta_\e(\co(\RR_1)\|\SS_1) \\
&= - \log \beta_\e(\RR_1\|\co(\SS_1)) \\
&= - \log \rel{\beta_\e}{\co(\RR_1)}{\co(\SS_1)} \\
&= \rel{D_H^\e}{\co(\RR_1)}{\co(\SS_1)}\, ,
\label{convexification_sets_beta_e}
\ee
where $\co$ denotes the convex hull, the rightmost side is defined according to the convention in~\eqref{divergence_sets}, and the \deff{hypothesis testing relative entropy} is given by~\cite{Buscemi2010}
\bb
D_H^\e(P\|Q) \coloneqq - \log \inf\left\{ \sumno_x A(x) Q(x):\ \ A:\XX\to [0,1],\ \ \sumno_x \big(1-A(x)\big) P(x) \leq \e \right\}
\label{D_H}
\ee
for all $P,Q\in \PP(\XX)$ and $\e\in (0,1)$. In particular, from~\eqref{Stein} and~\eqref{convexification_sets_beta_e} we deduce that
\bb
\stein(\RR\|\SS) &= \stein(\co(\RR)\|\SS) \\
&= \stein(\RR\|\co(\SS)) \\
&= \stein(\co(\RR)\|\co(\SS)) \\
&= \lim_{\e\to 0^+} \liminf_{n\to\infty} \frac1n\, \rel{D_H^\e}{\co(\RR_n)}{\co(\SS_n)}\, .\label{convexify_Stein}
\ee
where, with a slight abuse of notation, for a sequence $\FF = (\FF_n)_n$ of sets $\FF_n\subseteq \PP(\XX^n)$ we defined
\bb
\co(\FF) \coloneqq \big(\co(\FF_n)\big)_n\, .
\label{convexify_sequence_sets}
\ee
We record here the elementary but useful fact that, with the notation in~\eqref{F_n_av} and~\eqref{convexify_sequence_sets}, it holds that
\bb
\co\big(\FF_1^\mathrm{av}\big) = \co\big(\co(\FF_1)^\mathrm{av}\big)\, .
\label{convex_hull_goes_inside_av}
\ee

Perhaps surprisingly, the hypothesis testing relative entropy~\eqref{D_H} and the smooth max-relative entropy~\eqref{D_max_epsilon} are deeply related. The \deff{weak/strong converse duality}, first discovered in~\cite{Tomamichel2013, Anshu2019} and later refined in~\cite[Eq.~(59)]{tight-relations}, states that
\bb
D_{\max}^{1-\e}(P\|Q) + \log\frac1\e \leq D_H^\e(P\|Q) \leq D_{\max}^{1-\e-\mu}(P\|Q) + \log\frac1\mu
\label{weak_strong_converse_duality}
\ee
for all $P,Q\in \PP(\XX)$ and $0<\mu \leq 1-\e < 1$. Due to this fundamental relation, it is possible to use~\eqref{convexify_Stein} to express the Stein exponent in an alternative way, as previously observed many times, e.g.\ in~\cite[p.~24]{gap}. It is on this new expression that our entire approach to hypothesis testing hinges, and because of its importance we record it as an independent lemma. 

\begin{lemma} \label{Stein_expression_with_D_max_lemma}
For a finite alphabet $\XX$, let $\RR = (\RR_n)_n$ and $\SS = (\SS_n)_n$ be two sequences of hypotheses $\RR_n,\SS_n\subseteq \PP(\XX^n)$. Then, the corresponding Stein exponent, defined by~\eqref{Stein}, can be expressed as
\begin{align}
\stein(\RR\|\SS) &= \lim_{\e\to 1^-} \liminf_{n\to\infty} \frac1n\, \rel{D_{\max}^\e}{\co(\RR_n)}{\co(\SS_n)} \label{Stein_expression_with_D_max_1} \\
&= \inf_{\e\in (0,1)} \liminf_{n\to\infty} \frac1n\, \rel{D_{\max}^\e}{\co(\RR_n)}{\co(\SS_n)}\, . 
\label{Stein_expression_with_D_max_2}
\end{align}
\end{lemma}

\begin{proof}
For~\eqref{Stein_expression_with_D_max_1}, it suffices to plug~\eqref{weak_strong_converse_duality} into~\eqref{convexify_Stein} (setting, for example, $\mu = \e$, with $\e\in [0,1/2]$) and change variable $\e\mapsto 1-\e$. For~\eqref{Stein_expression_with_D_max_2}, we further observe that $\e\mapsto D_{\max}^\e(\co(\RR_n)\|\co(\SS_n))$ is a monotonically non-increasing function, as one sees by inspecting directly~\eqref{D_max_epsilon}.
\end{proof}

On a different note, a simple application of the data processing inequality under the action of the channel defined by an arbitrary test $A:\XX\to [0,1]$ as in~\eqref{D_H} shows that $D(P\|Q) \geq \rel{D_2}{\sumno_x A(x) P(x)}{\sumno_x A(x) Q(x)}$, where on the right-hand side we introduced the \deff{binary relative entropy}
\bb
D_2(p\|q) \coloneqq p \log \frac{p}{q} + (1-p) \log \frac{1-p}{1-q}\, .
\label{binary_relative_entropy}
\ee
Writing
\bb
D_2(p\|q) = -h_2(p) + p \log \frac{1}{q} + (1-p) \log \frac{1}{1-q} \geq -1 + p \log \frac{1}{q}\, ,
\ee
where
\bb
h_2(x) \coloneqq -x\log x - (1-x) \log(1-x)
\label{binary_entropy}
\ee
is the \deff{binary entropy}, and optimising over tests $A$ yields the handy inequality
\bb
D(P\|Q) \geq -1 + (1-\e) D_H^\e(P\|Q)\, .
\label{lower_bound_D_with_D_H}
\ee
This can be immediately used to establish a general converse bound on the Stein exponent. To this end, we need to introduce a further definition. For two sequences of sets $\RR_n,\SS_n\subseteq \PP(\XX^n)$, define their \deff{regularised relative entropy} as
\bb
D^\infty(\RR\|\SS) \coloneqq \liminf_{n\to\infty} \frac1n\, D(\RR_n\|\SS_n) = \liminf_{n\to\infty} \frac1n\, \inf_{P_n\in \RR_n,\ Q_n\in \SS_n} D(P_n\|Q_n)\, .
\label{regularised_relent}
\ee
Now, we have the following.

\begin{lemma} \label{converse_double_Stein_lemma}
For a finite alphabet $\XX$, let $\RR = (\RR_n)_n$ and $\SS = (\SS_n)_n$ be two sequences of hypotheses $\RR_n,\SS_n\subseteq \PP(\XX^n)$. Then, using the notation in~\eqref{convexify_sequence_sets} and~\eqref{regularised_relent}, we have
\bb
\stein(\RR\|\SS) \leq D^\infty(\co(\RR)\|\co(\SS))\, .
\label{converse_double_Stein}
\ee
\end{lemma}

\begin{proof}
It follows immediately by combining~\eqref{convexify_Stein} and~\eqref{lower_bound_D_with_D_H}. 
\end{proof}


\subsection{Asymptotic continuity}

Entropic functionals of random variables with finite range are typically continuous; moreover, they exhibit a strong form of uniform continuity known as `asymptotic continuity'. As the simplest example of this behaviour, consider the entropy. For an arbitrary $c\in (0,1]$, let us define the auxiliary function $F_c:[0,\infty) \to \R$ as
\bb
F_c(x) \coloneqq \left\{ \begin{array}{ll} x \log \tfrac1c + h_2(x) & \quad \text{if $x\leq \tfrac{1}{c+1}$,} \\[1.5ex] \log \left(1+\tfrac{1}{c}\right) & \quad \text{if $x>\tfrac{1}{c+1}$.} \end{array} \right.
\label{F_c}
\ee
For every fixed $c\in (0,1]$, $F_c$ is uniformly continuous on $[0,\infty)$; furthermore, $F_c(0)=0$. We list some elementary properties of this function in Appendix~\ref{app_proof_variational_F_lemma}; here, instead, we use it to state a useful continuity bound for the entropy, reported below. (A slightly more refined --- and in fact optimal --- version can be found in~\cite{Audenaert2007}.)

\begin{lemma}[{(Asymptotic continuity of the entropy~\cite{Fannes1973, Audenaert2007})}] \label{continuity_entropy_lemma}
Let $P,Q\in \PP(\XX)$ be two probability distributions on the finite alphabet $\XX$. If $\frac12 \|P-Q\|_1\leq \e\in [0,1]$, then
\bb
\big| H(P) - H(Q) \big| \leq F_{1/|\XX|}(\e)\, ,
\label{continuity_entropy}
\ee
where $F_{1/|\XX|}$ is defined by~\eqref{F_c}.
\end{lemma}

Asymptotic continuity is also a property of the relative entropy distance functional, provided that the set from which we are calculating the distance is somewhat `well behaved'. Here, `well behaved' may have many different technical meanings. The following result, essentially due to~\cite[Proposition~13]{continuity-via-integral}, deals with the case where the set obeys Axiom~\ref{new_ax_depolarising}. It differs from known results in the literature, such as the original one by Donald~\cite{Donald1999} and the subsequent generalisations and refinements by Christandl~\cite[Proposition~3.23]{MatthiasPhD} and Winter~\cite[Lemma~7]{tightuniform}, because it does not require convexity. With the convexity assumption, the filtered case has been essentially solved in~\cite[Proposition~3]{rel-ent-sq}, with improvements in~\cite[Theorem~11]{Schindler2023} and~\cite[Lemma~S12]{catboundent}.

\begin{lemma}[{(Asymptotic continuity of the relative entropy distance functional, without convexity~\cite[Proposition~13]{continuity-via-integral})}] \label{ac_relent_resource_lemma}
For a finite alphabet $\XX$, let $\FF=(\FF_n)_n$ a sequence of sets of probability distributions $\FF_n\subseteq \PP(\XX^n)$ that obeys Axiom~\ref{new_ax_depolarising} with respect to $R\in \FF_1$ and $c>0$. Then, for all $n\in \N^+$ and all $P_n,P'_n\in \PP(\XX^n)$ with $\supp(P_n) \subseteq \supp(R)^{n}$ and $\frac12 \|P_n-P'_n\|_1\leq \e$, it holds that
\bb
D(P_n\|\FF_n) \leq D(P'_n\|\FF_n) + n\e \log\tfrac1c + ng(\e) + h_2(\e)\, ,
\label{ac_relent_resource}
\ee
where $h_2$ is the binary entropy defined in~\eqref{binary_entropy}, and
\bb
g(x) \coloneqq (x+1)\log(x+1) - x\log x\, .
\label{g}
\ee
\end{lemma}

The proof is reported for completeness in Appendix~\ref{app_ac_relent_resource}.

\section{Proof of the main result} \label{sec_proof_main_result}

In this section we present the proofs of our main results, Theorem~\ref{double_Stein_thm} and the closely related Theorem~\ref{stronger_generalised_Sanov_thm}.

\subsection{A combinatorial detour}

In what follows, we will often employ the notion of \deff{Hamming distance} between two strings $x^n,y^n\in \XX^n$; this is defined as
\bb
d(x^n,y^n) \coloneqq \left| \left\{ i\in \{1,\ldots,n\}:\ x_i \neq y_i \right\} \right| .
\label{Hamming_distance}
\ee
A key technical tool in our analysis is Lemma~\ref{Alon_Spencer_lemma} below, which gives a relatively refined estimate of the size of Hamming distance neighbourhoods of sets in $\XX^n$ with large probability weight according to some i.i.d.\ probability distribution. We start by recalling the following well-known inequality:

\begin{lemma}[{(Azuma's inequality~\cite[Theorem~7.2.1]{ALON-SPENCER})}] \label{Azuma_inequality_lemma}
Let $Z_0,\ldots,Z_m$ be a martingale, with $|Z_{i+1}-Z_i|\leq 1$ for all $i=0,\ldots,m-1$. For all $\lambda\geq 0$,
\bb
\pr\left\{ Z_m > Z_0 + \lambda \sqrt{m}\right\} < e^{-\lambda^2/2}\, .
\label{Azuma_inequality}
\ee
\end{lemma}

We are now ready to establish a variation on~\cite[Theorem~7.5.3]{ALON-SPENCER}. Essentially, our goal is to show that subsets of $\XX^n$ that include a sizeable fraction of all the strings that are typical for some $P\in \PP(\XX)$ have `large' neighbourhoods with respect to the Hamming distance.

\begin{lemma} \label{Alon_Spencer_lemma}
Let $\XX$ be a finite alphabet, $P\in \PP(\XX)$ a probability distribution on $\XX$, $n\in \N^+$ a positive integer, and $\YY_n\subseteq \XX^n$ a set of strings of length $n$ over $\XX$. If 
\bb
P^{\otimes n}(\YY_n) \geq \e \in (0,1)\, ,
\ee
then, for all $\eta\in (0,1)$ and all
\bb
K \geq \sqrt{2 n \ln(1/\e)} + \sqrt{2n \ln(1/\eta)}\, ,
\ee
we have
\bb
P^{\otimes n}\left(B_d\left(\YY_n, K\right)\right) \geq 1-\eta\, ,
\ee
where
\bb
B_d\left(\YY_n, K\right) \coloneqq \left\{ x^n\in \XX^n:\ \min_{y^n\in \YY_n} d(x^n,y^n)\leq K \right\} ,
\label{B_d_definition}
\ee
and $d(x^n,y^n)$ is the Hamming distance~\eqref{Hamming_distance}.
\end{lemma}

\begin{proof}
The proof is very similar in spirit to that of~\cite[Theorem~7.5.3]{ALON-SPENCER}. We repeat the argument here in order to have a self-contained treatment.

For an arbitrary $x^n\in \XX^n$, set
\bb
\Delta(x^n) \coloneqq \min_{y^n\in \YY_n} d(x^n, y^n)\, .
\ee
Draw a random string $X^n\in \XX^n$ according to the i.i.d.\ probability distribution $P^{\otimes n}$. For $i=0,\ldots,n$, consider the non-negative random variables
\bb
Z_i \coloneqq F_i\big(X^i\big) \coloneqq \E_{X'_{i+1}\ldots X'_n\sim P^{\otimes (n-i)}}\, \Delta\big(X_1, \ldots, X_i, X'_{i+1},\ldots, X'_n\big)\, ,
\label{Alon_Spencer_lemma_proof_eq2}
\ee
which are obtained by exposing the first $i$ coordinates of $X^n$, grouped in the string $X^i \coloneqq (X_1,\ldots, X_i)$, and considering the others as random and drawn in an i.i.d.\ fashion from $P$. Note that each $Z_i$ can take on only finitely many (non-negative) values, $Z_0=\mu$ is a constant equal to the average distance of an i.i.d.\ string drawn from $P$ to $\YY_n$, and $Z_n = \Delta(X^n)$ is the actual distance of our initial (random) string from $\YY_n$. Furthermore, for all $i=0,\ldots, n-1$ and all collections $z^i \coloneqq (z_0, \ldots, z_i)$ of possible values of the variables $Z^i \coloneqq (Z_0, \ldots, Z_i)$ (so that necessarily $z_0=\mu$), a little thought reveals that
\bb
\E\left[ Z_{i+1}\big|\,Z^i=z^i\right] = z_i\, ,
\label{Alon_Spencer_lemma_proof_eq3}
\ee
entailing that $Z_0,\ldots, Z_n$ is a martingale. To verify~\eqref{Alon_Spencer_lemma_proof_eq3} rigorously, the simplest way is to consider the random variable $X^n|Z^i=z^i$, with probability distribution
\bb
P_{X^n|Z^i=z^i}(x^n) = P_{X^i|Z^i=z^i}(x^i) \prod_{j=i+1}^n P(x_j)\, .
\ee
Here, we observed that the last $n-i$ symbols of $X^n$ are independent of $Z^i$. We can now write
\begin{align}
\E\left[ Z_{i+1}\big|\,Z^i=z^i\right] &= \sum_{x^n} P_{X^n|Z^i=z^i}(x^n)\, F_{i+1}\big(x^{i+1}\big) \nonumber \\
&= \sum_{x^i} P_{X^i|Z^i=z^i}\big(x^i\big) \sum_{x_{i+1},\ldots, x_n} \left(\prodno_{j=i+1}^n P(x_j)\right) F_{i+1}\big(x^{i+1}\big) \nonumber \\
&= \sum_{x^i} P_{X^i|Z^i=z^i}\big(x^i\big) \sum_{x_{i+1}} P(x_{i+1})\, F_{i+1}\big(x^{i+1}\big) \\
&= \sum_{x^i} P_{X^i|Z^i=z^i}\big(x^i\big)\, F_i\big(x^i\big) \nonumber \\
&= z_i\, , \nonumber
\end{align}
where the equality on the second-to-last line holds because $\sum_{x_{i+1}} P(x_{i+1})\, F_{i+1}\big(x^{i+1}\big) = F_i\big(x^i\big)$ by construction (see~\eqref{Alon_Spencer_lemma_proof_eq2}), and that on the last line is a consequence of the fact that the only strings $x^i$ contributing to the sum are those for which $F_j\big(x^i\big) =z_j$ for all $j=0,\ldots,i$, and in particular they must satisfy $F_i\big(x^i\big) =z_i$. This establishes~\eqref{Alon_Spencer_lemma_proof_eq3}, proving that $Z_0,\ldots, Z_n$ is indeed a martingale.

Now, for all $i=0,\ldots,n-1$,
\bb
\left| Z_{i+1} - Z_i\right| &= \left| \E_{X'_{i+1}\ldots X'_n\sim P^{\otimes (n-i)}} \left( \Delta\big(X_1\ldots X_{i+1} X'_{i+2}\ldots X'_n\big) - \Delta(X_1\ldots X_i X'_{i+1}\ldots X'_n)\right) \right| \\
&\leq \E_{X'_{i+1}\ldots X'_n\sim P^{\otimes (n-i)}} \left| \Delta\big(X_1\ldots X_{i+1} X'_{i+2}\ldots X'_n\big) - \Delta(X_1\ldots X_i X'_{i+1}\ldots X'_n) \right| \\
&\leq 1\, ,
\ee 
simply because, by the triangle inequality, changing one symbol in a string can increase its Hamming distance from $\YY_n$ by at most $1$.

By Azuma's inequality (Lemma~\ref{Azuma_inequality_lemma}) applied to the martingales $Z_0,\ldots, Z_n$ and $-Z_0,\ldots,-Z_n$, for all $\lambda>0$ we have
\bb
\pr\left\{ Z_n < \mu - \lambda\sqrt{n} \right\} &< e^{-\lambda^2/2}\, , \\
\pr\left\{ Z_n > \mu + \lambda\sqrt{n} \right\} &< e^{-\lambda^2/2}\, .
\label{Alon_Spencer_lemma_proof_eq5}
\ee
For all $\lambda < \mu/\sqrt{n}$, the first inequality yields
\bb
\e \leq P^{\otimes n} (\YY_n) = \pr \left\{ Z_n = 0 \right\} = \pr \left\{ Z_n \leq 0 \right\} \leq \pr \left\{ Z_n < \mu - \lambda \sqrt{n} \right\} < e^{-\lambda^2/2}\, ,
\ee
entailing that $\e\leq e^{-\mu^2/(2n)}$, or, equivalently, $\mu \leq \sqrt{2n \ln(1/\e)}$, once one takes the limit $\lambda\to \big(\mu/\sqrt{n}\big)^-$. Then, from the second inequality in~\eqref{Alon_Spencer_lemma_proof_eq5} we obtain that
\bb
P^{\otimes n}\left(B_d\left(\YY_n, K \right)\right) &\geq
P^{\otimes n}\left(B_d\left(\YY_n, \sqrt{2 n \ln(1/\e)} + \sqrt{2 n \ln(1/\eta)} \right)\right) \\
&= \pr\left\{ \Delta(X^n) \leq \sqrt{2 n \ln(1/\e)} + \sqrt{2 n \ln(1/\eta)} \right\} \\
&= \pr\left\{ Z_n \leq \sqrt{2 n \ln(1/\e)} + \sqrt{2 n \ln(1/\eta)} \right\} \\
&\geq \pr\left\{ Z_n \leq \mu + \sqrt{2 n \ln(1/\eta)} \right\} \\
&= 1 - \pr\left\{ Z_n > \mu + \sqrt{2 n \ln(1/\eta)} \right\} \\
&\geq 1 - \eta\, ,
\ee
which concludes the proof.
\end{proof}

\subsection{Symbol-by-symbol blurring lemma}

In this section we will build our fundamental technical tool, Lemma~\ref{sbs_blurring_lemma} below. We start by proving two simple lemmas.

\begin{lemma} \label{transition_prob_lemma}
Let $R\in \PP(\XX)$ a probability distribution on a finite alphabet $\XX$. For some $\YY\subseteq \XX$, let $c\geq 0$ be such that $\min_{y \in \YY} R(y) \geq c$. Given two strings $x^n\in \XX^n$ and $y^n\in \YY^n$ at Hamming distance 
\bb
d(x^n,y^n) \leq ns\, ,
\label{transition_prob_assumption}
\ee
where $s\in \R$, and some $\delta\in \big(0,\tfrac{1}{c+1}\big]$, the probability that the map $\mm$ defined by~\eqref{depolarising} applied to every symbol turns $x^n$ into $y^n$ satisfies
\bb
\pr\big\{\mm^{\otimes n}: x^n\to y^n\big\} \geq (1-\delta)^n \left( \frac{c\delta}{1-\delta}\right)^{ns} .
\label{transition_prob}
\ee
\end{lemma}

\begin{proof}
Let $I \coloneqq \big\{i\in \{1,\ldots,n\}: x_i \neq y_i \big\}$, so that $|I| = d(x^n,y^n)$ by definition of Hamming distance, and $x_i=y_i$ for all $i\in I^c$. With the action of $\mm$, each symbol $x_i$ ($i= 1,\ldots, n$) has a probability $1-\delta$ of being left untouched, and a probability $\delta$ of being replaced with a symbol drawn according to $R$. Since $R(y)\geq c$ for all $y\in \YY$, such symbol coincides with $y_i$ with probability at least $c$. The events are independent, so the total probability can be estimated as
\bb
\pr\big\{\mm^{\otimes n}: x^n\to y^n\big\} &= \prod_{i=1}^n \pr\{\mm: x_i\to y_i\} \\
&= \left(\prodno_{i\in I} \pr\{\mm: x_i\to y_i\} \right) \left( \prodno_{i\in I^c} \pr\{\mm: x_i\to x_i\} \right) \\
&\geq (c\delta)^{|I|} (1-\delta)^{|I^c|} \\
&= (c\delta)^{d(x^n,y^n)} (1-\delta)^{n - d(x^n,y^n)} \\
&= (1-\delta)^n \left( \frac{c\delta}{1-\delta}\right)^{d(x^n,y^n)} \\
&\geq (1-\delta)^n \left( \frac{c\delta}{1-\delta}\right)^{ns} ,
\ee
where in the last line we used~\eqref{transition_prob_assumption} and observed that $\frac{c\delta}{1-\delta}\leq 1$. This concludes the proof.
\end{proof}

\begin{lemma} \label{nasty_estimate_lemma}
Let $x^n,y^n\in \XX^n$ be two strings of symbols taken from a finite alphabet $\XX$, assumed to be at a Hamming distance of at most $d(x^n,y^n) \leq ns$, for some $s\in \R$. Denote by $V_{x^n}, V_{y^n}\in \TT_n$ the types of $x^n, y^n$, respectively, and let $P\in \PP(\XX)$ be a probability distribution on $\XX$. Then
\bb
P^{\otimes n}(x^n) \leq \frac{(n+1)^{|\XX|} \exp\left[ n F_{1/|\XX|}(s)\right]}{\scaleobj{1.3}{|}T_{n,V_{y^n}}\scaleobj{1.3}{|}}\, ,
\label{nasty_estimate}
\ee
where $F_{1/|\XX|}$ is defined by~\eqref{F_c}.
\end{lemma}

\begin{proof}
We start by estimating the total variation distance between the types $V_{x^n}$ and $V_{y^n}$. A little thought reveals that
\bb
\frac12\,\big\|V_{x^n} - V_{y^n}\big\|_1 \leq s\, .
\label{nasty_estimate_proof_eq1}
\ee
To prove this formally, note that, for each $i=1,\ldots, n$,
\bb
\frac12 \sum_{x\in \XX} \left|\delta_{x,x_i} - \delta_{x,y_i}\right| = \left\{ \begin{array}{ll} 0 & \text{ if $x_i=y_i$,} \\[1ex] 1 & \text{ if $x_i\neq y_i$.} \end{array}\right.
\label{nasty_estimate_proof_eq2}
\ee
Here, $x_i$ is the $i^\text{th}$ symbol of $x^n$, and analogously for $y_i$; also, $\delta_{x,x'}$ is equal to $1$ if $x=x'$, and equal to $0$ otherwise. Summing~\eqref{nasty_estimate_proof_eq2} over all $i=1,\ldots, n$ yields
\bb
\frac12 \sum_{i=1}^n \sum_{x\in \XX} \left|\delta_{x,x_i} - \delta_{x,y_i}\right| = d(x^n,y^n)\, .
\ee
By the triangle inequality, the left-hand side can be lower bounded as
\bb
d(x^n,y^n) &\geq \frac12 \sum_{x\in \XX} \left|\sumno_{i=1}^n \left(\delta_{x,x_i} - \delta_{x,y_i}\right)\right| \\
&= \frac12 \sum_{x\in \XX} \left|N(x|x^n) - N(x|y^n) \right| \\
&= \frac{n}{2}\, \sum_{x\in \XX} \left|V_{x^n}(x) - V_{y^n}(x) \right| \\
&= \frac{n}{2}\,\big\|V_{x^n} - V_{y^n}\big\|_1\, ,
\ee
which proves~\eqref{nasty_estimate_proof_eq1} once one remembers that $d(x^n,y^n)\leq ns$ by assumption. We are now ready to write
\bb
\frac1n \log \left(P^{\otimes n}(x^n)\, \big|T_{n,V_{y^n}}\big|\right) &\leqt{(i)} \frac1n \log \frac{\big|T_{n,V_{y^n}}\big|}{\big|T_{n,V_{x^n}}\big|} \\
&\leqt{(ii)} H\big(V_{y^n}\big) - H\big(V_{x^n}\big) + \frac{|\XX| \log(n+1)}{n} \\
&\leqt{(iii)} F_{1/|\XX|}(s) + \frac{|\XX| \log(n+1)}{n}\, .
\ee
Here, in~(i) we observed that, due to permutational symmetry, $P^{\otimes n}(z^n)$ must be the same for all strings $z^n$ with the same type as $x^n$; since the total probability of the type class $T_{n,V_{x^n}}$ cannot exceed $1$, it follows that every single string can have probability at most equal to $1\big/\scaleobj{1.1}{|}T_{n,V_{x^n}}\scaleobj{1.1}{|}$. Continuing, the inequality~(ii) is deduced by applying~\eqref{cardinality_type_class_estimate} twice, while in~(iii) we employed Lemma~\ref{continuity_entropy_lemma} and the above estimate~\eqref{nasty_estimate_proof_eq1}. The claimed inequality~\eqref{nasty_estimate} is obtained via elementary algebraic manipulations.
\end{proof}

We are now ready to establish the following key technical result:

\begin{lemma}[(Symbol-by-symbol blurring lemma)] \label{sbs_blurring_lemma}
Let $P\in \PP(\XX)$ be a probability distribution on the finite alphabet $\XX$, and, for a positive integer $n\in \N^+\!$, let $Q_n\in \PP(\XX^n)$ be a (not necessarily permutationally symmetric) probability distribution on $n$ copies of $\XX$. For some $\lambda,\mu\geq 0$ and $\xi\in (0,1/3)$, assume that there exists a type $V\in \TT_n$ such that $\frac12 \|V-P\|_1\leq \xi$ and
\bb
\left|\left\{y^n\in T_{n,V}:\ Q_n(y^n) \geq \frac{\exp[-n\lambda]}{|T_{n,V}|} \right\}\right| \geq \exp[-n\mu]\, |T_{n,V}|\, ,
\ee
where $T_{n,V}$ is the type class with type $V$ (see~\eqref{type_class}). Then, picking some $R\in \PP(\XX)$ such that 
\bb
\min_{x\in \supp(P)} R(x) \geq c > 0\, ,
\label{sbs_blurring_0}
\ee
some $\eta\in(0,1)$, we have
\bb
\inf_{\delta\,\in\, \scaleobj{1.2}{(}0,\frac{1}{c+1}\scaleobj{1.2}{]}} \frac1n\, D_{\max}^\eta\big( P^{\otimes n} \,\big\|\, \mm^{\otimes n}(Q_n) \big) \leq \lambda + 2\,F_{\min\{c,\,1/|\XX|\}}\left(\sqrt{\tfrac{2\mu}{\log e}} + \theta_{|\XX|,\,\eta}(\xi,n)\right) + \widetilde{o}_{|\XX|,\,\eta}\big(\tfrac{1}{n}\big) \, ,
\label{sbs_blurring_1}
\ee
where we employed the auxiliary function given by~\eqref{F_c} and defined
\begin{align}
\theta_{|\XX|,\,\eta}(\xi,n) &\coloneqq \sqrt{4\xi \ln |\XX| + \tfrac{2}{\log e}\left( 3\xi \log\tfrac{|\XX|}{\xi} + h_2(3\xi) \right) + \tfrac{2|\XX| \ln(n+1)}{n}} + \sqrt{\tfrac2n \ln \tfrac1\eta} + 2\xi\, , \label{sbs_blurring_2} \\
\widetilde{o}_{|\XX|,\,\eta}\big(\tfrac{1}{n}\big) &\coloneqq \frac1n \left( |\XX| \log(n+1) + \log \tfrac{1}{1-\eta} \right) . \label{sbs_blurring_3}
\end{align}
\end{lemma}

\begin{rem} \label{hearling_fears_rem}
The explicit expressions of the functions in~\eqref{sbs_blurring_2}--\eqref{sbs_blurring_3} do not play a role in what follows, and are reported only for completeness. What \emph{will} play a role, instead, is the fact that
\bb
\lim_{\xi \to 0^+} \lim_{n \to\infty} \theta_{|\XX|,\,\eta}(\xi,n) = 0\, ,\qquad \lim_{n\to\infty} \widetilde{o}_{|\XX|,\,\eta}\big(\tfrac{1}{n}\big) = 0
\label{healing_fears}
\ee
for all fixed $|\XX|<\infty$ and all $\eta\in (0,1)$. Together with the continuity of $F_{c'}$ for any fixed $c'\in (0,1]$, this will immediately imply that the right-hand side of~\eqref{sbs_blurring_1} can be made arbitrarily close to $\lambda$ by taking $n$ large enough and $\xi$ and $\mu$ small enough.
\end{rem}

\begin{proof}
Define the set of strings
\bb
\YY_n \coloneqq \left\{y^n\in T_{n,V}:\ Q_n(y^n) \geq \frac{\exp[-n\lambda]}{|T_{n,V}|} \right\} ,
\label{sbs_blurring_proof_eq1}
\ee
so that 
\bb
|\YY_n|\geq \exp[-n\mu]\, |T_{n,V}|
\label{sbs_blurring_proof_eq2}
\ee
by assumption. We would like to apply Lemma~\ref{Alon_Spencer_lemma}. To this end, we need to obtain a lower bound on $P^{\otimes n}(\YY_n)$. Intuitively, this ought to be possible, because $P$ and $V$ are close in total variation distance, and $\YY_n$ is a subset of $T_{n,V}$ whose cardinality we just bounded from below. The problem with this line of reasoning, however, is that the type $V$ might assign some non-zero weight to symbols in $\XX$ outside of the support of $P$. The weight distributed in this way will be small, because $P$ and $V$ are close in total variation distance, but it can be non-zero. If this happens, then necessarily $P^{\otimes n}(T_{n,V}) = 0$, thwarting our attack on the problem right at the start.

To remedy this, we begin with a preliminary step that is designed to modify the set $\YY_n$ so as to eliminate, in every string, the symbols that are not in the support of $P$. More specifically, for some $\nu\in \big(0,\,\tfrac{1}{|\XX|}\big)$, to be fixed later, we can define
\bb
\XX_\nu \coloneqq \left\{ x\in \XX:\ P(x) \leq \nu \right\} .
\label{sbs_blurring_proof_preliminary_eq1}
\ee
Note that $\XX_\nu\neq \XX$, because $P$ must be normalised to $1$. Given any string $y^n = y_1\ldots y_n \in \YY_n$, we can replace every symbol $y_i\in \XX_\nu$, if any, with some fixed symbol $x_0\in \XX_\nu^c \coloneqq \XX\setminus \XX_\nu$. The symbols $y_j\in \XX_\nu^c$, instead, are left untouched. We denote the resulting string as $z^n(y^n)$.

How many symbols have been replaced in any given string $y^n\in \YY_n$? Since the type of $y^n$ is fixed and equal to $V$, it is not difficult to realise that this number does not in fact depend on $y^n$. To calculate it, it suffices to count how many symbols in a string with type $V$ belong to $\XX_\nu$: clearly, $nV(\XX_\nu)$. This number is small if $\nu$ and $\xi$ are small, because
\bb
\xi &\geq \frac12 \left\|V-P\right\|_1 \\
&= \max_{A\subseteq \XX} \left(V(A) - P(A)\right) \\
&\geq V(\XX_\nu) - P(\XX_\nu) \\
&\geq V(\XX_\nu) - \nu |\XX_\nu| \\
&\geq V(\XX_\nu) - \nu |\XX|\, .
\ee
Therefore, for all $y^n\in \YY_n$, we have
\bb
d\big(y^n\!, z^n(y^n)\big) = nV(\XX_\nu) \leq n\left(\xi + \nu |\XX|\right) ,
\label{sbs_blurring_proof_preliminary_eq3}
\ee
where $d$ is the Hamming distance (see~\eqref{Hamming_distance}). Let us now call $\ZZ_n$ the set obtained from $\YY_n$ by effecting the transformation $y^n\mapsto z^n(y^n)$ on every string $y^n\in \YY_n$; formally,
\bb
\ZZ_n \coloneqq \left\{ z^n(y^n):\ y^n\in \YY_n\right\} .
\ee
A little thought reveals that all strings in $\ZZ_n$ also have the same type: we can write
\bb
\ZZ_n \subseteq T_{n,\widebar{V}}\, , \quad \widebar{V} \coloneqq V\big|_{\XX_\nu^c} + V(\XX_\nu)\, E_{x_0}\, ,
\label{sbs_blurring_proof_preliminary_eq5}
\ee
where
\bb
V\big|_{\XX_\nu^c}(x) \coloneqq \left\{ \begin{array}{ll} V(x) & \quad \text{if $x\notin\XX_\nu$,} \\[.5ex] 0 & \quad \text{otherwise,} \end{array} \right.
\ee
and $E_{x_0}$ is the deterministic probability distribution concentrated on $x_0$, i.e.\ $E_{x_0}(x) = \delta_{x,x_0}$ for all $x\in \XX$.

We now have
\begin{align}
P^{\otimes n}(\ZZ_n) &= \sum_{z^n\in \ZZ_n} P^{\otimes n}(z^n) \nonumber \\
&\geqt{(i)} |\XX|^{-nV(\XX_\nu)} \sum_{y^n\in \YY_n} P^{\otimes n}\big(z^n(y^n)\big) \nonumber \\
&\geqt{(ii)} |\XX|^{-n\left(\xi + \nu |\XX|\right)}\, \frac{P^{\otimes n}\big(T_{n,\widebar{V}}\big)}{\big|T_{n,\widebar{V}}\big|}\, |\YY_n| \label{sbs_blurring_proof_preliminary_eq7} \\
&\geqt{(iii)} |\XX|^{-n\left(\xi + \nu |\XX|\right)} \exp[-n\mu]\, P^{\otimes n}\big(T_{n,\widebar{V}}\big) \nonumber \\
&\geqt{(iv)} (n+1)^{-|\XX|} |\XX|^{-n\left(\xi + \nu |\XX|\right)} \exp\left[-n\left(\mu + \rel{D}{\widebar{V}}{P}\right)\right] \nonumber \\
&\geqt{(v)} (n+1)^{-|\XX|} |\XX|^{-n\left(\xi + \nu |\XX|\right)} \exp\left[-n\left(\mu + (\nu |\XX| + 2\xi) \log\tfrac{1}{\nu} + h_2(\nu |\XX| + 2\xi) \right)\right] . \nonumber
\end{align}
We now present a detailed justification of the above derivation. 
\begin{enumerate}[(i), leftmargin=6ex, labelwidth=2ex, labelsep=1ex]
\item While the map $y^n\mapsto z^n(y^n)$ need not be injective in general, for all $z^n\in \ZZ_n$ we have
\bb
\left|\left\{ y^n:\ z^n(y^n)=z^n\right\}\right| \leq |\XX|^{n V(\XX_\nu)} ;
\label{sbs_blurring_proof_preliminary_eq8}
\ee
to see why, we ask ourselves: when do two strings $y^n,{y'}^n\in \YY_n$ satisfy $z^n(y^n) = z^n({y'}^n)$? Clearly, this happens if and only if, for all $i=1,\ldots,n$ such that $y_i\in \XX_\nu^c$, we have $y_i = y'_i$ --- indeed, these symbols will be left untouched by the transformation $y^n\mapsto z^n(y^n)$. There are exactly $nV(\XX_\nu)$ values of $i$ such that this condition is \emph{not} met, i.e.\ such that $y_i\in \XX_\nu$. Given $y^n$, a matching ${y'}^n$ can only differ by the symbols in these sites. Since there are at most $|\XX|^{nV(\XX_\nu)}$ ways to choose the symbols in $nV(\XX_\nu)$ sites, Eq.~\eqref{sbs_blurring_proof_preliminary_eq8} follows. Due to that identity, we see that the sum $\sum_{y^n\in \YY_n} P^{\otimes n}\big(z^n(y^n)\big)$ can contain every term $P^{\otimes n}\big(z^n\big)$, where $z^n\in \ZZ_n$, at most $|\XX|^{nV(\XX_\nu)}$ times. The inequality~(i) follows.

\item On the one hand we employed~\eqref{sbs_blurring_proof_preliminary_eq3}; on the other, we observed that, due to~\eqref{sbs_blurring_proof_preliminary_eq5}, all strings of the form $z^n(y^n)$ ($y^n\in \YY_n$) have the same type; hence, the value of $P^{\otimes n}\big(z^n(y^n)\big)$ does not depend on $y^n$. It thus holds that
\bb
P^{\otimes n}\big(z^n(y^n)\big) = \frac{P^{\otimes n}\big(T_{n,\widebar{V}}\big)}{\big|T_{n,\widebar{V}}\big|}\qquad \forall\ y^n\in \YY_n
\ee

\item Remembering~\eqref{sbs_blurring_proof_eq2}, here we are simply claiming that $\big|T_{n,\widebar{V}}\big| \leq |T_{n,V}|$; this is in fact quite obvious, and follows from the fact that the function $y^n\mapsto z^n(y^n)$, when extended to the whole domain $T_{n,V}$, is surjective on $T_{n,\widebar{V}}$. The same conclusion can be reached by calculating the cardinalities of both type classes with the help of the multinomial formula~\eqref{cardinality_type_class}.

\item This is an application of Sanov's theorem~\cite[Exercise~2.12, p.~29]{CSISZAR-KOERNER}.

\item Note that
\bb
\rel{D_{\max}}{\widebar{V}}{P} \leq \log\tfrac{1}{\nu}\, ,
\label{sbs_blurring_proof_preliminary_eq10}
\ee
simply because the support of $\widebar{V}$ is entirely contained in $\XX_\nu^c$, and $P(x)\geq \nu$ for all $x\in \XX_\nu^c$ by construction. Moreover,
\bb
\big\|\widebar{V} - P\big\|_1 &\leq \big\|\widebar{V} - V \big\|_1 + \|V-P\|_1 \\
&= \sum_{x\in \XX_\nu} V(x) + \big|\widebar{V}(x_0) - V(x_0)\big| + \|V-P\|_1 \\
&= 2V(\XX_\nu) + \|V-P\|_1 \\
&\leq 2\nu |\XX| + 4\xi\, ,
\label{sbs_blurring_proof_preliminary_eq11}
\ee
where the equalities follow from~\eqref{sbs_blurring_proof_preliminary_eq5}, while the last inequality is a consequence of~\eqref{sbs_blurring_proof_preliminary_eq3} together with the assumption that $\frac12 \|V-P\|_1\leq \xi$. 
As long as
\bb
\nu |\XX| + 2\xi \leq 1\, ,
\label{sbs_blurring_proof_preliminary_eq12}
\ee
Eq.~\eqref{sbs_blurring_proof_preliminary_eq10}--\eqref{sbs_blurring_proof_preliminary_eq11} allow us to employ the continuity estimate in~\cite[Eq.~(13)]{continuity-via-integral} to write
\bb
\rel{D}{\widebar{V}}{P} &\leq D(P\|P) + (\nu |\XX| + 2\xi) \log\tfrac1\nu + h_2(\nu |\XX| + 2\xi) \\
&= (\nu |\XX| + 2\xi) \log\tfrac1\nu + h_2(\nu |\XX| + 2\xi)\, ,
\ee
which is what we did in step~(v). This completes the justification of~\eqref{sbs_blurring_proof_preliminary_eq7}.
\end{enumerate}

Before proceeding, it is wise to simplify a bit the bound in~\eqref{sbs_blurring_proof_preliminary_eq7}. To this end, we can now fix
\bb
\nu \coloneqq \frac{\xi}{|\XX|}\, ,
\label{sbs_blurring_proof_preliminary_eq13}
\ee
which satisfies~\eqref{sbs_blurring_proof_preliminary_eq12}, due to fact that $\xi<1/3$, and lets us obtain
\bb
P^{\otimes n}(\ZZ_n) &\geq (n+1)^{-|\XX|} |\XX|^{-2n\xi} \exp\left[-n\left(\mu + 3\xi \log\tfrac{|\XX|}{\xi} + h_2(3\xi) \right)\right] \eqqcolon \e_n\, .
\label{sbs_blurring_proof_eq3}
\ee
Note that, using the definition in~\eqref{sbs_blurring_2}, we have
\bb
\frac1n\left(\sqrt{2 n \ln\tfrac{1}{\e_n}} + \sqrt{2 n \ln\tfrac1\eta}\right) &= \sqrt{\tfrac{2\mu}{\log e} + \left(\theta_{|\XX|,\,\eta}(\xi,n) - \sqrt{\tfrac2n \ln \tfrac1\eta} - 2\xi\right)^2} + \sqrt{2 n \ln\tfrac1\eta} \\
&\leq \sqrt{\tfrac{2\mu}{\log e}} + \theta_{|\XX|,\,\eta}(\xi,n) - 2\xi \\
&= s_n - 2\xi\, ,
\ee
where in the second line we observed that $\sqrt{A+B} \leq \sqrt{A} + \sqrt{B}$ for all $A,B\geq 0$, and in the last we defined
\bb
s_n \coloneqq \sqrt{\tfrac{2\mu}{\log e}} + \theta_{|\XX|,\,\eta}(\xi,n)\, .
\label{sbs_blurring_proof_eq4}
\ee

Due to Lemma~\ref{Alon_Spencer_lemma} applied with $\YY_n\mapsto \ZZ_n$, $\e\mapsto \e_n$, $K\mapsto n(s_n - 2\xi)$, Eq.~\eqref{sbs_blurring_proof_eq3} entails that
\bb
P^{\otimes n}\! \left(\widetilde{\ZZ}_n\right) \geq&\ 1-\eta\, , \\
\widetilde{\ZZ}_n \coloneqq&\ B_d\big(\ZZ_n,\, n(s_n - 2\xi)\big)\, .
\label{sbs_blurring_proof_eq5}
\ee
Moreover, because of the fact that the Hamming distance obeys the triangle inequality, Eq.~\eqref{sbs_blurring_proof_preliminary_eq3}, with the choice in~\eqref{sbs_blurring_proof_preliminary_eq13}, implies that
\bb
\widetilde{\ZZ}_n \subseteq \widetilde{\YY}_n \coloneqq B_d\left(\YY_n, ns_n \right) ,
\label{sbs_blurring_proof_eq6}
\ee
so that a fortiori
\bb
1-\eta' &\coloneqq P^{\otimes n}\! \left(\widetilde{\YY}_n\right) \geq 1-\eta > 0\, .
\label{sbs_blurring_proof_eq7}
\ee

Now, set
\bb
P'_n(x^n) \coloneqq \left\{ \begin{array}{ll} \frac{P^{\otimes n}(x^n)}{1-\eta'} & \quad \text{if $x^n\in \widetilde{\YY}_n$,} \\[1ex] 0 & \quad \text{otherwise.} \end{array} \right.
\label{sbs_blurring_proof_eq8}
\ee
Note that $P'_n$, unlike $P^{\otimes n}$, is not necessarily permutationally symmetric, because $\widetilde{\YY}_n$ is not necessarily closed under permutations. Nevertheless, a simple calculation reveals that
\bb
\frac12 \left\| P'_n - P^{\otimes n} \right\|_1 = \eta' \leq \eta\, ,
\label{sbs_blurring_proof_eq9}
\ee

We now consider an arbitrary string $x^n\in \widetilde{\YY}_n\cap \supp(P)^n$; in particular, by~\eqref{sbs_blurring_proof_eq6} there exists $y^n\in \YY_n$ satisfying 
\bb
d(x^n,y^n) \leq ns_n \, .
\label{sbs_blurring_proof_eq10}
\ee
For any $\delta\in \big(0,\tfrac{1}{c+1}\big]$, we then have
\bb
\left(\mm^{\otimes n}(Q_n) \right)(x^n)\ &\geq\ Q_n(y^n) \pr\big\{\mm^{\otimes n}: y^n\to x^n\big\} \\
&\geqt{(vi)}\ Q_n(y^n)\, (1-\delta)^n \left(\frac{c\delta}{1-\delta}\right)^{ns_n} \\
&\geqt{(vii)}\ \frac{\exp[-n\lambda]}{|T_{n,V}|}\, (1-\delta)^n \left(\frac{c\delta}{1-\delta}\right)^{ns_n} \\
&\geqt{(viii)}\ \frac{1-\eta}{(n+1)^{|\XX|}}\, \exp\left[-n\left( \lambda + F_{1/|\XX|}(s_n) \right)\right] (1-\delta)^n \left(\frac{c\delta}{1-\delta}\right)^{ns_n} P'_n(x^n)\, .
\label{sbs_blurring_proof_eq11}
\ee
The inequalities in the above derivation are justified as follows: 
\begin{enumerate}[(i), leftmargin=6ex, labelwidth=2ex, labelsep=1ex] \setcounter{enumi}{5}
\item We applied Lemma~\ref{transition_prob_lemma} with $x^n$ and $y^n$ exchanged, $\YY\mapsto \supp(P)$, and $s\mapsto s_n$. See~\eqref{sbs_blurring_proof_eq10} for the definition of $s_n$. We also remembered~\eqref{sbs_blurring_0} and used the fact that $x^n\in \supp(P)^n$.

\item Holds by definition of the set $\YY_n$ (see~\eqref{sbs_blurring_proof_eq1}).

\item Follows by observing that 
\bb
P'_n(x^n) = \frac{P^{\otimes n}(x^n)}{1-\eta'} \leq \frac{P^{\otimes n}(x^n)}{1-\eta} \leq \frac{(n+1)^{|\XX|} \exp\left[n\, F_{1/|\XX|}(s_n)\right]}{(1-\eta) |T_{n,V}|}\, ,
\ee
where the first inequality holds due to~\eqref{sbs_blurring_proof_eq5}, and in the second we applied Lemma~\ref{nasty_estimate_lemma} with $V_{y^n} \mapsto V$ and $s\mapsto s_n$.
\end{enumerate}

We have just established~\eqref{sbs_blurring_proof_eq11} in the case where $x^n\in \widetilde{\YY}_n\cap \supp (P)^n$. Yet, even if $x^n\notin \widetilde{\YY}_n \cap \supp(P)^n$, the inequality between the leftmost and the rightmost side of~\eqref{sbs_blurring_proof_eq11} still holds, simply because the latter vanishes (see~\eqref{sbs_blurring_proof_eq8}). We thus conclude that said inequality actually holds for all $x^n\in \XX^n$, implying that
\bb
&\inf_{\delta\,\in\, \scaleobj{1.2}{(}0,\frac{1}{c+1}\scaleobj{1.2}{]}} \frac1n\, D_{\max}^\eta\big( P^{\otimes n} \,\big\|\, \mm^{\otimes n}(Q_n) \big) \\
&\qquad \leqt{(ix)}\ \inf_{\delta\,\in\, \scaleobj{1.2}{(}0,\frac{1}{c+1}\scaleobj{1.2}{]}} \frac1n\, D_{\max}\big(P'_n\, \big\|\, \mm^{\otimes n}(Q_n)\big) \\
&\qquad \leqt{(x)}\ \inf_{\delta\,\in\, \scaleobj{1.2}{(}0,\frac{1}{c+1}\scaleobj{1.2}{]}} \frac1n \log \left[\frac{(n+1)^{|\XX|} \exp\left[n\left( \lambda + F_{1/|\XX|}(s_n) \right)\right]}{(1-\eta)(1-\delta)^n}\, \left(\frac{1-\delta}{c\delta}\right)^{ns_n} \right] \\
&\qquad \eqt{(xi)}\ \lambda + F_{1/|\XX|}(s_n) + \widetilde{o}_{|\XX|,\,\eta}\big(\tfrac{1}{n}\big) + \inf_{\delta\,\in\, \scaleobj{1.2}{(}0,\frac{1}{c+1}\scaleobj{1.2}{]}} \left\{\log\frac{1}{1-\delta} + s_n \log\left(\frac{1-\delta}{c\delta}\right) \right\} \\
&\qquad \eqt{(xii)}\ \lambda + F_{1/|\XX|}(s_n) + \widetilde{o}_{|\XX|,\,\eta}\big(\tfrac{1}{n}\big) + F_c(s_n) \\
&\qquad \leqt{(xiii)}\ \lambda + 2\, F_{\min\{c,\,1/|\XX|\}}(s_n) + \widetilde{o}_{|\XX|,\,\eta}\big(\tfrac{1}{n}\big)\, .
\label{sbs_blurring_proof_eq13}
\ee
To justify the above derivation, we can argue as follows:
(ix)~holds because of~\eqref{sbs_blurring_proof_eq9}, while in~(x) we used~\eqref{sbs_blurring_proof_eq11}. From now on, all that remains are elementary algebraic manipulations: in~(xi) we expanded the logarithm, using the notation in~\eqref{sbs_blurring_3}; the identity in~(xii) follows from the variational representation of the auxiliary function $F_{c'}$ provided in Lemma~\ref{variational_F_lemma}(c), and the inequality~(xiii) is an application of another elementary property of the same function, stated in Lemma~\ref{variational_F_lemma}(b).

Substituting~\eqref{sbs_blurring_proof_eq4} into~\eqref{sbs_blurring_proof_eq13} yields~\eqref{sbs_blurring_1}, thereby concluding the proof.
\end{proof}


\subsection{A meta-lemma}

The above Lemma~\ref{sbs_blurring_lemma} is a fairly technical statement that is best used sparingly. In fact, we will use it only \emph{once}, to prove the meta-lemma (Lemma~\ref{meta_lemma}), reported below for convenience:

\begin{manuallemma}{\ref{meta_lemma}}[(Meta-lemma)]
For a finite alphabet $\XX$, let $(\FF_n)_n$ be a sequence of sets $\FF_n\subseteq \PP(\XX^n)$ that obeys Axiom~\ref{new_ax_depolarising} with respect to a probability distribution $R\in \PP(\XX)$ and a constant $c$ such that $\min_{x\in \supp(R)} R(x) \geq c > 0$. Take two real-valued functions $o_L(n)$ and $o_R(n)$ with the property that $\lim_{n\to\infty} \frac{o_L(n)}{n} = \lim_{n\to\infty} \frac{o_R(n)}{n} = 0$. For any $\Delta>0$, we can find $N = N(\Delta,c,o_L,o_R,|\XX|) \in \N^+$ such that, for all integers $n\geq N$, the following holds: given some $Q_n\in \FF_n$, an $n$-type $V\in \TT_n$, $P\in \PP(\XX)$ with $\supp(P)\subseteq \supp(R)$ and $\frac12 \|V-P\|_1\leq \xi\in (0,1/3)$, and some $\lambda\geq 0$, if
\begin{equation}
\left|\left\{x^n\in T_{n,V}:\ Q_n(x^n) \geq \frac{\exp[-n\lambda - o_L(n)]}{|T_{n,V}|} \right\}\right| \geq \exp[- o_R(n)]\, |T_{n,V}|\, ,
\tag{\ref{meta_lemma_0}}
\end{equation}
then
\begin{equation}
\frac1n\, D\big(P^{\otimes n}\, \big\|\, \FF_n\big) \leq \lambda + \phi(\xi) + \Delta\, ,
\tag{\ref{meta_lemma_1}}
\end{equation}
where $\phi$ is a continuous function that depends only on $c$ and $|\XX|$ and vanishes at $0$.
\end{manuallemma}

\begin{rem}
In the proof below we will see that an explicit choice of $\phi$, for example, could be
\bb
\phi(\xi) &= 2\, F_{\min\{c,\,1/|\XX|\}}\!\left(\lim_{n\to\infty} \theta_{|\XX|,\,\eta}(\xi,n) \right) \\
&= 2\, F_{\min\{c,\,1/|\XX|\}}\!\left( \sqrt{4\xi \ln |\XX| + \tfrac{2}{\log e}\left( 3\xi \log\tfrac{|\XX|}{\xi} + h_2(3\xi) \right)} + 2\xi \right) ,
\label{explicit_choice_phi}
\ee
where $F_{c'}$ is defined in~\eqref{F_c} and $\theta_{|\XX|,\,\eta}(\xi,n)$ in~\eqref{sbs_blurring_2}. Note that, by continuity, one can set $\phi(0) \coloneqq \lim_{\xi\to 0^+} \phi(\xi) = 0$.
\end{rem}

To wrap our head around the above result, it is best to consider the simple case where $P=V$, so that $\xi=0$. The meta-lemma then encapsulates the somewhat intuitive fact that, if $\FF$ represents a `physically meaningful hypothesis', in that it obeys Axiom~\ref{new_ax_depolarising}, and some $Q_n\in \FF_n$ satisfies that $Q_n(x^n) \gtrsim \frac{\exp[-n\lambda]}{|T_{n,P}|}$ for a significant fraction of the strings $x^n$ with type $P$, then $\lambda \gtrsim \frac1n\, D(P^{\otimes n}\|\FF_n)$. Since, typically, whenever $P\notin \FF_1$ we have that $D(P^{\otimes n}\|\FF_n) \gtrsim \kappa n$ for some $\kappa > 0$ (this can be proved, for example, under Axiom~\ref{new_ax_type_stability}), we conclude that $\lambda > 0$ must hold whenever $P\notin \FF_1$: in other words, $Q_n(x^n)\, |T_{n,P}|$ must decay to zero exponentially fast. For an even more intuitive explanation, we refer the reader to the discussion after Lemma~\ref{meta_lemma_perm_symm}.

\begin{proof}[Proof of Lemma~\ref{meta_lemma}]
For any fixed $n$, if $o_L(n)$ and $o_R(n)$ are negative, we can always re-defined them to be zero, and the inequality~\eqref{meta_lemma_0} will be a fortiori obeyed. Therefore, from now on we will tacitly assume that $o_L(n),o_R(n)\geq 0$ for all $n$. Now, taking some $\eta>0$ to be specified later, we start by observing that
\bb
\frac1n\, D_{\max}^\eta\big( P^{\otimes n} \,\big\|\, \FF_n \big) &\leqt{(i)} \inf_{\delta\,\in\, \scaleobj{1.2}{(}0,\frac{1}{c+1}\scaleobj{1.2}{]}} \frac1n\, D_{\max}^\eta\big( P^{\otimes n} \,\big\|\, \mm^{\otimes n}(Q_n) \big) \\
&\leqt{(ii)} \lambda + \tfrac{o_L(n)}{n} + 2\,F_{\min\{c,\,1/|\XX|\}}\left(\sqrt{\tfrac{2o_R(n)}{n \log e}} + \theta_{|\XX|,\,\eta}(\xi,n)\right) + \widetilde{o}_{|\XX|,\,\eta}\big(\tfrac{1}{n}\big)\, ,
\label{meta_lemma_proof_eq1}
\ee
where~(i) holds because $\mm^{\otimes n}(Q_n)\in \FF_n$ due to Axiom~\ref{new_ax_depolarising}, and in~(ii) we employed the symbol-by-symbol blurring lemma (Lemma~\ref{sbs_blurring_lemma}) with the substitutions
\bb
\lambda \mapsto \lambda + \tfrac{o_L(n)}{n}\, ,\quad \mu \mapsto \tfrac{o_R(n)}{n}\, ,
\ee
and the notation is from~\eqref{sbs_blurring_2}--\eqref{sbs_blurring_3}. Note that by assumption
\bb
\min_{x\in \supp(P)} R(x) \geq \min_{x\in \supp(R)} R(x) \geq c > 0\, ,
\ee
meaning that condition~\eqref{sbs_blurring_0} is met.

We now fix $\eta>0$ small enough (depending on $\Delta$ and $c$) such that
\bb
\eta \log\tfrac1c + g(\eta) \leq \frac{\Delta}{3}\, ,
\label{meta_lemma_proof_eq4}
\ee
where $g$ is the function defined by~\eqref{g}. That this is possible, naturally, follows from the fact that $\lim_{\eta\to 0^+} \left(\eta \log\tfrac1c + g(\eta)\right) = 0$. 

Since the function $F_{c'}$ is uniformly continuous, from~\eqref{sbs_blurring_2}--\eqref{sbs_blurring_3} it is not difficult to see that we have
\bb
2\,F_{\min\{c,\,1/|\XX|\}}\left(\sqrt{\tfrac{2o_R(n)}{n \log e}} + \theta_{|\XX|,\,\eta}(\xi,n)\right) + \widetilde{o}_{|\XX|,\,\eta}\big(\tfrac{1}{n}\big) \tendsn{u} \phi(\xi)\, ,
\label{meta_lemma_proof_uniform_continuity_0}
\ee
uniformly for all $\xi \in (0,1/3)$. Here, $\phi$ is the function defined by~\eqref{explicit_choice_phi}.

The justification of~\eqref{meta_lemma_proof_uniform_continuity_0} requires some elaboration. First, due to the second identity in~\eqref{healing_fears}, for any $\e_0>0$ we have that $\left|\widetilde{o}_{|\XX|,\,\eta}\big(\tfrac{1}{n}\big)\right| \leq \frac{\e_0}{3}$ for all sufficiently large $n$ (depending only on $|\XX|$ and on $\eta$, which has been fixed as a function of $\Delta$ and $c$ alone). Secondly, since $F_{\min\{c,\,1/|\XX|\}}$ is uniformly continuous, we will also have 
\bb
\left| F_{\min\{c,\,1/|\XX|\}}(t) - F_{\min\{c,\,1/|\XX|\}}(t') \right|\leq \frac{\e_0}{6}
\label{meta_lemma_proof_uniform_continuity_1}
\ee
if we can guarantee that $|t-t'|\leq \e_1$, for some sufficiently small $\e_1$ (depending only on $\e_0$, $c$, and $|\XX|$). Thirdly, up to taking $n$ sufficiently large (depending only on $o_L$ and $o_R$), we can also make sure that $\left|\frac{o_L(n)}{n}\right| \leq \frac{\e_0}{3}$ and $\sqrt{\tfrac{2o_R(n)}{n \log e}} \leq \frac{\e_1}{2}$. Fourthly, inspect the explicit expression of $\theta_{|\XX|,\,\eta}(\xi,n)$ in~\eqref{sbs_blurring_2}, recalling: (a)~the aforementioned fact that $\eta$ is fixed, and (b)~the uniform continuity of the square root over the whole half-line $[0,\infty)$. Using~(a) and~(b), it is elementary to see that, for all sufficiently large $n$ (depending on $\Delta$, $c$, and $|\XX|$, but not on $\xi$), we have
\bb
\left| \theta_{|\XX|,\,\eta}(\xi,n) - \theta_{|\XX|,\,\eta}(\xi,\infty) \right| \leq \frac{\e_1}{2}\, ,
\ee
where $\theta_{|\XX|,\,\eta}(\xi,\infty) \coloneqq \lim_{m\to\infty} \theta_{|\XX|,\,\eta}(\xi,m)$. Hence,
\bb
\left| \sqrt{\tfrac{2o_R(n)}{n \log e}} + \theta_{|\XX|,\,\eta}(\xi,n) - \theta_{|\XX|,\,\eta}(\xi,\infty) \right| \leq \frac{\e_1}{2} + \frac{\e_1}{2} = \e_1\, ,
\ee
implying, via~\eqref{meta_lemma_proof_uniform_continuity_1}, that
\bb
&\left| 2\,F_{\min\{c,\,1/|\XX|\}}\left( \sqrt{\tfrac{2o_R(n)}{n \log e}} + \theta_{|\XX|,\,\eta}(\xi,n) \right) - \phi(\xi) \right| \\
&\qquad = 2 \left| F_{\min\{c,\,1/|\XX|\}}\left( \sqrt{\tfrac{2o_R(n)}{n \log e}} + \theta_{|\XX|,\,\eta}(\xi,n) \right) - F_{\min\{c,\,1/|\XX|\}}\left( \theta_{|\XX|,\,\eta}(\xi,\infty)\right) \right|\\
&\qquad \leq \frac{\e_0}{3}\, ;
\ee
putting all together, we have
\bb
&\left| \tfrac{o_L(n)}{n} + 2\,F_{\min\{c,\,1/|\XX|\}}\left(\sqrt{\tfrac{2o_R(n)}{n \log e}} + \theta_{|\XX|,\,\eta}(\xi,n)\right) + \widetilde{o}_{|\XX|,\,\eta}\big(\tfrac{1}{n}\big) - \phi(\xi) \right| \\
&\qquad \leq \left| \tfrac{o_L(n)}{n} \right| + \left|2\,F_{\min\{c,\,1/|\XX|\}}\left(\sqrt{\tfrac{2o_R(n)}{n \log e}} + \theta_{|\XX|,\,\eta}(\xi,n)\right) - \phi(\xi) \right| + \left|\widetilde{o}_{|\XX|,\,\eta}\big(\tfrac{1}{n}\big) \right| \\
&\qquad \leq \frac{\e_0}{3} + \frac{\e_0}{3} + \frac{\e_0}{3} \\
&\qquad = \e_0\, .
\ee

This completes the justification of~\eqref{meta_lemma_proof_uniform_continuity_0}, which in turn entails the existence of some $N = N(\Delta,c,o_L,o_R,|\XX|)$ such that
\bb
\frac1n\, D_{\max}^\eta\big( P^{\otimes n} \,\big\|\, \FF_n \big) \leq \lambda + \phi(\xi) + \frac{\Delta}{3}
\label{meta_lemma_proof_eq5}
\ee
for all $n\geq N$. For future use, up to increasing $N$ we can also make sure that
\bb
N \geq \frac{3}{\Delta}\, .
\label{meta_lemma_proof_eq6}
\ee

We now use the above bound on the smooth max-relative entropy distance from $\FF_n$ to constrain the \emph{standard} relative entropy distance from $\FF_n$. For all $n\geq N$ and all $P'_n\in \PP(\XX^n)$ with $\frac12 \left\|P'_n - P^{\otimes n}\right\|_1\leq \eta$, we have
\bb
D\big(P^{\otimes n}\, \big\|\, \FF_n\big) &\leqt{(iii)} D\big(P'_n\, \big\|\, \FF_n\big) + n \left(\eta \log\tfrac1c + g(\eta)\right) + h_2(\eta) \\
&\leqt{(iv)} D_{\max}\big(P'_n\, \big\|\, \FF_n\big) + n \left(\eta \log\tfrac1c + g(\eta)\right) + h_2(\eta) \\
&\leqt{(v)} D_{\max}\big(P'_n\, \big\|\, \FF_n\big) + \frac{2n}{3}\,\Delta\, .
\label{meta_lemma_proof_eq7}
\ee
Here, in~(iii) we used Lemma~\ref{ac_relent_resource_lemma}, which is applicable because Axiom~\ref{new_ax_depolarising} holds, with $P_n\mapsto P^{\otimes n}$ and $\e\mapsto \eta$; the inequality in~(iv), instead, follows from~\eqref{D_vs_D_max}, while in~(v) we used~\eqref{meta_lemma_proof_eq4} and observed that $h_2(\eta)\leq 1 \leq \frac{N\Delta}{3} \leq \frac{n\Delta}{3}$ due to~\eqref{meta_lemma_proof_eq6}. Minimising the rightmost side of~\eqref{meta_lemma_proof_eq7} over $P'_n$ shows that
\bb
D\big(P^{\otimes n}\, \big\|\, \FF_n\big) \leq D_{\max}^\eta \big(P^{\otimes n}\, \big\|\, \FF_n\big) + \frac{2n}{3}\,\Delta\, .
\label{meta_lemma_proof_eq8}
\ee
Combining~\eqref{meta_lemma_proof_eq5} and~\eqref{meta_lemma_proof_eq8} shows that
\bb
\frac1n\, D\big(P^{\otimes n}\, \big\|\, \FF_n\big) &\leq \lambda + \phi(\xi) + \Delta
\ee
holds for all $n\geq N$, thereby concluding the proof.
\end{proof}

Considering the special case of Lemma~\ref{meta_lemma} where $Q_n$ is permutationally symmetric and also $\lambda = 0$, we obtain the following simplified statement.

\begin{lemma}[(Meta-lemma, simplified form)] \label{meta_lemma_perm_symm}
For a finite alphabet $\XX$, let $(\FF_n)_n$ be a sequence of convex sets $\FF_n\subseteq \PP(\XX^n)$ that obeys Axioms~\ref{new_ax_depolarising} and~\ref{new_ax_closed_permutations}, the former with respect to a probability distribution $R\in \PP(\XX)$ and a constant $c$ such that $\min_{x\in \supp(R)} R(x) \geq c > 0$. For any $\Delta>0$, we can find $N = N(\Delta,c,|\XX|) \in \N^+$ such that, for all $n\geq N$, $Q_n\in \FF_n$, $V\in \TT_n$, and $P\in \PP(\XX)$ such that $\supp(P)\subseteq \supp(R)$ and $\frac12\|V-P\|_1\leq \xi\in (0,1/3)$, we have
\bb
Q_n(T_{n,V}) \leq \exp\left[ - D(P^{\otimes n} \| \FF_n) + n (\phi(\xi) + \Delta) \right] ,
\label{meta_lemma_perm_symm_bound}
\ee
where $\phi$ is a continuous function that depends only on $c$ and $|\XX|$ and vanishes at $0$. If $\FF$ obeys also Axiom~\ref{new_ax_tensor_products_n_and_m}, then we can furthermore write, again for $n\geq N$,
\bb
Q_n(T_{n,V}) \leq \exp\left[ - n \left(D^\infty(P \| \FF) - \phi(\xi) - \Delta\right) \right] .
\label{meta_lemma_perm_symm_bound_infty}
\ee
\end{lemma}

Before we delve into the proof, let us pause for a moment to appreciate the intuitive nature of the above result. To this end, we set as usual $P=V$, so that $\xi=0$. In short, Lemma~\ref{meta_lemma_perm_symm} states that any sequence of hypotheses $\FF = (\FF_n)_n$ that obeys some minimal assumptions, such as Axioms~\ref{new_ax_depolarising} and~\ref{new_ax_tensor_products_n_and_m}, must have the property that any $Q_n\in \FF_n$ assigns an exponentially suppressed weight to all type classes $T_{n,P}$ with $D^\infty(P\|\FF) > 0$. This will typically hold for all $P\notin \FF_1$, at least whenever Axiom~\ref{BP_ax_6} is obeyed. When that is the case, any $Q_n\in \FF_n$ will output strings that have, with high probability, approximately free type. Another more compact way of expressing the same concept is that $\FF$ should be approximately closed under the operation of taking types.

\begin{proof}[Proof of Lemma~\ref{meta_lemma_perm_symm}]
Start by observing that $Q_n(T_{n,V})$ is invariant under permutations of the random variables $X^n\sim Q_n$. Therefore, without affecting the value of $Q_n(T_{n,V})$, thanks to Axiom~\ref{new_ax_closed_permutations} and to the convexity of $\FF_n$, we can assume that $Q_n\in \FF_n$ is permutationally invariant. With this in mind, note that
\bb
Q_n(x^n) = \frac{Q_n(T_{n,V})}{|T_{n,V}|} \qquad \forall\ x^n\in T_{n,V}\, .
\ee
We can therefore apply Lemma~\ref{meta_lemma} with the substitutions
\bb
o_L,o_R\mapsto 0\, ,\qquad \lambda \mapsto -\frac1n \log Q_n(T_{n,V})\, ,
\ee
which lets us obtain the bound
\bb
\frac1n\, D(P^{\otimes n}\|\FF_n) \leq \lambda + \phi(\xi) + \Delta = -\frac1n \log Q_n(T_{n,V}) + \phi(\xi) + \Delta\, .
\ee
Massaging the above inequality yields~\eqref{meta_lemma_perm_symm_bound}. Finally, if Axiom~\ref{new_ax_tensor_products_n_and_m} then the sequence $n\mapsto \rel{D}{P^{\otimes n}}{\FF_n}$ is easily seen to be sub-additive, implying, via Fekete's lemma~\cite{Fekete1923}, that $D^\infty(P\|\FF) \leq \frac1n \rel{D}{P^{\otimes n}}{\FF_n}$ for all $n$. Plugging this inequality into~\eqref{meta_lemma_perm_symm_bound} gives~\eqref{meta_lemma_perm_symm_bound_infty}.
\end{proof}

\subsection{Verifying type stability (Axiom~\ref{new_ax_type_stability})}

As discussed, Axiom~\ref{new_ax_type_stability} might be rather impractical to verify directly. To facilitate this step, we have proposed Axiom~\ref{new_ax_filtering}, and mentioned that it can be used to check Axiom~\ref{new_ax_type_stability}. We now set out to explain why. The following key lemma is a slight rephrasing of a result due to Piani~\cite[Theorem~1]{Piani2009}.

\begin{lemma} \label{superadd_filtered_lemma}
For a finite alphabet $\XX$, let $(\FF_n)_n$ be a sequence of sets $\FF_n \subseteq \PP(\XX^n)$ that obeys Axioms~\ref{new_ax_filtering} and that is closed under the operation of discarding all but the last symbol, in the sense that for all $n\in \N^+$ and all $Q_n = Q_{X_1\ldots X_n}\in \FF_n$, we have $Q_{X_n}\in \FF_1$. Then, for all $n\in \N^+$ and all $P_1,\ldots, P_n\in \PP(\XX)$,
\bb
\rel{D}{P_1 \otimes \ldots \otimes P_n}{\FF_n} \geq \rel{D}{P_1 \otimes \ldots \otimes P_{n-1}}{\FF_{n-1}} + D^W(P_n\| \FF_1)\, ,
\label{superadd_filtered_1}
\ee
where $W:\XX\to \YY$ is the channel from Axiom~\ref{new_ax_filtering}. In particular, for any $P\in \PP(\XX)$, using the notation in~\eqref{general_filtered_relent} we have
\bb
D^W(P\|\FF_1) \leq \frac1n\, \rel{D}{P^{\otimes n}}{\FF_n}\, .
\label{superadd_filtered_2}
\ee
\end{lemma}

\begin{proof}
For any pair of random variables $X,Y$ and associated probability distributions $P_X$, $P_Y$ or $Q_{XY}$, an explicit calculation reveals that
\bb
\rel{D}{P_X\otimes P_Y}{Q_{XY}} 
&= D(P_Y\|Q_Y) + \sum_y P_Y(y)\, \rel{D}{P_X}{Q_{X|Y=y}}\, .
\label{superadd_filtered_proof_eq1}
\ee
Now, consider $n$ random variables $X_1,\ldots, X_n$ on $\XX$, for whose distribution we have the two hypotheses $P_{X_1} \otimes \ldots \otimes P_{X_n} = P_1 \otimes \ldots \otimes P_n$ or $Q_{X_1\ldots X_n} = Q_n \in \FF_n$. We can apply the channel $W$ to $X_n$, thus obtaining the random variable $Y_n$; the two joint probability distributions of $X_1,\ldots, X_{n-1}$ and $Y_n$ will be denoted by $P_{X_1} \otimes \ldots \otimes P_{X_{n-1}}\otimes P_{Y_n}$ and $Q_{X_1\ldots X_{n-1}Y_n}$, respectively. We can then write
\bb
\rel{D}{P_1 \otimes \ldots \otimes P_n}{Q_n} &= \rel{D}{P_{X_1} \otimes \ldots \otimes P_{X_n}}{Q_{X_1\ldots X_n}} \\
&\geqt{(i)} \rel{D}{P_{X_1} \otimes \ldots \otimes P_{X_{n-1}}\otimes P_{Y_n}}{Q_{X_1\ldots X_{n-1} Y_n}} \\[2pt]
&\eqt{(ii)} \rel{D}{P_{Y_n}}{Q_{Y_n}} + \sum_{y_n} P_{Y_n}(y_n)\, \rel{D}{P_{X_1} \otimes \ldots \otimes P_{X_{n-1}}}{Q_{X_1\ldots X_{n-1}\,|\,Y_n=y_n}} \\[-7pt]
&\geqt{(iii)} D^W(P_n\|\FF_1) + \rel{D}{P_1 \otimes \ldots \otimes P_{n-1}}{\FF_{n-1}}\, .
\ee
Here, (i)~follows from data processing, (ii)~comes from~\eqref{superadd_filtered_proof_eq1}, and in~(iii) we observed that on the one hand $Q_{Y_n} = W(Q_{X_n})$, and $Q_{X_n} \in \FF_1$ because $(\FF_n)_n$ is closed under the operation of discarding all symbols except the last one, while on the other $Q_{X_1\ldots X_{n-1}\,|\,Y_n=y_n}\in \FF_{n-1}$ for all $y_n\in \YY$ by Axiom~\ref{new_ax_filtering}, so that $\rel{D}{P_{X_1} \otimes \ldots \otimes P_{X_{n-1}}}{Q_{X_1\ldots X_{n-1}\,|\,Y_n=y_n}}\geq  \rel{D}{P_1 \otimes \ldots \otimes P_{n-1}}{\FF_{n-1}}$. This proves~\eqref{superadd_filtered_1}.

To derive also~\eqref{superadd_filtered_2}, we simply apply~\eqref{superadd_filtered_1} iteratively $n$ times, isolating all variables one by one, from the last to the first.
\end{proof}

\begin{prop} \label{verifying_type_stability_prop}
For a finite alphabet $\XX$, let $\FF = (\FF_n)_n$ be a sequence of convex sets $\FF_n \subseteq \PP(\XX^n)$ that obeys Axioms~\ref{new_ax_depolarising},~\ref{new_ax_closed_permutations}, and~\ref{new_ax_filtering}, and such that $\FF_1$ is topologically closed. Then $\FF$ also obeys Axiom~\ref{new_ax_type_stability}.
\end{prop}

\begin{proof}
We claim that convexity of $\FF_n$, closedness under permutations (Axiom~\ref{new_ax_closed_permutations}), and Axiom~\ref{new_ax_filtering} together imply that $\FF$ is closed under the operation of discarding all but the last symbol, in the sense of the statement of Lemma~\ref{superadd_filtered_lemma}. Indeed, if $Q_n = Q_{X_1\ldots X_n} \in \FF_n$, denoting with $Y_i = W(X_i)$ the variables induced by acting with the channel $W$ from Axiom~\ref{new_ax_filtering}, we have
\bb
Q_{X_n}(x) &= \sum_{y_1,\ldots,y_{n-1}} Q_{Y_1\ldots Y_{n-1}X_n}(y_1,\ldots,y_{n-1},x) \\
&= \sum_{y_1,\ldots,y_{n-1}} Q_{Y_1\ldots Y_{n-1}}(y_1,\ldots,y_{n-1})\, Q_{X_n \,|\, Y_1=y_1,\ldots,Y_{n-1}=y_{n-1}}(x)\, .
\label{verifying_type_stability_proof_eq1}
\ee
Each one of the probability distributions $Q_{X_n \,|\, Y_1=y_1,\ldots,Y_{n-1}=y_{n-1}}$ belongs to $\FF_1$, because they are obtained by conditioning on the values of the variables $Y_1,\ldots, Y_{n-1}$; by Axiom~\ref{new_ax_filtering}, conditioning on these variables, one by one, sends elements of $\FF_m$ to elements of $\FF_{m-1}$; note that the fact that we are conditioning on the \emph{first} $n-1$ variables rather than on the \emph{last} is immaterial, thanks to Axiom~\ref{new_ax_closed_permutations}. Now, due to convexity, Eq.~\eqref{verifying_type_stability_proof_eq1} entails that $Q_{X_n}\in \FF_1$, as claimed. We can now immediately apply Lemma~\ref{superadd_filtered_lemma}, which guarantees that 
\bb
D^W(P\|\FF_1) \leq \frac1n\, \rel{D}{P^{\otimes n}}{\FF_n}\qquad \forall\ P\in \PP(\XX)\, .
\label{verifying_type_stability_proof_eq2}
\ee

We now set out to verify Axiom~\ref{new_ax_type_stability}. Let $P\in \PP(\XX)$ and $K>0$ be such that, for all $\delta>0$,
\bb
\sup_{Q_n\in \FF_n} \pr_{X^n\sim Q_n}\!\big\{ \tfrac12 \|P_{X^n} - P \|_1 \leq \delta \big\} \geq \frac{1}{n^K} \qquad \forall\ n\in I\, ,
\label{verifying_type_stability_proof_eq3}
\ee
where $I\subseteq \N^+$ is infinite. Note that we can assume without loss of generality that $\delta$ is sufficiently small, e.g.\ that $\delta<1/3$. The first property of $P$ we record is that 
\bb
\supp(P)\subseteq \supp(R)\, ,
\label{verifying_type_stability_proof_eq4}
\ee
where $R\in \PP(\XX)$ is the probability distribution given by Axiom~\ref{new_ax_depolarising}. In fact, if this were not the case we could take some $x_0\notin \supp(R)$ and some $0 <\delta < P(x_0)$, and observe that any $x^n\in \XX^n$ produced by any $Q'_n\in \FF_n$ with non-zero probability would satisfy
\bb
\frac12 \|P_{x^n} - P \|_1 \geq P(x_0) - P_{x^n}(x_0) = P(x_0) > \delta\, ,
\label{verifying_type_stability_proof_eq5}
\ee
where the last equality holds because $x^n\in \supp(Q'_n) \subseteq \supp(R)^n$ by Axiom~\ref{new_ax_depolarising}, implying that $\supp(P_{x^n})\subseteq \supp(R)$, and so $P_{x^n}(x_0) = 0$, as $x_0\notin \supp(R)$ by construction. Due to~\eqref{verifying_type_stability_proof_eq5}, we would have
\bb
\pr_{X^n\sim Q'_n}\!\big\{ \tfrac12 \|P_{X^n} - P \|_1 \leq \delta \big\} = 0
\ee
for all $n\in \N^+$ and all $Q'_n\in \FF_n$, in contradiction with~\eqref{verifying_type_stability_proof_eq3}.

Now, for every $n\in I$ (see~\eqref{verifying_type_stability_proof_eq3}), pick some $Q_n\in \FF_n$ satisfying
\bb
\pr_{X^n\sim Q_n}\!\big\{ \tfrac12 \|P_{X^n} - P \|_1 \leq \delta \big\} \geq \frac{1}{2n^K}\, .
\ee
Note that the left-hand side is invariant under permutations of the variables. Since $\FF_n$ is convex and closed under permutations, we can assume without loss of generality that $Q_n$ is permutation invariant. Re-writing then yields
\bb
\frac{1}{2n^K} &\leq \pr_{X^n\sim Q_n}\!\big\{ \tfrac12 \|P_{X^n} - P \|_1 \leq \delta \big\} = \sumno_{V\in \TT_n:\ \frac12\|V-P\|_1\leq \delta} Q_n(T_{n,V})\, ,
\ee
implying that there exists some type $V_n\in \TT_n$ obeying $\frac12\|V_n - P\|_1\leq \delta$ and
\bb
Q_n(T_{n,V_n}) \geq \frac{1}{2n^K |\TT_n|} \geq \frac{1}{2n^K(n+1)^{|\XX|}}\, ,
\label{verifying_type_stability_proof_eq5}
\ee
where the last estimate comes from~\eqref{counting_types}. 

Due to~\eqref{verifying_type_stability_proof_eq4}, we are in position to apply Lemma~\ref{meta_lemma_perm_symm} with the substitutions $V \mapsto V_n$ and $\xi \mapsto \delta\in (0,1/3)$. We obtain that for all $\Delta>0$ the inequality
\bb
Q_n(T_{n,V_n}) \leq \exp\left[ - D(P^{\otimes n}\|\FF_n) + n(\Delta + \phi(\delta)) \right]
\ee
holds for all sufficiently large $n\in I$. Using also~\eqref{verifying_type_stability_proof_eq2} and~\eqref{verifying_type_stability_proof_eq5}, this yields
\bb
\frac{1}{2n^K(n+1)^{|\XX|}} \leq Q_n(T_{n,V_n}) \leq \exp\left[ - n \left( D^W(P\|\FF_1) - \Delta - \phi(\delta) \right) \right] ,
\ee
which again must hold for all sufficiently large $n\in I$. Since on the left-hand side we have an inverse polynomial and on the right-hand side an exponential function, taking the limit $n\to\infty$ along $n\in I$ gives us the inequality
\bb
D^W(P\|\FF_1) \leq \Delta + \phi(\delta)\, .
\ee
Now, remembering that $\Delta$ and $\delta$ can be taken to be arbitrarily small, we see that this is only possible if in fact $D^W(P\|\FF_1) \leq 0$. Together with the trivial inequality $D^W(P\|\FF_1)\geq 0$, this shows that in fact $D^W(P\|\FF_1) = 0$. Owing to the lower semi-continuity of the (filtered) relative entropy with respect to the second argument and to the fact that $\FF_1$ is closed (and hence compact) by assumption, this implies that $D^W(P\|P') = 0$ for some $P'\in \FF_1$. Due to the information completeness of $W$ guaranteed by Axiom~\ref{new_ax_filtering}, this can only hold if $P=P'$. This completes the proof.
\end{proof}

\subsection{Proof of the doubly composite Chernoff--Stein's lemma (Theorem~\ref{double_Stein_thm})} \label{subsec_proof_double_Stein}

Here we present the proof of our main result, restated below for convenience.

\begin{manualthm}{\ref{double_Stein_thm}}[(Doubly composite Chernoff--Stein lemma)]
Let $\XX$ be a finite alphabet, and let $\RR = (\RR_n)_n$ and $\SS = (\SS_n)_n$ be two families of sets of probability distributions $\RR_n,\SS_n\subseteq \PP(\XX^n)$, representing the null and the alternative hypotheses, respectively. Assume that:
\begin{enumerate}[(a)]
\item $\RR$ satisfies Axioms~\ref{new_ax_tensor_powers} and~\ref{new_ax_type_stability}; also, $\RR_1$ is topologically closed;
\item $\SS$ satisfies Axiom~\ref{new_ax_depolarising};
\item either $\RR$ satisfies Axiom~\ref{new_ax_permutational_symmetry}, or $\SS$ satisfies Axiom~\ref{new_ax_closed_permutations}.
\end{enumerate}
Then the Stein exponent, defined by~\eqref{Stein}, is given by
\begin{equation}
\stein(\RR\|\SS) = \inf_{P\in \RR_1} D^\infty(P\|\co(\SS)) = \inf_{P\in \RR_1} \liminf_{n\to\infty} \frac1n\, \rel{D}{P^{\otimes n}}{\co(\SS_n)}\, .
\tag{\ref{double_Stein}}
\end{equation}
In particular, Eq.~\eqref{double_Stein} holds under assumption~(b), if in addition
\begin{enumerate}[(a')]
\item $\RR$ satisfies Axioms~\ref{new_ax_depolarising},~\ref{new_ax_tensor_powers}, and~\ref{new_ax_filtering}, all sets $\RR_n$ are convex, and $\RR_1$ is topologically closed; and 
\setcounter{enumi}{2}
\item either $\RR$ satisfies Axiom~\ref{new_ax_permutational_symmetry}, or both $\RR$ and $\SS$ satisfy Axiom~\ref{new_ax_closed_permutations}.
\end{enumerate}
\end{manualthm}

Before we delve into the proof, it is instructive to examine a simple class of examples showing that the formula in~\eqref{double_Stein}, in general, does not single-letterise in an obvious way. The following construction is designed to mimic a famous quantum example, that of Werner states~\cite{Werner}, where we take as $\FF$ the classical representation of the set of `positive partial transpose' Werner states~\cite{Audenaert2001}.

\begin{ex} \label{additivity_violation_ex}
Let $\XX = \{0,1\}$, and consider the lexicographic ordering on $\{0,1\}^n$. For some $\gamma \geq 1$ and all $n\in \N^+$, set
\bb
\FF_{\gamma,n} \coloneqq \left\{ P_n \in \PP\big(\{0,1\}^n\big):\ H_\gamma^{\otimes n} P_n \geq 0 \right\} ,\qquad H_\gamma \coloneqq \begin{pmatrix} \gamma & 1 \\ -1 & 1 \end{pmatrix} .
\ee
Here, we thought of $P_n$ as a (column) vector in $\R^{2^n}$, and the above inequality between vectors is to be understood entry-wise. It is a simple exercise to verify that $\FF_\gamma = \big(\FF_{\gamma,n}\big)_n$ satisfies all Axioms~\ref{BP_ax_convex_closed}--\ref{BP_ax_permutations} (and hence also Axioms~\ref{new_ax_depolarising}--\ref{new_ax_closed_permutations}, by the forthcoming Lemma~\ref{BP_ax_1_5_imply_new_ax_1_3_lemma}) for all $\gamma\geq 1$, and even Axiom~\ref{new_ax_filtering} (and so also Axiom~\ref{new_ax_type_stability}, by Proposition~\ref{verifying_type_stability_prop}) as long as $\gamma > 1$.\footnote{For instance, to verify Axiom~\ref{BP_ax_partial_traces}, note that $(1,1) = V_\gamma H_\gamma$, where $V_\gamma \coloneqq \frac{1}{\gamma+1}\left(2,\gamma\!-\!1\right) \geq 0$. This means that to discard any single symbol out of the initial $n$, which corresponds to multiplying by the row vector $(1,1)$ from the right at the corresponding location in the tensor product, we can first apply $H_\gamma$ and then multiply by $V_\gamma\geq 0$. Applying $H_\gamma^{\otimes (n-1)}$ then necessarily results in a non-negative vector. To verify Axiom~\ref{new_ax_filtering}, one defines the channel $\{0,1\}\to \{0,1\}$ given by the stochastic matrix $W_\gamma \coloneqq \big(W_\gamma(x|y)\big)_{x,y} = \lsmatrix 1 & 1/\gamma \\ 0 & 1-1/\gamma \rsmatrix$. Clearly, this is an informationally complete channel if $\gamma > 1$. Now, the key observation is that $W_\gamma = T_\gamma H_\gamma$, where $T_\gamma \coloneqq \frac{1}{\gamma(\gamma+1)} \lsmatrix \gamma+1 & 0 \\ \gamma-1 & \gamma(\gamma-1) \rsmatrix$ is entry-wise positive.}

However, for $P = (1,0)^\intercal$ and $\gamma<3$ one sees that
\bb
D(P\|\FF_1) = \log 2 > \frac12 \log(\gamma+1) \geq \frac12\, \rel{D}{P^{\otimes 2}}{\FF_2} \geq D^\infty(P\|\FF)\, .
\ee
The first equality follows by observing that $\FF_1 \coloneqq \big\{(p,1-p)^\intercal:\ p\in [0,1/2]\big\}$, the second inequality can be derived by writing $D(P^{\otimes 2}\|\FF_2) \leq D(P^{\otimes 2}\|Q_2) = \log(\gamma+1)$, with the ansatz $Q_2 \coloneqq \frac{1}{\gamma+1} \left(1,0,0,\gamma\right)^\intercal$, and the last inequality holds as usual by Fekete's lemma~\cite{Fekete1923}. Hence, in general the Stein exponent in~\eqref{double_Stein} cannot be written as $D(\RR_1\|\SS_1)$, even for a simple i.i.d.\ null hypothesis and under a much stronger set of axioms.
\end{ex}

\begin{proof}[Proof of Theorem~\ref{double_Stein_thm}]
We start by showing that~(a') and~(c') together imply~(a) and~(c). In fact, (c')~implies~(c) directly. Also, due to the fact that Axiom~\ref{new_ax_permutational_symmetry} is strictly stronger than Axiom~\ref{new_ax_closed_permutations}, if~(c') holds then necessarily $\RR$ satisfies Axiom~\ref{new_ax_closed_permutations}. With~(a'), we then have that $\RR$ satisfies  Axioms~\ref{new_ax_depolarising},~\ref{new_ax_tensor_powers},~\ref{new_ax_closed_permutations}, and~\ref{new_ax_filtering}, all sets $\RR_n$ are convex, and $\RR_1$ is also topologically closed. All assumptions of Proposition~\ref{verifying_type_stability_prop} are therefore met, implying that $\RR$ also obeys the type stability axiom (Axiom~\ref{new_ax_type_stability}). This completes the requirements needed for~(a). In what follows, we can therefore assume without loss of generality that $\RR$ and $\SS$ satisfy~(a),~(b), and~(c).

The converse statement in the main claim~\eqref{double_Stein} follows from the general bound in Lemma~\ref{converse_double_Stein_lemma}, once one observes that 
\bb
D^\infty(\co(\RR)\|\co(\SS)) \leq \inf_{P\in \RR_1} D^\infty(P\|\co(\SS))\, , 
\ee
as $P^{\otimes n}\in \RR_n \subseteq \co(\RR_n)$ for all $P\in \RR_1$ due to Axiom~\ref{new_ax_tensor_powers}. 

We now move on to achievability. In what follows, we will denote as $R\in \SS_1$ the probability distribution whose existence is guaranteed by Axiom~\ref{new_ax_depolarising} for $\SS$. The same axiom guarantees also that
\bb
\supp(Q_n) \subseteq \supp(R)^n,\qquad \forall\ n\in \N^+,\quad \forall\ Q_n\in \SS_n\, .
\label{double_Stein_proof_eq0}
\ee
We will also call $c$ the constant from Axiom~\ref{new_ax_depolarising}, so that $\min_{x\in \supp(R)} R(x) \geq c >0$. 

We start from the expression of the Stein exponent in terms of the regularised smooth max-relative entropy presented in Lemma~\ref{Stein_expression_with_D_max_lemma}, and precisely in~\eqref{Stein_expression_with_D_max_2}, proceeding by contradiction. Assume that there exists some $\e\in (0,1)$ and some real $\lambda > 0$ such that 
\bb
\liminf_{n\to\infty} \frac1n\, \rel{D_{\max}^\e}{\co(\RR_n)}{\co(\SS_n)} < \lambda < \inf_{P'\in \RR_1} D^\infty(P'\|\co(\SS))\, .
\label{double_Stein_proof_eq1}
\ee
This entails that there exists an infinite subset $I\subseteq \N$ such that for all $n\in I$ we can find 
\bb
P_n\in \co(\RR_n)\, ,\quad P'_n\in \PP(\XX^n)\,,\quad Q_n\in \co(\SS_n)\, ,
\label{double_Stein_proof_eq1_5}
\ee
such that  
\bb
\frac12 \left\|P_n - P'_n\right\|_1\leq \e\, ,\qquad P'_n \leq \exp[n\lambda]\, Q_n\, .
\label{double_Stein_proof_eq2}
\ee

We are now presented with two cases, according to which alternative holds in condition~(c) of the statement. We start by assuming that $\RR$ obeys Axiom~\ref{new_ax_permutational_symmetry}. Then, from the first inequality in~\eqref{double_Stein_proof_eq2} we see that
\bb
1-\e &\leq 1 - \frac12 \left\|P_n - P'_n\right\|_1 \\
&= \sum_{x^n} \min\left\{ P_n(x^n),\, P'_n(x^n)\right\} \\
&= \sum_{V\in \TT_n} \sum_{x^n\in T_{n,V}} \min\left\{ P_n(x^n),\, P'_n(x^n)\right\} \\
&\eqt{(i)} \sum_{V\in \TT_n} \sum_{x^n\in T_{n,V}} \min\left\{ \frac{P_n(T_{n,V})}{|T_{n,V}|},\, P'_n(x^n)\right\} \\
&\eqt{(ii)} \sum_{\substack{\\[1pt] V\in \TT_n:\\[1pt] \supp(V) \,\subseteq\, \supp(R)}} \sum_{x^n\in T_{n,V}} \min\left\{ \frac{P_n(T_{n,V})}{|T_{n,V}|},\, P'_n(x^n)\right\} \, ,
\label{double_Stein_proof_eq3}
\ee
where in~(i) we leveraged the fact that $P_n$ is necessarily permutationally symmetric (Eq.~\eqref{double_Stein_proof_eq1_5} together with Axiom~\ref{new_ax_permutational_symmetry} for $\RR$), while in~(ii) we noticed that only types $V$ such that $\supp(V) \subseteq \supp(R)$ contribute to the sum. In fact, if $\supp(V)\not\subseteq \supp(R)$, then $T_{n,V}\cap \supp(R)^n = \emptyset$, entailing, via~\eqref{double_Stein_proof_eq0}, that $Q_n(x^n) = 0$ for all $x^n\in T_{n,V}$; due to~\eqref{double_Stein_proof_eq2}, we thus have $P'_n(x^n) = 0$ for all $x^n\in T_{n,V}$, implying that the term of the outer sum corresponding to $V$ vanishes.

From~\eqref{double_Stein_proof_eq3} we infer that for all $n\in I$ there must exist a type $V_n\in \TT_n$ such that
\bb
\supp(V_n)\subseteq \supp(R)
\label{double_Stein_proof_eq4}
\ee
and
\bb
\sum_{x^n\in T_{n,V_n}} \min\left\{ \frac{P_n(T_{n,V_n})}{|T_{n,V_n}|},\, P'_n(x^n)\right\} \geq \frac{1-\e}{|\TT_n|}\, \geqt{(iii)}\, \frac{1-\e}{(n+1)^{|\XX|}}\, ,
\label{double_Stein_proof_eq5}
\ee
where~(iii) is from~\eqref{counting_types}. Neglecting the second terms in the above minimisation, we also obtain that
\bb
P_n(T_{n,V_n}) \geq \frac{1-\e}{(n+1)^{|\XX|}}\, .
\label{double_Stein_proof_eq6}
\ee

Since $\PP(\XX)$ is a compact set due to the finiteness of $\XX$, from the sequence $(V_n)_{n\in I}$ we can extract a subsequence $(V_n)_{n\in J}$, with $J\subseteq I$ infinite, such that
\bb
V_n \tends{}{\ n\in J\ } P \in \PP(\XX)\, ,\qquad \supp(P)\subseteq \supp(R)\, ,
\label{double_Stein_proof_eq7}
\ee
where the support inclusion relation is a consequence of~\eqref{double_Stein_proof_eq4}. For any $\delta>0$ and for all sufficiently large $n\in J$ (depending on $\delta$) we thus have
\bb
\sup_{\widetilde{P}_n\,\in\, \RR_n} \pr_{X^n\sim \widetilde{P}_n} \left\{ \tfrac12 \|P_{X^n} - P\|_1 \leq \delta \right\} &\eqt{(iv)} \sup_{\widetilde{P}_n\,\in\, \co(\RR_n)} \pr_{X^n\sim \widetilde{P}_n} \left\{ \tfrac12 \|P_{X^n} - P\|_1 \leq \delta \right\} \\
&\geqt{(v)} \pr_{X^n\sim P_n} \left\{ \tfrac12 \|P_{X^n} - P\|_1 \leq \delta \right\} \\
&\geqt{(vi)} P_n(T_{n,V_n}) \\
&\geqt{(vii)} \frac{1-\e}{(n+1)^{|\XX|}}\, .
\label{double_Stein_proof_eq8}
\ee
Here, (iv)~holds by linearity and (v)~due to~\eqref{double_Stein_proof_eq1_5}; in~(vi) we assumed that $n\in J$ is large enough so that $\frac12 \|V_n - P\|_1 \leq \delta$, while in~(vii) we employed~\eqref{double_Stein_proof_eq6}. We are now in a position to apply Axiom~\ref{new_ax_type_stability} for $\RR$, which guarantees that~\eqref{double_Stein_proof_eq7} can hold for infinitely many values of $n$ for each $\delta>0$ only if 
\bb
P\in \RR_1\, .
\label{double_Stein_proof_eq9}
\ee

So far we have analysed only the $\RR$ side of things. It is now time to bring in $\SS$, i.e.\ the alternative hypothesis. We start by going back to~\eqref{double_Stein_proof_eq5}, this time without simplifying away the term containing $P'_n(x^n)$. Setting
\bb
\YY_n \coloneqq \left\{ x^n\in T_{n,V_n}:\ P'_n(x^n) \geq \frac{1-\e}{2(n+1)^{|\XX|} |T_{n,V_n}|} \right\} ,
\label{double_Stein_proof_eq10}
\ee
Eq.~\eqref{double_Stein_proof_eq5} immediately implies that
\bb
\frac{1-\e}{(n+1)^{|\XX|}} &\leq \sum_{x^n\in T_{n,V_n}} \min\left\{ \frac{P_n(T_{n,V_n})}{|T_{n,V_n}|},\, P'_n(x^n)\right\} \\
&\leqt{(viii)}\ |\YY_n|\cdot \frac{1}{|T_{n,V_n}|} + \left( |T_{n,V_n}| - |\YY_n| \right) \cdot \frac{1-\e}{2(n+1)^{|\XX|} |T_{n,V_n}|} \\
&\leq\ \frac{|\YY_n|}{|T_{n,V_n}|} + \frac{1-\e}{2(n+1)^{|\XX|}}\, ,
\ee
where in~(viii) we partitioned the sum into two partial sums, comprising the terms where $x^n\in \YY_n$ and $x^n\notin \YY_n$, respectively. Therefore,
\bb
|\YY_n| \geq \frac{1-\e}{2(n+1)^{|\XX|}}\, |T_{n,V_n}|\, .
\label{double_Stein_proof_eq12}
\ee
Now, pick some small $\xi\in (0,1/3)$; from~\eqref{double_Stein_proof_eq7}, we infer that
\bb
\frac12 \left\| V_n - P \right\|_1 \leq \xi
\label{double_Stein_proof_eq13}
\ee
for all large enough $n\in J$. Remembering~\eqref{double_Stein_proof_eq2} and~\eqref{double_Stein_proof_eq10}, we see that
\bb
Q_n(y^n) \geq \exp[-n\lambda]\, P'_n(y^n) \geq \frac{(1-\e) \exp[-n\lambda]}{2(n+1)^{|\XX|} |T_{n,V_n}|} \qquad \forall\ y^n\in \YY_n\, .
\label{double_Stein_proof_eq14}
\ee
We can now apply our meta-lemma. To this end, we effect the following substitutions in the statement of Lemma~\ref{meta_lemma}:
\bb
\FF_n \mapsto \co(\SS_n)\, ,\quad V\mapsto V_n\, ,\quad o_L(n),\, o_R(n)\mapsto \log\tfrac{2(n+1)^{|\XX|}}{1-\e}\, ;
\ee
note that $\co(\SS) = \big(\co(\SS_n)\big)_n$ satisfies Axiom~\ref{new_ax_depolarising} because $\SS$ does. Also,
\bb
\left|\left\{x^n\in T_{n,V_n}:\ Q_n(x^n) \geq \frac{\exp[-n\lambda - o_L(n)]}{|T_{n,V_n}|} \right\}\right|\ \, &\geqt{(ix)}\ \, \left|\left\{x^n\in T_{n,V_n}:\ P'_n(x^n) \geq \frac{\exp[-o_L(n)]}{|T_{n,V_n}|} \right\}\right| \\ 
&\eqt{(x)}\ \, |\YY_n| \\
&\geqt{(xi)}\ \, \exp[- o_R(n)]\, |T_{n,V_n}|\, ,
\label{double_Stein_proof_eq16}
\ee
where~(ix) holds because the set on the right-hand side is included in that on the left-hand side, due to~\eqref{double_Stein_proof_eq2}, in~(x) we remembered~\eqref{double_Stein_proof_eq10}, and in~(xi) we employed~\eqref{double_Stein_proof_eq12}. We are thus truly in a position to apply Lemma~\ref{meta_lemma}: for all $\Delta>0$, we obtain that
\bb
\frac1n\, D\big(P^{\otimes n}\,\big\|\, \co(\SS_n)\big) \leq \lambda + \phi(\xi) + \Delta
\label{double_Stein_proof_eq17}
\ee
for all sufficiently large $n\in J$ (depending on $\Delta$, $\e$, $c$, and $|\XX|$), i.e.
\bb
\limsup_{n\in J} \frac1n\, D\big(P^{\otimes n}\,\big\|\, \co(\SS_n)\big) \leq \lambda + \phi(\xi)\, .
\label{double_Stein_proof_eq18}
\ee
In~\eqref{double_Stein_proof_eq17}--\eqref{double_Stein_proof_eq18}, $\phi$ is the function whose existence is predicted by Lemma~\ref{meta_lemma}. (An explicit choice is available in~\eqref{explicit_choice_phi}.) Since $\xi\in (0,1/3)$ was arbitrary (and $J$ is independent of $\xi$), we can now take the limit $\xi\to 0^+$, obtaining that
\bb
\limsup_{n\in J} \frac1n\, D\big(P^{\otimes n}\,\big\|\, \co(\SS_n)\big)\, &\leqt{(xii)}\, \lambda + \lim_{\xi\to 0^+} \phi(\xi) = \lambda\, ,
\label{double_Stein_proof_eq19}
\ee
where (xii)~holds because $\phi$ is continuous, with $\phi(0)=0$.

Therefore,
\bb
\inf_{P' \in \RR_1} D^\infty(P'\|\co(\SS))\ &=\, \inf_{P' \in \RR_1} \liminf_{n\to\infty} \frac1n\, \rel{D}{{P'}^{\hspace{.7pt}\otimes n}}{\co(\SS_n)} \\
&\leqt{(xiii)}\, \liminf_{n\in J} \frac1n\, \rel{D}{P^{\otimes n}}{\co(\SS_n)} \\
&\leqt{(xiv)}\, \lambda\, .
\label{double_Stein_proof_eq20}
\ee
Here, in~(xiii) we used the ansatz $P'=P$ and restricted $n$ to the subsequence $J$, while~(xiv) holds because of~\eqref{double_Stein_proof_eq19}. Eq.~\eqref{double_Stein_proof_eq20} is in contradiction with~\eqref{double_Stein_proof_eq1}, and this concludes the proof in the case where $\RR$ obeys Axiom~\ref{new_ax_permutational_symmetry} in condition~(c).

If, instead, in~(c) we only assume that $\SS$ obeys Axiom~\ref{new_ax_closed_permutations}, we can run more or less the same argument, with relatively minor modifications. Most importantly, in~\eqref{double_Stein_proof_eq1_5} and~\eqref{double_Stein_proof_eq2} we can symmetrise $P'_n$ and $Q_n$, obtaining new distributions $\widebar{P}'_n \coloneqq \E_\pi \left(P'_n\circ \pi\right)$ and $\widebar{Q}_n \coloneqq \E_\pi \left( Q_n\circ \pi\right)$, where $\pi$ is a uniformly random permutation of a string of $n$ symbols; we again have $\widebar{P}'_n \leq \exp[n\lambda]\, \widebar{Q}_n$ and moreover $\widebar{Q}_n\in \co(\SS_n)$, due to Axiom~\ref{new_ax_closed_permutations} for $\SS$; defining also $\widebar{P}_n \coloneqq \E_\pi \left(P_n\circ \pi\right)$, the convexity of the total variation distance yields
\bb
\frac12\, \big\| \widebar{P}_n - \widebar{P}'_n\big\|_1 \leq \e\, .
\ee
Naturally, in general we will have $\widebar{P}_n \notin \co(\RR_n)$; however, this will turn out not to matter.

We can repeat the calculation in~\eqref{double_Stein_proof_eq3} with $\widebar{P}_n$ and $\widebar{P}'_n$ instead of $P_n$ and $P'_n$. This means, in particular, that~\eqref{double_Stein_proof_eq5} still holds. Leveraging the permutational symmetry of $\widebar{P}'_n$ to write 
\bb
\widebar{P}'_n(x^n) = \frac{\widebar{P}'_n(T_{n,V_n})}{|T_{n,V_n}|}\qquad \forall\ x^n\in T_{n,V_n}
\label{double_Stein_proof_eq20}
\ee
in~\eqref{double_Stein_proof_eq5}, we are led to the inequality
\bb
\min\left\{ P_n(T_{n,V_n}),\, P'_n(T_{n,V_n})\right\} = \min\left\{ \widebar{P}_n(T_{n,V_n}),\, \widebar{P}'_n(T_{n,V_n})\right\} \geq \frac{1-\e}{(n+1)^{|\XX|}}\, ,
\label{double_Stein_proof_eq21}
\ee
where we also observed that permutational symmetrisation does not change the total weight on a given type class. This means, in particular, that Eq.~\eqref{double_Stein_proof_eq6} still holds. Then, also Eq.~\eqref{double_Stein_proof_eq7}--\eqref{double_Stein_proof_eq9} go through without any change. 

We can now re-write~\eqref{double_Stein_proof_eq10} with $\widebar{P}'_n$ instead of $P'_n$. Due to~\eqref{double_Stein_proof_eq20}--\eqref{double_Stein_proof_eq21}, we see that the new set $\YY_n$ produced by~\eqref{double_Stein_proof_eq10} actually coincides with $T_{n,V_n}$. Eq.~\eqref{double_Stein_proof_eq12} \emph{a fortiori} holds, so that~\eqref{double_Stein_proof_eq13}--\eqref{double_Stein_proof_eq16}, again with $P'_n \mapsto \widebar{P}'_n$ and $Q_n \mapsto \widebar{Q}_n$, follow. The rest of the proof can be run unchanged, leading to the contradiction~\eqref{double_Stein_proof_eq20}.
\end{proof}

\begin{rem}
In the case where, in Theorem~\ref{double_Stein_thm}(c), $\SS$ satisfies Axiom~\ref{new_ax_closed_permutations}, it is possible to devise a more direct proof of the claim. Defining
\bb
A_n(x^n) = \left\{ \begin{array}{ll} 1 & \text{ if $\min_{P\in \RR_1} \tfrac12 \|P_{x^n} - P\|_1\leq \delta$,} \\[4pt] 0 & \text{ otherwise,} \end{array} \right.
\ee
it is possible to show, using Axiom~\ref{new_ax_type_stability}, that the tests $A_n$ achieve a vanishing type I error probability. 
Using Lemma~\ref{meta_lemma}, one can then prove that these tests also achieve a type II error exponent that is arbitrarily close to $\inf_{P\in \RR_1} D^\infty(P\|\co(\SS))$.
\end{rem}

\subsection{Proof of Theorem~\ref{stronger_generalised_Sanov_thm}} \label{subsec_proof_stronger_Sanov}

In the forthcoming Section~\ref{sec_classical_applications} we will show how several of the prior result listed in Section~\ref{subsec_prior_results} can be subsumed, and in many cases refined, by our Theorem~\ref{double_Stein_thm}. To make this process smoother, we will first use Theorem~\ref{double_Stein_thm} to establish the slightly simplified Theorem~\ref{stronger_generalised_Sanov_thm},  already reported in Section~\ref{subsec_main_result_applications}. This latter result covers a more specialised class of alternative hypotheses than Theorem~\ref{double_Stein_thm}, but has the decisive advantage of leading to single-letter formulas. We start with two preliminary lemmas that on the one hand will put us in position to wield Theorem~\ref{double_Stein_thm} more easily, and on the other will allow us to efficiently derive useful corollaries from Theorem~\ref{stronger_generalised_Sanov_thm} itself.

\begin{lemma} \label{closed_F_1_iid_lemma}
Let $\FF_1\subseteq \PP(\XX)$ be a topologically closed set of probability distributions on the finite alphabet $\XX$, and let $\FF_1^{\mathrm{iid}} \coloneqq \big(\FF_1^{\otimes n,\, \mathrm{iid}}\big)_n$ be the associated sequence of composite i.i.d.\ hypotheses, defined as in~\eqref{F_n_iid}. Then $\FF_1^{\mathrm{iid}}$ satisfies the type stability axiom (Axiom~\ref{new_ax_type_stability}). Furthermore,
\bb
\rel{D^\infty}{P}{\co\big(\FF_1^{\mathrm{iid}}\big)} = \lim_{n\to\infty} \frac1n\, \rel{D}{P^{\otimes n}}{\co\big(\FF_1^{\otimes n,\,\mathrm{iid}}\big)} = D(P\|\FF_1) = \min_{Q\in \FF_1} D(P\|Q)\, ,
\label{closed_F_1_iid_simplifying_D_infty}
\ee
and the limit exists.
\end{lemma}

\begin{proof}
We start from the first claim. For some $P\in \PP(\XX)$, define
\bb
B_\delta(P) &\coloneqq \left\{ P'\in \PP(\XX):\ \tfrac12\|P-P'\|_1\leq \delta \right\} , \\
T_{n,\, B_\delta(P)} &\coloneqq \left\{ x^n\in \XX^n:\ P_{x^n} \in B_\delta(P) \right\} .
\label{closed_F_1_iid_proof_eq1}
\ee
Then
\bb
\sup_{Q_n\in \FF_1^{\otimes n,\, \mathrm{iid}}} \pr_{X^n\sim Q_n}\!\big\{ \tfrac12 \|P_{X^n} - P \|_1 \leq \delta \big\} 
&=\, \max_{Q\in \FF_1} Q^{\otimes n}\big(T_{n,\, B_\delta(P)}\big) \\[-2ex]
&\leqt{(i)}\, \max_{Q\in \FF_1} \exp\left[ - n\, \rel{D}{B_\delta(P)}{Q} \right] \\
&=\, \exp\left[ - n\, \rel{D}{B_\delta(P)}{\FF_1} \right]
\label{closed_F_1_iid_proof_eq2}
\ee
where in~(i) we used Sanov's theorem in the form~\cite[Exercise~2.12(c), p.~29]{CSISZAR-KOERNER} without polynomial fudge terms, due to the fact that $B_\delta(P)$ is convex. Since $\FF_1$ is closed, if $P\notin \FF_1$ we will also have $B_\delta(P) \cap \FF_1 = \emptyset$ for a small enough $\delta>0$, in turn entailing that the rightmost side of~\eqref{closed_F_1_iid_proof_eq2} vanishes exponentially fast as $n\to\infty$. Thus, if we require that the leftmost side vanish at most polynomially (even if on a single subsequence) for all $\delta>0$, the only possibility is that $P\in \FF_1$. This shows that $\FF_1^{\mathrm{iid}}$ does indeed satisfy Axiom~\ref{new_ax_type_stability}.

We now move on to the proof of~\eqref{closed_F_1_iid_simplifying_D_infty}. The case where $\FF_1$ is also convex follows immediately from more general, quantum results~\cite[Lemma~3.11]{Sutter2016}, but we do not need these prior findings here. Indeed, the general case where $\FF_1$ is only closed can be tackled rather directly. We write
\bb
\rel{D}{P^{\otimes n}}{\co\big(\FF_1^{\otimes n,\,\mathrm{iid}}\big)} &= \inf_{Q_n\in \co\big(\FF_1^{\otimes n,\,\mathrm{iid}}\big)} \rel{D}{P^{\otimes n}}{Q_n} \\
&\geqt{(ii)} \inf_{Q_n\in \co\big(\FF_1^{\otimes n,\,\mathrm{iid}}\big)} \rel{D_2}{P^{\otimes n}(T_{n,B_\delta(P)})}{Q_n(T_{n,B_\delta(P)})} \\
&\geqt{(iii)} -1 + P^{\otimes n}(T_{n,B_\delta(P)}) \log \frac{1}{\sup_{Q_n\in \co\big(\FF_1^{\otimes n,\,\mathrm{iid}}\big)} Q_n(T_{n,B_\delta(P)})} \\
&\eqt{(iv)} -1 + P^{\otimes n}(T_{n,B_\delta(P)}) \log \frac{1}{\max_{Q\in \FF_1} Q^{\otimes n}(T_{n,B_\delta(P)})} \\
&\geqt{(v)} -1 + n P^{\otimes n}(T_{n,B_\delta(P)})\, \rel{D}{B_\delta(P)}{\FF_1}\, .
\label{closed_F_1_iid_proof_eq3}
\ee
Here, in~(ii) we used the data processing inequality and introduced the binary relative entropy given by~\eqref{binary_relative_entropy}; in~(iii) we used~\eqref{lower_bound_D_with_D_H}; in~(iv) we eliminated the convex hull due to the linearity of the function $Q_n\mapsto Q_n(T_{n,B_\delta(P)})$; finally, in~(v) we employed our previous calculation~\eqref{closed_F_1_iid_proof_eq2}. Dividing by $n$, taking the limit infimum as $n\to\infty$, and remembering that $\lim_{n\to\infty} P^{\otimes n}(T_{n,B_\delta(P)}) = 1$ by the law of large numbers gives the inequality
\bb
\rel{D^\infty}{P}{\co\big(\FF_1^{\mathrm{iid}}\big)} = \liminf_{n\to\infty} \frac1n\, \rel{D}{P^{\otimes n}}{\co\big(\FF_1^{\otimes n,\,\mathrm{iid}}\big)} \geq D(B_\delta(P)\|\FF_1)\, .
\ee
Using the lower semi-continuity of the relative entropy together with the fact that $\FF_1$ is closed, we see that the limit $\delta\to 0^+$ yields\footnote{In fact, the second inequality in~\eqref{closed_F_1_iid_proof_eq7} is tight: it actually holds that $\lim_{\delta \to 0^+} D(B_\delta(P)\|\FF_1) = D(P\|\FF_1)$.}
\bb
D^\infty(P\|\FF) \geq \liminf_{\delta \to 0^+} D(B_\delta(P)\|\FF_1) \geq D(P\|\FF_1)\, ,
\label{closed_F_1_iid_proof_eq7}
\ee
which, together with the much more straightforward inequality
\bb
\limsup_{n\to\infty} \frac1n\, \rel{D}{P^{\otimes n}}{\co\big(\FF_1^{\otimes n,\,\mathrm{iid}}\big)} \leq \min_{Q\in \FF_1} \limsup_{n\to\infty} \frac1n\, \rel{D}{P^{\otimes n}}{Q^{\otimes n}} = D(P\|\FF_1)\, ,
\label{closed_F_1_iid_proof_eq8}
\ee
concludes the proof.
\end{proof}

The following result is entirely analogous to the one above, but it deals with the case of an arbitrarily varying instead of a composite i.i.d.\ alternative hypothesis. Its proof, however, is significantly different from that of Lemma~\ref{closed_F_1_iid_lemma}.

\begin{lemma} \label{closed_F_1_av_lemma}
Let $\FF_1\subseteq \PP(\XX)$ be a topologically closed set of probability distributions on the finite alphabet $\XX$, and let $\co\big(\FF_1^{\mathrm{av}}\big) \coloneqq \big(\co\big(\FF_1^{\otimes n,\, \mathrm{av}}\big)\big)_n$, where $\FF_1^{\otimes n,\, \mathrm{av}}$ is defined as in~\eqref{F_n_av}. Then $\co\big(\FF_1^{\mathrm{av}}\big)$ satisfies the type stability axiom (Axiom~\ref{new_ax_type_stability}). Furthermore,
\bb
\rel{D^\infty}{P}{\co\big(\FF_1^{\mathrm{av}}\big)} = \lim_{n\to\infty} \frac1n\, \rel{D}{P^{\otimes n}}{\co\big(\FF_1^{\otimes n,\,\mathrm{av}}\big)} = D(P\|\co(\FF_1)) = \min_{Q\in \co(\FF_1)} D(P\|Q)\, ,
\label{closed_F_1_av_simplifying_D_infty}
\ee
and the limit exists.
\end{lemma}

\begin{proof}
The first claim follows from Proposition~\ref{verifying_type_stability_prop}. Let us see why. First, let us check that $\co\big(\FF_1^{\mathrm{av}}\big)$ satisfies Axiom~\ref{new_ax_depolarising}. Taking an arbitrary $R$ in the relative interior of $\co(\FF_1)$, we have immediately that $\supp(Q)\subseteq \supp(R)$ for all $Q\in \co(\FF_1)$, which also entails that $\supp(Q_n)\subseteq \supp(R)^n$ for all $Q_n\in \co\big(\FF_1^{\otimes n,\,\mathrm{av}}\big)$. Also, since $\mm$ maps $\co(\FF_1)$ into itself, an elementary calculation reveals that $\mm^{\otimes n}$ does the same on $\co\big(\FF_1^{\otimes n,\,\mathrm{av}}\big)$, for all $\delta\in [0,1]$. To see why, take an arbitrary
\bb
Q_n = \sum_j \lambda_j\, Q_{1,j} \otimes \ldots \otimes Q_{n,j} \in \co\big(\FF_1^{\otimes n,\,\mathrm{av}}\big)\, ,\qquad Q_{i,j} \in \FF_1\quad \forall\ i,j\, ,
\ee
and observe that
\bb
\mm^{\otimes n}(Q_n) = \sum_j \lambda_j\, \big((1-\delta) Q_{1,j} + \delta R\big) \otimes \ldots \otimes \big((1-\delta) Q_{n,j} + \delta R\big) \in \co\big(\FF_1^{\otimes n,\,\mathrm{av}}\big)\, ,
\ee
as one sees by expanding the tensor product. This completes the verification of Axiom~\ref{new_ax_depolarising}. Axiom~\ref{new_ax_closed_permutations} is immediate, while Axiom~\ref{new_ax_filtering} holds for $W$ equal to the identity channel. Since $\FF_1$ is closed, and hence compact, the same is true of $\co(\FF_1)$. This shows that we can indeed apply Proposition~\ref{verifying_type_stability_prop} to establish the first claim.

The identity in~\eqref{closed_F_1_av_simplifying_D_infty}, instead, follows from a reasoning essentially identical to that used in the proof of~\cite[Lemma~3.11]{Sutter2016}. The upper bound 
\bb
\limsup_{n\to\infty} \frac1n\, \rel{D}{P^{\otimes n}}{\co\big(\FF_1^{\otimes n,\,\mathrm{av}}\big)} \leq D(P\|\co(\FF_1))
\label{closed_F_1_av_proof_eq3}
\ee
is straightforward, following from the family of ansatzes $Q^{\otimes n} \in \co\big(\FF_1^{\otimes n,\,\mathrm{av}}\big)$ in the second argument of the relative entropy, where $Q\in \co(\FF_1)$ is arbitrary. This is analogous to~\eqref{closed_F_1_iid_proof_eq8} above.

As for the lower bound, it suffices to observe that~\cite[Eq.~(40)--(41)]{Sutter2016} hold in the same way if in the first lines one replaces $\int \mu(\dd x)\, \sigma_x^{\otimes n}$ with an arbitrary $\sigma_n\in \co\{\sigma_{x_1}\otimes \ldots \otimes \sigma_{x_n}:\ x_1,\ldots,x_n \in \mathds{X} \} = \{\sigma_x\!: x\in \mathds{X}\}^{\otimes n,\,\mathrm{av}}$, where we followed the notation of~\cite{Sutter2016}, together with (an obvious quantum extension of) our own in~\eqref{F_n_av}. Then, one can proceed like in the rest of the proof of~\cite[Lemma~3.11]{Sutter2016}, obtaining
\bb
\frac1n \inf_{\sigma_n \in \co\left(\{\sigma_x:\ x\in \mathds{X}\}^{\otimes n,\,\mathrm{av}}\right)} \rel{D_{\mathds{M}}}{\rho^{\otimes n}}{\sigma_n} \geq \min_{\sigma \in \co \{\sigma_x:\, x\in \mathds{X}\}} D_{\mathds{M}}(\rho\|\sigma)\, .
\ee
Specialising this to classical probability distributions, we deduce that
\bb
\frac1n\, \rel{D}{P^{\otimes n}}{\co\big(\FF_1^{\otimes n,\,\mathrm{av}}\big)} \geq D(P\|\co(\FF_1))\, ,
\ee
for all positive integers $n\in \N^+$. Taking the limit inferior as $n\to\infty$ shows that
\bb
\rel{D^\infty}{P}{\co\big(\FF_1^{\mathrm{av}}\big)} = \liminf_{n\to\infty} \frac1n\, \rel{D}{P^{\otimes n}}{\co\big(\FF_1^{\otimes n,\,\mathrm{av}}\big)} \geq D(P\|\co(\FF_1))\, ,
\ee
which, together with~\eqref{closed_F_1_av_proof_eq3}, completes the proof.
\end{proof}

We are now ready to present the proof of Theorem~\ref{stronger_generalised_Sanov_thm}, reported below for the reader's convenience. 

\begin{manualthm}{\ref{stronger_generalised_Sanov_thm}}
Let $\XX$ be a finite alphabet, $\SS_1\subseteq \PP(\XX)$ a set of probability distributions on $\XX$, and $\RR = (\RR_n)_n$ a family of sets $\RR_n\subseteq \PP(\XX^n)$. Assume that either
\begin{enumerate}
\item[(a)] $\RR$ satisfies Axioms~\ref{new_ax_tensor_powers} and~\ref{new_ax_type_stability}; also, $\RR_1$ is topologically closed; or
\item[(a')] $\RR$ satisfies Axioms~\ref{new_ax_depolarising},~\ref{new_ax_tensor_powers},~\ref{new_ax_closed_permutations}, and~\ref{new_ax_filtering}, all sets $\RR_n$ are convex, and $\RR_1$ is topologically closed.
\end{enumerate}
Then, with the notation in~\eqref{F_n_av}, the Stein exponent defined as in~\eqref{Stein} is given by
\begin{equation}
\rel{\stein}{\RR}{\SS_1^{\mathrm{av}}} = D(\RR_1\|\co(\SS_1)) = \inf_{P\in \RR_1,\, Q\in \co(\SS_1)} D(P\|Q)\, .
\tag{\ref{stronger_generalised_Sanov_av}}
\end{equation}
If, moreover, 
\begin{enumerate}[(a)] \setcounter{enumi}{1}
\item $\SS_1$ is star-shaped around some $R\in \SS_1$ such that $\supp(Q)\subseteq \supp(R)$ for all $Q\in \SS_1$,
\end{enumerate}
then it also holds that
\begin{equation}
\rel{\stein}{\RR}{\SS_1^{\mathrm{iid}}} = D(\RR_1\|\SS_1) = \inf_{P\in \RR_1,\, Q\in \SS_1} D(P\|Q)\, , \tag{\ref{stronger_generalised_Sanov_iid}}
\end{equation}
where the notation is defined in~\eqref{F_n_iid} and~\eqref{Stein}.
\end{manualthm}

\begin{proof}
Clearly, (a')~implies~(a), due to Proposition~\ref{verifying_type_stability_prop}. Hence, we can assume that~(a) holds without loss of generality. For~\eqref{stronger_generalised_Sanov_av}, we can then write
\bb
\rel{\stein}{\RR}{\SS_1^{\mathrm{av}}} &\eqt{(i)}\, \rel{\stein}{\RR}{\co\big(\SS_1^{\mathrm{av}}\big)} \\
&\eqt{(ii)}\, \inf_{P\in \RR_1} \rel{D^\infty}{P}{\co\big(\SS_1^{\mathrm{av}}\big)} \\
&\eqt{(iii)}\,\inf_{P\in \RR_1} D(P \| \co(\SS_1)) \\
&=\, D(\RR_1\|\co(\SS_1))\, .
\ee
Here, (i)~holds due to~\eqref{convexify_Stein}, while in~(ii) we applied Theorem~\ref{double_Stein_thm}. To see why this is possible, recall that $\co\big(\SS_1^{\mathrm{av}}\big)$ satisfies Axiom~\ref{new_ax_depolarising} if one takes as $R$ a probability distribution in the relative interior of $\co(\SS_1)$, as we already argued in the first part of the proof of Lemma~\ref{closed_F_1_av_lemma}; note also that condition~(a) is identical in Theorems~\ref{double_Stein_thm} and~\ref{stronger_generalised_Sanov_thm}, and that $\co(\SS_1)^{\mathrm{av}}$ satisfies Axiom~\ref{new_ax_closed_permutations} by construction (see~\eqref{F_n_av}). Finally, in~(iii) we applied Lemma~\ref{closed_F_1_av_lemma} to remove the regularisation.

The proof of~\eqref{stronger_generalised_Sanov_iid} is essentially analogous: one writes
\bb
\rel{\stein}{\RR}{\SS_1^{\mathrm{iid}}} \eqt{(iv)} \inf_{P\in \RR_1} \rel{D^\infty}{P}{\co\big(\SS_1^{\,\mathrm{iid}}\big)} \,\eqt{(v)}\, D(\RR_1\|\SS_1)\, .
\ee
Here, in~(iv) we applied Theorem~\ref{double_Stein_thm}, and (v)~follows from Lemma~\ref{closed_F_1_iid_lemma}. Applying Theorem~\ref{double_Stein_thm} here is possible, because, due to assumption~(b), the sequence $\SS_1^{\mathrm{iid}}$ satisfies Axiom~\ref{new_ax_depolarising}, meeting condition~(b) in Theorem~\ref{double_Stein_thm}; the other conditions can be verified as before.
\end{proof}

\section{Applications} \label{sec_classical_applications}

Throughout this section we explore some applications of our main results (Theorem~\ref{double_Stein_thm} and~\ref{stronger_generalised_Sanov_thm}) to classical information theory. Applications to quantum information theory are detailed in a companion paper~\cite{doubly-comp-quantum}.

\subsection{Composite i.i.d.\ null hypothesis with closed (non-convex) base set}

We start with setting~(A) in Section~\ref{subsec_prior_results}, which features a composite i.i.d.\ null hypothesis and a simple i.i.d.\ alternative hypothesis. The following statement, reported here as~\eqref{Sanov_theorem}, is due to Sanov~\cite{Sanov1957, Hoeffding1965}. Here we show that it is easily implied by our general result, Theorem~\ref{stronger_generalised_Sanov_thm}.

\begin{cor}[{\cite{Sanov1957, Hoeffding1965}}] \label{composite_iid_null_cor}
Let $\RR_1\subseteq \PP(\XX)$ be a closed set of probability distributions on the finite alphabet $\XX$, and let $\RR_1^{\mathrm{iid}} \coloneqq \big(\RR_1^{\otimes n,\, \mathrm{iid}}\big)_n$ be the associated sequence of composite i.i.d.\ hypotheses, defined as in~\eqref{F_n_iid}. Then, for all $Q\in \PP(\XX)$, 
\bb
\rel{\stein}{\RR_1^{\mathrm{iid}}}{Q} = D(\RR_1\|Q) = \min_{P\in \RR_1} D(P\|Q)\, .
\ee
\end{cor}

\begin{proof}
Setting $\SS_1 = \{Q\}$, we see immediately that condition~(b) in Theorem~\ref{stronger_generalised_Sanov_thm} is met. The sequence $\RR_1^{\mathrm{iid}}$ clearly satisfies Axiom~\ref{new_ax_tensor_powers}, and it also satisfies Axiom~\ref{new_ax_type_stability} because of Lemma~\ref{closed_F_1_iid_lemma}. Thus, condition~(a) is also met, and the conclusion follows from~\eqref{stronger_generalised_Sanov_iid}.
\end{proof}

\subsection{The case where both hypotheses are either composite i.i.d.\ or arbitrarily varying} \label{subsec_recovering_settings_BCD}

Next, we deal with settings~(C) and~(D) in Section~\ref{subsec_prior_results}. Curiously, we cannot recover the result in~(B), i.e.\ Eq.~\eqref{Stein_iid_without_convexity}, deduced from~\cite[Theorem~III.2]{Mosonyi2022}, which solves the case where both $\RR_1$ and $\SS_1$ are finite, as our approach relies heavily on Axiom~\ref{new_ax_depolarising}, which requires $\SS_1$ to be star-shaped. However, we can state a different result that covers instead settings~(C) and~(D), subsuming both~\eqref{Stein_av}, which is taken from~\cite[Theorem~III.7]{Mosonyi2022}, and~\eqref{Stein_iid_or_av_convex}, due to~\cite{Fangwei1996, Levitan2002, brandao_adversarial, Fang2025}.

\begin{cor} \label{both_composite_iid_or_av_cor}
Let $\RR_1,\SS_1\subseteq \PP(\XX)$ be closed sets of probability distributions on the finite alphabet $\XX$. Then
\begin{align}
\rel{\stein}{\RR_1^{\mathrm{iid}}}{\SS_1^{\mathrm{av}}} &= D(\RR_1\|\co(\SS_1))\, , \label{iid_vs_av} \\
\rel{\stein}{\RR_1^{\mathrm{av}}}{\SS_1^{\mathrm{av}}} &= D(\co(\RR_1)\|\co(\SS_1))\, , \label{both_av}
\end{align}
where the hypotheses $\RR_1^{\mathrm{a}}$ and  $\SS_1^{\mathrm{b}}$, with $\mathrm{a},\mathrm{b}\in \{\mathrm{iid},\mathrm{av}\}$, are defined in~\eqref{F_n_iid}--\eqref{F_n_av}, and we adopted the convention~\eqref{divergence_sets} to define the relative entropy between sets.  
Furthermore, if $\SS_1$ is star-shaped around some $R\in \SS_1$ with the property that $\supp(Q) \subseteq \supp(R)$ for all $Q\in \SS_1$ (for example, this holds if $\SS_1$ is convex), then we also have
\begin{align}
\rel{\stein}{\RR_1^{\mathrm{iid}}}{\SS_1^{\mathrm{iid}}} &= D(\RR_1\|\SS_1)\, , \label{both_composite_iid} \\
\rel{\stein}{\RR_1^{\mathrm{av}}}{\SS_1^{\mathrm{iid}}} &= D(\co(\RR_1)\|\SS_1)\, . \label{av_vs_iid}
\end{align}
Consequently, if both $\RR_1$ and $\SS_1$ are closed and convex, then we recover the result due to~\cite{Fangwei1996, Levitan2002, brandao_adversarial, Fang2025} and reported here in~\eqref{Stein_iid_or_av_convex}:
\bb
\rel{\stein}{\RR_1^\mathrm{a}}{\SS_1^\mathrm{b}} = D(\RR_1\|\SS_1)\qquad \forall\ \mathrm{a},\mathrm{b}\in \{\mathrm{iid},\mathrm{av}\}\, .
\label{both_convex_composite_iid_or_av}
\ee
\end{cor}

\begin{proof}
To prove~\eqref{iid_vs_av}, simply apply Theorem~\ref{stronger_generalised_Sanov_thm} (specifically,~\eqref{stronger_generalised_Sanov_av}) with $\RR \mapsto \RR_1^\mathrm{iid}$: this sequence satisfies Axiom~\ref{new_ax_type_stability} by Lemma~\ref{closed_F_1_iid_lemma}, and also Axiom~\ref{new_ax_tensor_powers} holds.
The proof of~\eqref{both_av} is similar, but we first need to convexify the null hypothesis:
\bb
\rel{\stein}{\RR_1^\mathrm{av}}{\SS_1^\mathrm{av}} = \rel{\stein}{\co\big(\RR_1^\mathrm{av}\big)}{\SS_1^\mathrm{av}} = D(\co(\RR_1)\|\co(\SS_1))\, ,
\ee
where the first equality holds by~\eqref{convexify_Stein}, and in the second we applied~\eqref{stronger_generalised_Sanov_av} in Theorem~\ref{stronger_generalised_Sanov_thm}, noting that $\co\big(\RR_1^\mathrm{av}\big)$ satisfies Axiom~\ref{new_ax_type_stability} due to Lemma~\ref{closed_F_1_av_lemma}. To establish~\eqref{both_composite_iid} and~\eqref{av_vs_iid} one can argue similarly, but using~\eqref{stronger_generalised_Sanov_iid} instead of~\eqref{stronger_generalised_Sanov_av} in Theorem~\ref{stronger_generalised_Sanov_thm}. Eq.~\eqref{both_convex_composite_iid_or_av} follows trivially.
\end{proof}

\subsection{Generalised classical Stein's lemma: an almost-i.i.d.\ extension} \label{subsec_GSL_almost_iid}

In what follows, we will extend the result reported in point~(E) of Section~\ref{subsec_prior_results}, namely the generalised classical Stein's lemma~\cite[Theorem~4]{GQSL}, to a broader --- and more natural --- class of almost i.i.d.\ sources than was treated in~\cite[Theorem~32]{GQSL}. Indeed, that result, reproduced in~\eqref{GSL_almost_iid}, only covered sources with a constant number of defects. Here, we show how to handle any \emph{sublinear} number of defects. This corresponds to a more satisfactory notion of what it means for a source to be `almost i.i.d.', and removes the obstacles that prevented the extension of the proof in~\cite[Theorem~32]{GQSL}, which were primarily technical.

We denote by $\varphi(n)$ the maximum number of defects in a source outputting strings in $\XX^n$, where $\varphi:\N^+\to\N$ is some integer-valued function. Given such a function $\varphi$ and a distribution $P\in\PP(\XX)$, we define the associated sequence of \deff{almost i.i.d.\ hypotheses} as~\cite{RennerPhD, Renner-talk-Cambridge}
\bb
\RR^{\mathrm{aiid}}_{\varphi,P} \coloneqq \big( \RR^{\mathrm{aiid}}_{n,\varphi,P}\big)_n\, ,\qquad  \RR^{\mathrm{aiid}}_{n,\varphi,P} \coloneqq \left\{ P^{\otimes I^c} \otimes Q^{I}:\ I\subseteq [n],\ \, |I|\leq \varphi(n),\ \, Q\in \PP\scaleobj{1.22}{(}\XX^{|I|}\scaleobj{1.22}{)} \right\} ,
\ee
where superscripts denote the sites to which each probability distribution pertains. Instead of assuming that $\varphi$ is bounded, as done in~\cite[Theorem~32]{GQSL}, here we will consider the general sublinear case, in which we only know that
\bb
\lim_{n\to\infty} \frac{\varphi(n)}{n} = 0\, .
\ee
When this happens, the source is, in some sense, locally indistinguishable from a perfectly i.i.d.\ source in the asymptotic limit, in the sense that any collection of random variables $X_{i_1},\ldots, X_{i_k}$, with $k$ constant, is distributed according to $P^{\otimes k}$ in the limit of large $n$. An indeed, the following result shows that in the context of hypothesis testing such a source behaves precisely like a perfectly i.i.d.\ one.

\begin{cor} \label{almost_iid_GSL_cor}
Let $\XX$ be a finite alphabet, $P\in \PP(\XX)$ a probability distribution, and $\SS = (\SS_n)_n$ a sequence of sets $\SS_n \subseteq \PP(\XX^n)$ that obeys Axioms~\ref{new_ax_depolarising} and~\ref{new_ax_closed_permutations}. Then, for each function $\varphi:\N^+ \to \N$ such that $\lim_{n\to\infty} \frac{\varphi(n)}{n} = 0$, it holds that
\bb
\rel{\stein}{\RR^{\mathrm{aiid}}_{\varphi,P}}{\SS} = D^\infty(P\|\co(\SS))\, .
\label{almost_iid_GSL}
\ee
\end{cor}

Before we report the proof of the above result, we take a moment to highlight why exactly it is a strict generalisation of~\cite[Theorem~4]{GQSL}. In essence, this is because the latter theorem requires all the Brand\~{a}o--Plenio axioms, and these together are much stronger than the assumptions of Corollary~\ref{almost_iid_GSL_cor}, as we now show.

\begin{lemma} \label{BP_ax_1_5_imply_new_ax_1_3_lemma}
Axioms~\ref{BP_ax_convex_closed}--\ref{BP_ax_permutations} together imply Axioms~\ref{new_ax_depolarising}--\ref{new_ax_closed_permutations}.
\end{lemma}

\begin{proof}[Proof of Lemma~\ref{BP_ax_1_5_imply_new_ax_1_3_lemma}]
The only non-trivial part of the claim is to show that Axioms~\ref{BP_ax_convex_closed}--\ref{BP_ax_permutations} imply Axiom~\ref{new_ax_depolarising}. Choose as $R$ the probability distribution with full support whose existence is guaranteed by Axiom~\ref{BP_ax_full_rank}, and consider a random string $X^n\sim Q_n\in \FF_n$. The map $\mm^{\otimes n}$ can be implemented on $X^n$ by: (i)~appending $n$ independent variables $X_{n+1},\ldots, X_{2n}$ distributed according to $R$ (which maps $\FF_n$ to $\FF_{2n}$ by Axiom~\ref{BP_ax_tensor_products}); (ii)~for all $j=1,\ldots,n$, swapping $X_j$ and $X_{n+j}$ independently with probability $\delta$ (which maps $\FF_{2n}$ to $\FF_{2n}$ by convexity and Axiom~\ref{BP_ax_permutations}); and (iii)~discarding the last $n$ variables (which maps $\FF_{2n}$ back to $\FF_n$ by Axiom~\ref{BP_ax_partial_traces}). Therefore, $\mm^{\otimes n}(Q_n)\in \FF_n$, as claimed.
\end{proof}

We are now ready to present the proof of Corollary~\ref{almost_iid_GSL_cor}.

\begin{proof}[Proof of Corollary~\ref{almost_iid_GSL_cor}]
A preliminary step is to re-define the value of the function $\varphi$ at $n=1$, so that $\varphi(1) = 0$. Clearly, this can be done without affecting either the Stein exponent or the sublinear behaviour of $\varphi$, since these are purely asymptotic notions. The condition that $\varphi(1) = 0$ simply ensures that $\RR_{1,\varphi,P} = \{P\}$.

Now, requirements~(b) and~(c) in Theorem~\ref{double_Stein_thm} are met by assumption. As for~(a), first note that $\RR^{\mathrm{aiid}}_{\varphi,P}$ clearly satisfies Axiom~\ref{new_ax_tensor_powers}, because $\RR^{\mathrm{aiid}}_{1,\varphi,P} = \{P\}$ and $P^{\otimes n} \in \RR^{\mathrm{aiid}}_{n,\varphi,P}$ for all $n\in \N^+$. The only nontrivial assumption that remains to be checked is that $\RR^{\mathrm{aiid}}_{\varphi,P}$ meets Axiom~\ref{new_ax_type_stability}. To this end, one can modify slightly the argument used in the first part of the proof of Lemma~\ref{closed_F_1_iid_lemma}. For any $V \in \PP(\XX)$, we can replicate~\eqref{closed_F_1_iid_proof_eq2} and write, using the notation of~\eqref{closed_F_1_iid_proof_eq1},
\begin{align}
\sup_{P_n\in \RR^{\mathrm{aiid}}_{n,\varphi,P}} \pr_{X^n\sim P_n}\!\big\{ \tfrac12 \|P_{X^n} \!-\! V \|_1 \leq \delta \big\} &= \max_{0\leq r\leq \varphi(n),\ Q_r\in \PP(\XX^r)} \sum_{x^n\in \XX^n:\ \frac12\|P_{x^n} - V\|_1\leq \delta}\big(P^{\otimes (n-r)} \otimes Q_r\big)(x^n) \nonumber \\
&\eqt{(i)} \max_{Q_{\varphi(n)}\in \PP(\XX^{\varphi(n)})} \sum_{x^n\in \XX^n:\ \frac12\|P_{x^n} - V\|_1\leq \delta}\big(P^{\otimes (n-\varphi(n))} \otimes Q_{\varphi(n)}\big)(x^n) \nonumber \\
&\leqt{(ii)} \sum_{x^n\in \XX^n:\ \frac12\|P_{x^n} - V\|_1\leq \delta} P^{\otimes (n-\varphi(n))}(x^{n-\varphi(n)}) \label{almost_iid_GSL_proof_eq1} \\
&\leqt{(iii)} |\XX|^{\varphi(n)}\, P^{\otimes (n-\varphi(n))}\big(T_{n-\varphi(n),\, B_{\delta + \varphi(n)/n}(V)}\big) \nonumber \\
&\leqt{(iv)} |\XX|^{\varphi(n)} \exp\left[ - (n-\varphi(n))\, \rel{D}{B_{\delta+\varphi(n)/n}(V)}{P} \right] \nonumber
\end{align}
The above derivation can be justified as follows. In~(i) we observed that setting $r=\varphi(n)$ causes no loss of generality, as we can always include in $Q_{\varphi(n)}$ a few copies of $P$ to effectively reduce the number of defects. In~(ii) we denoted by $x^{n-\varphi(n)}$ the string composed of the first $n-\varphi(n)$ symbols of $x^n$, and observed that $Q_{\varphi(n)}\big(y^{\varphi(n)}\big)\leq 1$ for all $y^{\varphi(n)}\in \XX^{\varphi(n)}$. To see why~(iii) holds, start by noting that, for all $r$, the type of $x^{n-r}$ has a total variation distance of at most $r/n$ from that of $x^n$, simply because, for all $A\subseteq \XX$,
\bb
n \sum_{y\in A} \left( P_{x^{n-r}}(y) - P_{x^n}(y) \right) &= \sum_{y\in A} \left( (n-r) P_{x^{n-r}}(y) - n P_{x^n}(y) \right) + r \sum_{y\in A} P_{x^{n-r}}(y)\leq r \sum_{y\in A} P_{x^{n-r}}(y) \leq r\, ,
\ee
where the first inequality holds because, adopting the notation of~\eqref{type_of_sequence}, $(n-r) P_{x^{n-r}}(y) = N(y|x^{n-r}) \leq N(y|x^n) = nP_{x^n}(y)$. Dividing by $n$ and taking the maximum over all sets $A\subseteq \XX$ yields precisely $\frac12 \|P_{x^{n-r}} - P_{x^n}\|_1\leq \frac{r}{n}$. What this shows, in particular, is that any string $x^{n-\varphi(n)}$ that appears on the right-hand side of~(iii) satisfies $\frac12 \|P_{x^{n-\varphi(n)}} - V\|_1\leq \delta + \frac{\varphi(n)}{n}$, and it thus belongs to $T_{n-\varphi(n),\, B_{\delta + \varphi(n)/n}(V)}$. Now we should ask ourselves: how many different strings $x^n$ can be mapped to the same string $x^{n-\varphi(n)}$? The answer, rather obviously, is: precisely $|\XX|^{\varphi(n)}$.  This explains also the coefficient on the right-hand side of~(iii), and completes the justification of this step. Finally, in~(iv) we used once again Sanov's theorem, in the stronger form of~\cite[Exercise~2.12(c), p.~29]{CSISZAR-KOERNER}, which is applicable because $B_{\delta + \varphi(n)/n}(V)$ is convex.

Now that we have proved~\eqref{almost_iid_GSL_proof_eq1}, we can proceed as in the proof of Lemma~\ref{closed_F_1_iid_lemma}. If the leftmost side of~\eqref{almost_iid_GSL_proof_eq1} vanishes no faster than polynomially (in $n$), at least on a subsequence, then the only possibility is that $P\in B_{\delta'}(V)$ for all $\delta'>\delta>0$. Since $\delta'$ and $\delta$ are otherwise arbitrary, it must be the case that $P=V$, which completes the verification of Axiom~\ref{new_ax_type_stability}. In turn, this allows us to apply Theorem~\ref{double_Stein_thm}, which yields immediately~\eqref{almost_iid_GSL} and completes the proof.
\end{proof}

\subsection{Relation with the generalised classical Sanov theorem}

Finally, we comment briefly on why Theorem~\ref{stronger_generalised_Sanov_thm} constitutes a strict extension of~\cite[Theorem~8, Eq.~(C4)]{generalised-Sanov}. In the setting of this latter result, the alternative hypothesis $\SS_1$ is simple and i.i.d., and, as such, it obviously obeys assumption~(b) of Theorem~\ref{stronger_generalised_Sanov_thm}. On the null hypothesis side, in~\cite[Theorem~8, Eq.~(C4)]{generalised-Sanov} it is assumed that $\RR$ satisfies all of the Brand\~{a}o--Plenio axioms (Axioms~\ref{BP_ax_convex_closed}--\ref{BP_ax_permutations}) and moreover Axiom~\ref{BP_ax_6}. As it turns out, these assumptions together are strictly stronger than, and hence imply, Axioms~\ref{new_ax_tensor_powers} and~\ref{new_ax_type_stability}. This shows that Theorem~\ref{stronger_generalised_Sanov_thm} strictly subsumes~\cite[Theorem~8, Eq.~(C4)]{generalised-Sanov}, as claimed.

\begin{lemma} \label{BP_ax_1_6_imply_new_ax_1_4_lemma}
Axioms~\ref{BP_ax_convex_closed}--\ref{BP_ax_6} together imply Axioms~\ref{new_ax_depolarising}--\ref{new_ax_type_stability}.
\end{lemma}

\begin{proof}
Let $\FF = (\FF_n)_n$ be a sequence of sets $\FF_n\subseteq \PP\big(\XX^n\big)$. Due to Lemma~\ref{BP_ax_1_5_imply_new_ax_1_3_lemma}, we need only to show that, in the presence of Axioms~\ref{BP_ax_convex_closed}--\ref{BP_ax_permutations}, Axiom~\ref{BP_ax_6} implies Axiom~\ref{new_ax_type_stability}. The same lemma also tells us that we can assume without loss of generality that Axiom~\ref{new_ax_depolarising} holds with respect to a constant $c>0$ and some probability distribution $R\in \FF_1$ with $\supp(R) = \XX$ (as guaranteed by Axiom~\ref{BP_ax_full_rank}). Note that Axiom~\ref{new_ax_tensor_products_n_and_m} is satisfied, too, as it coincides with Axiom~\ref{BP_ax_tensor_products}. We can thus directly apply Lemma~\ref{meta_lemma_perm_symm}, and in particular~\eqref{meta_lemma_perm_symm_bound_infty}, and conclude the following: for all $\Delta,\delta>0$, with $\delta<1/3$, all $P\in \PP(\XX)$, and all sufficiently large $n$, we have
\bb
\sup_{Q_n\in \FF_n} Q_n(T_{n,V}) \leq \exp\left[ - n \left(D^\infty(P \| \FF) - \phi(\delta) - \Delta\right) \right]
\ee
for all types $V\in \TT_n$ such that $\frac12 \|V-P\|_1\leq \delta$. (The support condition is empty, as $\supp(R) = \XX$.) Here, $\phi$ is a continuous function satisfying $\phi(0) = 0$. Thus,
\bb
\sup_{Q_n\in \FF_n} \pr_{X^n\sim Q_n}\!\big\{ \tfrac12 \|P_{X^n} - P \|_1 \leq \delta \big\} &\leq |\TT_n| \sup_{\substack{Q_n\in \FF_n,\ V\in \TT_n:\ \frac12\|V-P\|_1\leq \delta}} Q_n(T_{n,V}) \\
&\leq (n+1)^{|\XX|} \exp\left[ - n \left(D^\infty(P \| \FF) - \phi(\delta) - \Delta\right) \right] .
\ee
If the leftmost side decays at most polynomially in $n$ as $n\to\infty$, even if on a single subsequence, and since $\Delta$ and $\delta$ can be chosen to be as small as one pleases, the only possibility is that $D^\infty(P \| \FF) = 0$. By Axiom~\ref{BP_ax_6}, this can only be the case if $P\in \FF_1$. This establishes Axiom~\ref{new_ax_type_stability} and concludes the proof.
\end{proof}

\subsection{Constrained de Finetti reduction}

De Finetti theorems provide a way to reduce general permutationally symmetric probability distributions to convex combinations of i.i.d.\ distributions~\cite{deFinetti-original, deFinetti}. Originally studied in the classical setting, they have been thoroughly investigated also in the framework of quantum information theory~\cite{deFinetti0, RennerPhD, deFinetti2, 1-1/2-de-Finetti, deFinetti4-CMP}. It is in this latter context that a special class of these statements, called \emph{de Finetti reductions} (or `post-selection lemmas'), have been first proposed~\cite{Christandl2009}. 
We focus here on the classical case first, and then state a conjecture concerning possible quantum generalisations. In its most elementary form, a de Finetti reduction shows the existence of a universal probability measure $\dd P$ on $\PP(\XX)$ such that, for all $n\in \N^+$, every permutationally symmetric probability distribution $Q_n\in \PP(\XX^n)$ satisfies the entry-wise inequality
\bb
Q_n \leq L(n) \int_{\PP(\XX)} \!\!\dd P\ P^{\otimes n} ,
\ee
where $L(n)$ is a polynomial --- and thus, in particular, a sub-exponential function --- that depends only on $|\XX|$. The distribution on the right-hand side is an example of a \emph{universal distribution}, in the sense of~\cite[Axiom~4 and Lemma~14]{Tomamichel2018}.

Here we follow the philosophy of~\cite{Lancien2017}, where it is argued that the universality of the above construction is both a blessing and a curse. It is a blessing because it simplifies the analysis of arbitrary permutationally symmetric distribution immensely, reducing the general case to the i.i.d.\ case; yet, it is also a curse, because its universality means that any information on $Q_n$ is lost. For example, we might know that $Q_n\in \FF_n$ belongs to the $n$-symbol instance of some some special sequence of sets $\FF = (\FF_n)_n$, with $\FF_n\subseteq \PP(\XX^n)$, and we might want a de Finetti reduction that makes use of this information, in that it features only i.i.d.\ distributions $P^{\otimes n}$ in which $P$ is also in $\FF_1$, or at least very close to it. In~\cite{Lancien2017}, \emph{constrained de Finetti reductions} of this sort were put forward, even in the quantum case. Typically, those results can be phrased as follows: given a sequence $\FF = (\FF_n)_n$ that obeys some stability constraints (typically, some of the Axioms~\ref{BP_ax_convex_closed}--\ref{BP_ax_6}), any $Q_n\in \FF_n$ satisfies that
\bb
Q_n \leq L(n) \int_{\PP(\XX)} \!\!\dd P\ \exp\left[- \rel{D_{1/2}}{P^{\otimes n}}{\FF_n} \right] P^{\otimes n} ,
\ee
where $D_{1/2}(P\|Q) \coloneqq - 2\log \sum_x \sqrt{P(x)Q(x)} = - 2 \log F(P,Q)$ is the R\'enyi-\sfrac{1}{2} relative entropy.\footnote{Using statements analogous to our Lemma~\ref{superadd_filtered_lemma}, in several cases of interest the authors of~\cite{Lancien2017} were then able to show that $\rel{D_{1/2}}{P^{\otimes n}}{\FF_n}$ grows linearly in $n$ whenever $P\notin \FF_1$, which yields the sought exponential suppression of the single-copy distributions that are outside of $\FF_1$.}

With our techniques we can now provide a tighter estimate, in which the R\'enyi-\sfrac{1}{2} relative entropy is replaced by the more fundamental relative entropy. Below we conjecture that this result might be extended to the quantum setting, potentially yielding a new interpretation of the regularised relative entropy of resource in the context of de Finetti reductions.

\begin{lemma}[(Classical constrained de Finetti reduction)] \label{constrained_deFinetti_lemma}
For a finite alphabet $\XX$, let $\FF = (\FF_n)_n$ be a sequence of convex sets $\FF_n\subseteq \PP(\XX^n)$ that obeys Axioms~\ref{new_ax_depolarising} and~\ref{new_ax_closed_permutations}, the former with respect to a probability distribution $R\in \PP(\XX)$ and a constant $c$ such that $\min_{x\in \supp(R)} R(x) \geq c > 0$. Then there exists a measure $\dd P$ on $\PP(\XX)$ with the following property: for any $\Delta>0$, we can find $N = N(\Delta,c,|\XX|) \in \N^+$ such that, for all $n\geq N$, all permutationally symmetric $Q_n\in \FF_n$ satisfy the entry-wise inequality
\bb
Q_n \leq \int_{\PP(\XX)} \dd P\ \exp\left[ - D(P^{\otimes n} \| \FF_n) + n \Delta \right]\, P^{\otimes n} .
\label{constrained_deFinetti}
\ee
If $\FF$ obeys also Axiom~\ref{new_ax_tensor_products_n_and_m}, then we can even write, again for $n\geq N$,
\bb
Q_n \leq \int_{\PP(\XX)} \dd P\ \exp\left[ - n\left( D^\infty(P\| \FF) - \Delta \right) \right]\, P^{\otimes n} ,
\label{constrained_deFinetti_special}
\ee
where $D^\infty(P\| \FF)$ is defined by~\eqref{regularised_relent} (and the limit infimum can be replaced with an ordinary limit).
\end{lemma}

\begin{proof}
Let $d\coloneqq |\supp(R)|$ denote the cardinality of the support of $R$. Consider the map from the $(d-1)$-sphere $S_{d-1}$ embedded in $\R^d$ to the probability simplex $\PP(\XX)$ given by 
\bb
\R^d\supseteq S_{d-1} \ni \Psi \mapsto P_\Psi \in \PP(\XX)\, , \qquad P_\Psi(x) \coloneqq \Psi(x)^2.
\ee
Denote with $\dd P$ the push-forward of the uniform measure $\dd \Psi$ on $S_{d-1}$ to $\PP(\XX)$ obtained via this map. Due to the Fuchs--van de Graaf inequalities~\cite{Fuchs1999}, for any two $\Psi,\Phi\in S_{d-1}$ we have
\bb
\|\Psi - \Phi\|_2 = \sqrt{\sumno_x \left(\Psi(x) - \Phi(x)\right)^2} = \sqrt{2\left(1 - \sumno_x \Psi(x) \Phi(x)\right)} \geq \frac12\left\|P_\Psi - P_\Phi \right\|_1\, .
\ee
Hence, for any fixed $V\in \PP(\XX)$ and $\xi\in [0,1]$, we obtain the estimate
\bb
\int_{P:\ \frac12 \|P-V\|_1\leq \xi} \dd P \geq \int_{\Psi:\ \|\Psi - \Phi_V\|_2\leq \xi} \dd \Psi \eqqcolon A(\xi)\, ,
\label{constrained_deFinetti_proof_eq3}
\ee
where we defined $\Phi_V(x) \coloneqq \sqrt{V(x)}$ for all symbols $x\in \XX$, and $A(\xi)$ denotes the surface area of the hyperspherical cap $\big\{\Psi:\ \|\Psi - \Phi_V\|_2\leq \xi\big\}$, which, by rotational invariance, does not depend on $V$. The only property of this function we will use is that~\cite[Lemma~2.3]{BALL}
\bb
A(\xi) \geq C_d\, \xi^{d-1}
\label{constrained_deFinetti_proof_eq4}
\ee
for all $\xi\in [0,2]$, where $C_d>0$ is a universal constant that depends only on $d$. (For example, \cite[Lemma~2.3]{BALL}~shows that one can set $C_d = 2^{-d}$.)

We now claim that~\eqref{constrained_deFinetti} holds for the above choice of $\dd P$. Our starting point is Lemma~\ref{meta_lemma_perm_symm}, which tells us that for any $\Delta>0$ we can find some positive integer $N = N(\Delta,c,|\XX|)$ such that, for all $n\geq N$, $Q_n\in \FF_n$, $V\in \TT_n$, and $P\in \PP(\XX)$ obeying $\supp(P) \subseteq \supp(R)$ and $\frac12\|P - V\|_1\leq \xi\in (0,1/3)$, we have
\bb
Q_n(T_{n,V}) \leq \exp\left[ - D(P^{\otimes n} \| \FF_n) + n \big(\tfrac{\Delta}{5} + \phi(\xi)\big) \right] ,
\label{constrained_deFinetti_proof_eq5}
\ee
where $\phi$ is a continuous function that depends only on $c$ and $|\XX|$ and vanishes at $0$. We now fix some $V\in \TT_n$ and $\xi\in (0,1/3)$, and integrate the above inequality over the set of $P$'s that meet the assumptions. This yields
\begin{align}
&C_d\, \xi^{d-1} Q_n(T_{n,V}) \nonumber \\
&\qquad \leqt{(i)} A(\xi)\, Q_n(T_{n,V}) \nonumber \\
&\qquad \leqt{(ii)} \int_{P:\ \frac12\|P-V\|_1\leq \xi} \dd P\ Q_n(T_{n,V}) \nonumber \\
&\qquad \eqt{(iii)} \int_{P:\ \supp(P) \subseteq \supp(R),\ \frac12\|P-V\|_1\leq \xi} \dd P\ Q_n(T_{n,V}) \nonumber \\
&\qquad \leqt{(iv)} \int_{P:\ \supp(P) \subseteq \supp(R),\ \frac12\|P-V\|_1\leq \xi} \dd P\ \exp\left[ - D(P^{\otimes n} \| \FF_n) + n \big(\tfrac{\Delta}{5} + \phi(\xi)\big) \right] \label{constrained_deFinetti_proof_eq6} \\
&\qquad \eqt{(v)} \int_{P:\ \frac12\|P-V\|_1\leq \xi} \dd P\ \exp\left[ - D(P^{\otimes n} \| \FF_n) + n \big(\tfrac{\Delta}{5} + \phi(\xi)\big) \right] \nonumber \\
&\qquad \leqt{(vi)} (n+1)^{|\XX|}\! \int_{P:\ \frac12\|P-V\|_1\leq \xi} \!\!\dd P\ \exp\left[ - D(P^{\otimes n} \| \FF_n) + n \big(\tfrac{\Delta}{5} + \phi(\xi) + D(V\|P)\big) \right] P^{\otimes n}(T_{n,V}) \nonumber \\
&\qquad \leqt{(vii)} (n+1)^{|\XX|}\! \int_{P:\ \frac12\|P-V\|_1\leq \xi} \!\!\dd P\ \exp\left[ - D(P^{\otimes n} \| \FF_n) + n \big(\tfrac{\Delta}{5} + \phi(\xi) + \lambda_n (\xi) \big) \right] P^{\otimes n}(T_{n,V}) \nonumber \\
&\qquad \leq (n+1)^{|\XX|}\! \int \dd P\ \exp\left[ - D(P^{\otimes n} \| \FF_n) + n \big(\tfrac{\Delta}{5} + \phi(\xi) + \lambda_n (\xi) \big) \right] P^{\otimes n}(T_{n,V})\, . \nonumber
\end{align}
The justification of the above steps is as follows. The inequality~(i) is an application of~\eqref{constrained_deFinetti_proof_eq4}, in~(ii) we employed~\eqref{constrained_deFinetti_proof_eq3}, while~(iii) and~(v) follow from the observation that the measure $\dd P$ is concentrated by construction on the $P$'s such that $\supp(P) \subseteq \supp(R)$. In~(iv) we used~\eqref{constrained_deFinetti_proof_eq5}, noting that the right-hand side is a continuous and therefore measurable function of $P$, due to Lemma~\ref{ac_relent_resource_lemma}. In~(vi) we applied Sanov's theorem~\cite[Exercise~2.12(a), p.~29]{CSISZAR-KOERNER}, and finally in~(vii) we defined the ancillary function
\bb
\lambda_n (\xi) \coloneqq \max_{V\in \TT_n} \sup_{P:\ \frac12\|P-V\|_1\leq \xi} D(V\|P)\, .
\label{constrained_deFinetti_proof_eq6}
\ee

Since $Q_n$ is permutationally symmetric and the same is true of any convex combination of i.i.d.\ distributions, the inequality in~\eqref{constrained_deFinetti_proof_eq5} entails that
\bb
C_d\, \xi^{d-1} Q_n \leq (n+1)^{|\XX|} \int \dd P\ \exp\left[ - D(P^{\otimes n} \| \FF_n) + n \big(\tfrac{\Delta}{5} + \phi(\xi) + \lambda_n (\xi) \big) \right] P^{\otimes n}
\ee
holds as an entry-wise inequality. Massaging this, we obtain
\bb
Q_n \leq \exp\left[ n \left( \tfrac{|\XX|}{n} \log (n\!+\!1) + \tfrac1n \log\tfrac{1}{C_d \xi^{d-1}} + \tfrac{\Delta}{5} + \phi(\xi) + \lambda_n (\xi) \right) \right] \int \dd P\ \exp\left[ - D(P^{\otimes n} \| \FF_n) \right] P^{\otimes n} .
\ee
To proceed further, we fix $\xi = \min\big\{\frac{1}{2n}, \frac13\big\}$, which gives us (as long as $n\geq 2$)
\bb
Q_n \leq \exp\left[ n \left( \tfrac{|\XX|}{n} \log (n\!+\!1) + \tfrac1n \log\tfrac{(2n)^{d-1}}{C_d} + \tfrac{\Delta}{5} + \phi\big(\tfrac{1}{2n}\big) + \lambda_n\big(\tfrac{1}{2n}\big) \right) \right] \int \dd P\ \exp\left[ - D(P^{\otimes n} \| \FF_n) \right] P^{\otimes n} .
\ee
The only thing that remains to be shown to complete the proof of~\eqref{constrained_deFinetti} is that we can make the term inside the round brackets in the first exponential smaller than $\Delta$ for a sufficiently large $n$. We can definitely make sure that
\bb
\tfrac{|\XX|}{n} \log (n\!+\!1) \leq \frac{\Delta}{5}\, ,\qquad \frac1n \log\frac{(2n)^{d-1}}{C_d} \leq \frac{\Delta}{5}\, ,\qquad \phi\big(\tfrac{1}{2n}\big) \leq \frac{\Delta}{5}\, ,
\label{constrained_deFinetti_proof_eq10}
\ee
as long as we choose $n$ to be sufficiently large, because all the functions on the left-hand sides vanish as $n\to\infty$. The problem is whether we can also guarantee that
\bb
\lambda_n\big(\tfrac{1}{2n}\big) \leqt{?} \frac{\Delta}{5}
\label{constrained_deFinetti_proof_eq11}
\ee
for all sufficiently large $n$, which reduces to the problem of establishing whether
\bb
\lim_{n\to\infty} \lambda_n\big(\tfrac{1}{2n}\big) \eqt{?} 0\, .
\label{constrained_deFinetti_proof_eq12}
\ee

To prove~\eqref{constrained_deFinetti_proof_eq12}, start by observing that every $P\in \PP(\XX)$ such that $\frac12 \|P-V\|_1\leq \tfrac{1}{2n}$ satisfies that
\bb
V(x) \leq 2 P(x)\qquad \forall\ x\in \XX\, .
\label{constrained_deFinetti_proof_eq13}
\ee
Indeed, the inequality in~\eqref{constrained_deFinetti_proof_eq13} is obvious when $x\notin\supp(V)$. When, on the contrary, $x\in \supp(V)$, it must be that $V(x) \geq \tfrac1n$, because $V$ is an $n$-type (see~\eqref{all_types}); hence,
\bb
P(x) \geq V(x) - \max_{x'} |V(x') - P(x')| \geq V(x) - \frac12 \|P-V\|_1 \geq \frac1n - \frac{1}{2n} = \frac{1}{2n}\, ,
\ee
so that
\bb
V(x) \leq P(x) + \max_{x'} |V(x') - P(x')| \leq P(x) + \frac12\|P-V\|_1 \leq P(x) + \frac{1}{2n} \leq 2P(x)\, ,
\ee
as claimed. Another way to phrase the now proven~\eqref{constrained_deFinetti_proof_eq13} is by stating that $D_{\max}(V\|P)\leq \log 2$, where the max-relative entropy is defined in~\eqref{D_max}. We can now use~\cite[Eq.~(13)]{continuity-via-integral} to estimate
\bb
D(V\|P) \leq D(P\|P) + \tfrac{1}{2n}\, D_{\max}(V\|P) + h_2\big(\tfrac{1}{2n}\big) \leq \tfrac{1}{2n} \log 2 + h_2\big(\tfrac{1}{2n}\big)
\ee
for any $P$ such that $\frac12 \|P-V\|_1\leq \frac{1}{2n}$, which plugged into~\eqref{constrained_deFinetti_proof_eq6} gives
\bb
\lambda_n\big(\tfrac{1}{2n}\big) \leq \tfrac{1}{2n}\log 2 + h_2\big(\tfrac{1}{2n}\big)\, ,
\ee
which immediately implies~\eqref{constrained_deFinetti_proof_eq12}, and hence also~\eqref{constrained_deFinetti_proof_eq11}. Together with~\eqref{constrained_deFinetti_proof_eq10}, this completes the proof of~\eqref{constrained_deFinetti}. To deduce~\eqref{constrained_deFinetti_special}, simply observe that under Axiom~\ref{new_ax_tensor_products_n_and_m} the sequence $n\mapsto \rel{D}{P^{\otimes n}}{\FF_n}$ is sub-additive, implying, by Fekete's lemma~\cite{Fekete1923}, that 
\bb
D^\infty(P\|\FF) = \inf_{k\in \N^+} \frac1k\, \rel{D}{P^{\otimes k}}{\FF_k} \leq \frac1n\, \rel{D}{P^{\otimes n}}{\FF_n}\qquad \forall\ n\in \N^+. 
\ee
Plugging this estimate into~\eqref{constrained_deFinetti} yields~\eqref{constrained_deFinetti_special} and concludes the proof.
\end{proof}

\begin{rem}
It is possible to simplify the above proof considerably if one is content with a slightly weaker result in which the measure $\dd P$ is allowed to depend on $n$. In this case, one can simply take $\dd P$ as the uniform measure over types. The details are left to the interested reader.
\end{rem}

In light of the above findings, we find the following conjecture quite natural. Note that the second inequality is trivially true due to the sub-additivity of the relative entropy of entanglement.

\begin{cj}
Let $AB$ be a finite-dimensional bipartite quantum system with Hilbert space $\HH_{AB}$. There exists a measure $\dd \omega$ on $\D\big(\HH_{AB}\big)$ with the following property: for any $\Delta>0$, we can find $N = N\big(\Delta, \dim \HH_{AB}\big) \in \N^+$ such that, for all $n\geq N$, all permutationally symmetric separable states $\sigma_n = \sigma_{A^nB^n}\in \SEP_{A^n:B^n} = \SEP_n$ satisfy that
\bb
\sigma_n \leq \int \!\!\dd \omega\ \exp\left[ - D(\omega^{\otimes n} \| \SEP_n) + n \Delta \right]\, \omega^{\otimes n} \leq \int \!\!\dd \omega\ \exp\left[ - n\left( D^\infty(\omega \| \SEP) - \Delta\right) \right]\, \omega^{\otimes n} .
\ee
Here, $\SEP_{A^n:B^n}$ denotes the set of states that are separable (i.e.\ un-entangled)~\cite{Werner} across the cut $A^n:B^n$, where on one side we have $n$ copies of the system $A$, and on the other $n$ copies of the system $B$.
\end{cj}


\bigskip
\noindent \emph{Acknowledgements.} I thank Ronald de Wolf and Marco Tomamichel for independently suggesting the idea of a symbol-by-symbol blurring procedure when I discussed with them some of the ideas in~\cite{GQSL}. I am also grateful to Mario Berta and Bartosz Regula for many discussions on hypothesis testing, and to Mil\'an Mosonyi for helping me navigate the prior literature on composite Chernoff--Stein's lemmas. Funded by the European Union under the ERC StG ETQO, Grant Agreement no.~101165230.

\bibliography{../../biblio}

\appendix

\section{Proof of the asymptotic continuity of the relative entropy distance (Lemma~\ref{ac_relent_resource_lemma})} \label{app_ac_relent_resource}

In what follows, we will present a self-contained proof of Lemma~\ref{ac_relent_resource_lemma}. The argument is essentially derived from that in~\cite[Proposition~13]{continuity-via-integral}, with minor modifications.

\begin{proof}[Proof of Lemma~\ref{ac_relent_resource_lemma}]
For generic $Q_n\in \FF_n$ and $\delta\in [0,1]$, to be fixed later, we can write
\begin{align}
D(P_n\|\FF_n) &\leqt{(i)} \rel{D}{P_n}{\mm^{\otimes n}(Q_n)} \nonumber \\
&\leqt{(ii)} \rel{D}{P'_n}{\mm^{\otimes n}(Q_n)} + \e\, \rel{D_{\max}}{P_n}{\mm^{\otimes n}(Q_n)} + h_2(\e) \label{ac_relent_resource_proof_eq1} \\
&\leqt{(iii)} D(P'_n\|Q_n) + n\log\tfrac{1}{1-\delta} + n\e \log\tfrac{1}{\delta c} + h_2(\e)\, . \nonumber
\end{align}
Here, (i)~holds because $\mm^{\otimes n}(Q_n)\in \FF_n$ due to Axiom~\ref{new_ax_depolarising}; step~(ii), instead, is an application of~\cite[Eq.~(13)]{continuity-via-integral}. Finally, the critical inequality~(iii) can be justified as follows: on the one hand, by construction $\mm^{\otimes n}(Q_n) \geq (1-\delta)^n Q_n$; this implies, via the monotonicity of the logarithm, that 
\bb
\rel{D}{P'_n}{\mm^{\otimes n}(Q_n)} \leq D(P'_n\|Q_n) + n\log\tfrac{1}{1-\delta}\, ; 
\ee
on the other, the complementary inequality $\mm^{\otimes n}(Q_n) \geq \delta^n R^{\otimes n} \geq (\delta c)^n P_n$, which holds because $\supp(P_n) \subseteq \supp(R)^{n}$, entails that
\bb
\rel{D_{\max}}{P_n}{\mm^{\otimes n}(Q_n)} \leq n \log\tfrac{1}{\delta c}\, .
\ee
We can now minimise the rightmost side of~\eqref{ac_relent_resource_proof_eq1} with respect to $\delta\in [0,1]$. Using the easily verified formula
\bb
\inf_{\delta\in (0,1)} \left\{ \log\tfrac{1}{1-\delta} + \e \log \tfrac1\delta \right\} = g(\e)\, ,
\ee
we obtain immediately that
\bb
D(P_n\|\FF_n) &\leq D(P'_n\| Q_n) + n\e \log\tfrac{1}{c} + ng(\e) + h_2(\e) \, .
\ee
A further minimisation over $Q_n\in \FF_n$ yields~\eqref{ac_relent_resource}.
\end{proof}

\section{Elementary properties of the auxiliary function}
\label{app_proof_variational_F_lemma}

Here we state and prove some useful properties of the auxiliary function $F_c$ defined by~\eqref{F_c}.

\begin{lemma} \label{variational_F_lemma}
For all $c,c_1,c_2\in (0,1]$ and all $x\geq 0$, the function $F_c$ defined by~\eqref{F_c} satisfies the following properties:
\begin{enumerate}[(a)]
\item $F_c(x) = \sup_{y \in [0,x]} \left\{ y \log\frac1c + h_2(y) \right\}$;

\item $F_{c_1}(x) + F_{c_2}(x) \leq 2\, F_{\min\{c_1,c_2\}}(x)$; and

\item $F_c(x) = \inf_{\delta\in \left(0, \scaleobj{1}{\frac{1}{c+1}}\right]} \left\{ x \log \frac{1-\delta}{c\delta} + \log\frac{1}{1-\delta} \right\}$.
\end{enumerate}
\end{lemma}

\begin{proof}
We start from~(a). The function $y\mapsto y \log\frac1c + h_2(y)$ has derivative
\bb
\frac{1}{\log e}\, \partial_y \left( y \log\frac1c + h_2(y) \right) = \ln\frac1c + \ln \left(\frac1y - 1\right) .
\ee
This is positive for $y\in \big(0, \frac{1}{c+1}\big)$, and negative for $y\in \big(\frac{1}{c+1},1\big)$. Hence, the maximum is achieved at $y=x$ if $x\leq \frac{1}{c+1}$, and at $y=\frac{1}{c+1}$ otherwise. In this latter case, the value of the maximum is precisely $\log\left(1+\frac1c\right)$. This proves~(a).

We now move on to~(b). It suffices to use~(a) to write
\bb
F_{c_1}(x) + F_{c_2}(x) &= \sup_{y\in [0,x]} \left\{ y \log\frac{1}{c_1} + h_2(y) \right\} + \sup_{z\in [0,x]} \left\{ z \log\frac{1}{c_2} + h_2(z) \right\} \\
&= \sup_{y,z\in [0,x]} \left\{ y \log\frac{1}{c_1} + z \log\frac{1}{c_2} + h_2(y) + h_2(z) \right\} \\
&\leq \sup_{y,z\in [0,x]} \left\{ (y+z) \log\frac{1}{\min\{c_1,c_2\}} + h_2(y) + h_2(z) \right\} \\
&= 2 \sup_{y,z\in [0,x]} \left\{ \frac{y+z}{2} \log\frac{1}{\min\{c_1,c_2\}} + \frac12 \left(h_2(y) + h_2(z)\right) \right\} \\
&\leq 2 \sup_{y,z\in [0,x]} \left\{ \frac{y+z}{2} \log\frac{1}{\min\{c_1,c_2\}} + h_2\left(\frac{y+z}{2}\right) \right\} \\
&= 2 \sup_{w\in [0,x]} \left\{ w \log\frac{1}{\min\{c_1,c_2\}} + h_2(w) \right\} \\
&= 2\, F_{\min\{c_1,c_2\}}(x)\, ,
\ee
where the second inequality is the concavity of the binary entropy function, and on the second-to-last line we introduced the parameter $w\coloneqq (y+z)/2$.

As for~(c), note that the derivative of the objective function is given by
\bb
\frac{1}{\log e}\, \partial_\delta \left( x \log \frac{1-\delta}{c\delta} + \log \frac{1}{1-\delta} \right) &= \frac{1-x}{1-\delta} - \frac{x}{\delta}\, .
\ee
If $x\geq 1$, this is negative for all $\delta \in (0,1)$. If $\frac{1}{c+1} < x < 1$, it is negative for all $\delta \in (0,x)$, an in particular for all $\delta$ in the range. In both cases, i.e.\ whenever $x \geq \frac{1}{c+1}$, the minimum of the objective function is achieved for $\delta=\frac{1}{c+1}$, giving $\log\big(1+\frac1c\big) = F_c(x)$ as the result of the optimisation in this case. If $x\leq \frac{1}{c+1}$, instead, the derivative is non-positive for $0<\delta \leq x$ and non-negative for $\delta\geq x$, implying that the minimum of the objective function is achieved for $\delta=x$, yielding
\bb
x\log \frac{1-x}{cx} + \log \frac{1}{1-x} = x \log \frac{1}{c} + h_2(x) = F_c(x)
\ee
and thus completing the proof.
\end{proof}

\end{document}